\documentclass[12pt]{elsarticle}
\usepackage{amsmath,amsfonts,amssymb}
\usepackage{url}
\usepackage{pstricks,graphicx}

\usepackage[english]{babel}
\usepackage[latin1]{inputenc} 
\usepackage[T1]{fontenc}
\usepackage{epsfig}
\usepackage{graphicx}
\usepackage{amssymb}
\usepackage{amsfonts}

\usepackage{amsmath}
\usepackage{verbatim}
\usepackage{url}
 \usepackage{algorithmic}
\usepackage{algorithm}
\usepackage[all]{xypic}
\usepackage{epic,eepic}
\usepackage{tikz}
\usetikzlibrary{automata}
\usetikzlibrary{positioning}
\usetikzlibrary{arrows}
\newcommand{\nat}{\mathbb{N}}
\newcommand{\rel}{\mathbb{Z}}
\newcommand{\interval}[2]{[#1,#2]}
\newcommand{\Nat}{\mathbb{N}}
\newcommand{\Zed}{\mathbb{Z}}

\newcommand{\tuple}[1]{\langle #1 \rangle}
\newcommand{\pair}[2]{\langle #1, #2 \rangle}
\newcommand{\triple}[3]{\langle #1, #2,#3 \rangle}
\newcommand{\set}[1]{\{ #1 \}}

\newcommand{\defeq}{\overset{\textsf{def}}{\Leftrightarrow}}

\newcommand{\defstyle}[1]{{\em #1}}

\newcommand{\until}{\mathtt{U}}

\newcommand{\mynext}{\mathtt{X}}
\newcommand{\always}{\mathtt{G}}
\newcommand{\sometimes}{\mathtt{F}}

\newcommand{\eventually}{\mathtt{F}}

\newcommand{\myprevious}{\mathtt{X}^{-1}}
\newcommand{\since}{\mathtt{S}}

\newcommand{\ptime}{\textsc{PTime}}

\newcommand{\np}{\textsc{NP}}
\newcommand{\nc}{\textsc{NC}}
\newcommand{\pspace}{\textsc{PSpace}}
\newcommand{\exptime}{\textsc{ExpTime}}

\newcommand{\abigformula}{\Phi}
\newcommand{\aformula}{\phi} 
\newcommand{\aformulabis}{\psi} 

\newcommand{\aalphabet}{\Sigma}
\newcommand{\aletter}{a}

\newcommand{\sizeof}[1]{{\rm size}(#1)}
\newcommand{\vect}[1]{\mathtt{\textbf{#1}}}

\newcommand{\counters}{{\tt C}}
 \newcommand{\acounter}{\avariable}
\newcommand{\afactor}{a}
\newcommand{\aconstant}{b}
\newcommand{\guards}{{\tt G}}
\newcommand{\aguard}{{\tt g}}
\newcommand{\anupdate}{\vect{u}}

\newcommand{\avect}{\vect{v}}

\newcommand{\aterm}{{\tt t}}

\newcommand{\asys}{S}
\newcommand{\transsysof}[1]{TS(#1)}
\newcommand{\states}{Q}
\newcommand{\astate}{q}

\newcommand{\edges}{\Delta}
\newcommand{\anedge}{\delta}

\newcommand{\source}[1]{\mathit{source}(#1)}
\newcommand{\target}[1]{\mathit{target}(#1)}
\newcommand{\guard}[1]{\mathit{guard}(#1)}
\newcommand{\update}[1]{\mathit{update}(#1)}

\newcommand{\footprint}[1]{{\sf ft}(#1)}

\newcommand{\trans}{\rightarrow}
\newcommand{\labtrans}[1]{\xrightarrow{#1}}
\newcommand{\confs}{C}
\newcommand{\aconf}{c}
\newcommand{\aword}{w}

\newcommand{\arun}{\rho}
\newcommand{\amodel}{\sigma}
\newcommand{\wordof}[1]{\mathit{lab}(#1)}
\newcommand{\aregexp}{E}

\newcommand{\PLTL}{{\rm PLTL}}
\newcommand{\mc}[2]{{\rm MC}(#1,#2)}

\newcommand{\logicfrag}{{\rm L}}
\newcommand{\csfrag}{\mathcal{C}}
\newcommand{\flatcs}{\mathcal{CFS}}
\newcommand{\flatks}{\mathcal{KFS}}
\newcommand{\kps}{\mathcal{KPS}}
\newcommand{\cps}{\mathcal{CPS}}

\newcommand{\aseg}{p}
\newcommand{\first}[1]{\mathit{first}(#1)}
\newcommand{\last}[1]{\mathit{last}(#1)}
\newcommand{\effect}[1]{\mathit{effect}(#1)}
\newcommand{\preeffect}[1]{\mathit{effect}^{<}(#1)}
\newcommand{\aloop}{l}
\newcommand{\aschema}{P}
\newcommand{\languageof}[1]{\mathcal{L}(#1)}
\newcommand{\loopsof}[2]{iter_{#1}(#2)}

\newcommand{\awordbis}{u}
\newcommand{\awordter}{v}
\newcommand{\tempdepth}[1]{\mathit{td}(#1)}

\newcommand{\avariable}{{\sf x}}
\newcommand{\avariablebis}{{\sf y}}

\newcommand{\card}[1]{{\rm card}(#1)}
\newcommand{\egdef}{\stackrel{\mbox{\begin{tiny}def\end{tiny}}}{=}} 
\newcommand{\equivdef}{\stackrel{\mbox{\begin{tiny}def\end{tiny}}}{\equivaut}} 
\newcommand{\equivaut}{\;\Leftrightarrow\;}

\newcommand{\td}[1]{td(#1)}

\newcommand{\aset}{X}
\newcommand{\asetbis}{Y}
\newcommand{\asetter}{Z}

\newcommand{\aconstraintsystem}{\mathcal{E}}

\newcommand{\equivrel}[1]{\approx_{#1}^{ }}
\newcommand{\alabelling}{\mathbf{l}}
\newcommand{\atermmap}{\vec{m}}
\newcommand{\powerset}[1]{2^{#1}}
\newcommand{\atransition}{\anedge}
\newcommand{\step}[1]{\xrightarrow{\!\!#1\!\!}}
\newcommand{\amap}{f}
\newcommand{\length}[1]{{\rm len}(#1)} 
\newcommand{\nbloops}[1]{{\rm nbloops}(#1)}
\newcommand{\askeleton}{{\tt sk}}
\newcommand{\varprop}{{\rm AT}}
\newcommand{\avarprop}{p}

\newcommand{\amatrix}{\mathcal{M}}

\newcommand{\mapstate}{\mathit{\pi_{\astate}}}
\newcommand{\mapedge}{\mathit{\pi_{\anedge}}}
\newcommand{\amapbis}{h}

\newcommand{\size}[1]{{\rm size}(#1)}

\newcommand{\aresource}{\mathsf{R}}
\newcommand{\afootprint}{\mathsf{ft}}
\newcommand{\proj}{\mathsf{proj}}
\newcommand{\symbmodels}{\models_{{\small \sf symb}}}

\newcommand{\commzone}[1]{\textbf{[#1]}}

\newcommand{\labels}{L}
\newcommand{\nodistrans}{\tilde\Delta}
\newcommand{\cut}[1]{}

\renewcommand{\vec}[1]{\mathbf{#1}}

 \newif \ifshort \shortfalse 
 \newif \iflong \longtrue
\newif \ifready \readytrue

 \newtheorem{proposition}{Proposition}[section]

 \newtheorem{definition}[proposition]{Definition}
 
 \newtheorem{theorem}[proposition]{Theorem}
\newtheorem{lemma}[proposition]{Lemma}
 \newtheorem{corollary}[proposition]{Corollary}
 
\newenvironment{proof}{\vspace{1ex}\noindent{\bf Proof}\hspace{0.5em}}
	{\hfill\qed\vspace{1ex}}

\begin{document}

\begin{frontmatter}

\title{%
Taming Past LTL and Flat Counter Systems\footnote{
Supported by ANR project REACHARD 
ANR-11-BS02-001.  This is the completed version of~\cite{DemriDharSangnier12}.}
}
\author[LSV]{St\'ephane Demri}
\ead{demri@lsv.ens-cachan.fr}
\author[LIAFA]{Amit Kumar Dhar}
\ead{dhar@liafa.univ-paris-diderot.fr}
\author[LIAFA]{Arnaud Sangnier}
\ead{sangnier@liafa.univ-paris-diderot.fr}

\address[LSV]{LSV, CNRS, ENS Cachan, INRIA, France}
\address[LIAFA]{LIAFA, Univ. Paris Diderot, Sorbonne Paris Cité, CNRS, France}


\begin{abstract}
Reachability  and LTL model-checking problems for flat coun\-ter systems are known to be decidable
but whereas the reachability problem can be shown in \np, the best known complexity 
upper bound for the latter problem
is made of a tower of several exponentials. Herein, we show that the problem is 
only \np-complete even if LTL admits
past-time operators and arithmetical constraints on counters. 
\iflong
For instance, adding past-time operators
to LTL immediately leads to complications;  an \np \ upper bound cannot 
be deduced by translating
formulae into B\"uchi automata. 
\fi 
Actually, the \np \ upper bound is shown by
adequately combining a new stuttering theorem for Past LTL  and the property of small 
integer solutions  for quantifier-free Presburger formulae. Other complexity results
are  proved, for instance for restricted classes of flat counter 
\iflong
systems such as path schemas.
\fi 
\iflong
Our \np \ upper bound extends known and recent results on 
model-checking weak Kripke structures with LTL formulae 
as well as  reachability problems for flat counter systems.
\fi 
\end{abstract}

\end{frontmatter}

\section{Introduction}
\label{section-introduction}
\iflong \paragraph{Flat counter systems}
\fi 
Counter systems are finite-state automata equip\-ped with program variables (counters) interpreted
over non-negative integers. They are  used in many places like,
broadcast protocols~\cite{Esparza&Finkel&Mayr99} and programs with pointers~\cite{Finkel&Lozes&Sangnier09} 
to quote a few 
examples.   But, alongwith their large scope of usability, many problems on general counter 
systems are known to be 
\iflong
undecidable~\cite{minsky-computation-67}.
\fi 
Indeed, this computational model can simulate 
Turing machines. 
\iflong
This is not the end of the story since decidability 
\fi 
of  reachability problems or model-checking problems based on temporal 
logics, can be regained by considering subclasses of  counter 
\iflong
systems (this includes 
restrictions on the instructions, on the control graphs or on more semantical properties). 
\fi 
An important and natural class of counter
systems, in which various practical cases of infinite-state
systems (e.g. broadcast protocols~\cite{Finkel&Leroux02b}) can be
modelled, are those with a \emph{flat} control graph, i.e, those
where no control state occurs in more than one simple cycle,
\iflong
see e.g.~\cite{Boigelot98,Comon&Jurski98,Finkel&Leroux02b,Leroux&Sutre05,Bozga&Iosif&Lakhnech09}.
\fi 
Decidability results on verifying safety and reachability properties on
flat counter systems have been obtained
in~\cite{Comon&Jurski98,Finkel&Leroux02b,Bozga&Iosif&Konecny10}.
However, so far, such properties have
been rarely considered in the framework of any formal specification
language (see an exception in~\cite{Comon&Cortier00}). 
In~\cite{demri-model-10}, 
 a class of  Presburger counter systems is identified for which the
 local model checking problem for 
Presburger-CTL$^{\star}$ is shown decidable. These are
Presburger counter systems defined over flat control graphs
with arcs labelled by
adequate Presburger formulae (representing constraints on counters).
Even though flatness is clearly a substantial restriction,
it is shown in~\cite{Leroux&Sutre05} that many classes of counter systems with
computable
Presburger-definable reachability sets are \emph{flattable}, i.e. there exists
a flat unfolding of the counter system with identical reachability sets.
Hence, the possibility of flattening a counter system is strongly
related to semilinearity of its reachability set. 
Moreover, in~\cite{Comon&Cortier00} model-checking relational counter systems
over LTL formulae is shown decidable when restricted to flat formulae (their translation
into automata leads to flat 
\iflong
structures). 
\fi
\iflong 
\paragraph{Towards the complexity of temporal model-checking flat counter systems}
\fi 
In~\cite{demri-model-10}, it is shown that CTL$^{\star}$ model-checking
over the class of so-called \defstyle{admissible} counter systems is decidable
by reduction into the satisfiability problem for Presburger arithmetic,
the decidable first-order theory of natural numbers with addition. 
Obviously  CTL$^{\star}$ properties are more expressive than reachability
properties but this has a cost. 
However, for the class of counter systems considered in this paper, this 
provides a very rough complexity upper bound in 4\exptime. 
Herein,  our goal is to revisit standard  decidability results
for subclasses of counter systems obtained  by translation into 
Presburger arithmetic in order to 
obtain optimal complexity upper bounds.
Indeed, effectively composing the translation of
a verification problem
into Presburger arithmetic (PrA) and then using a solver for (PrA) is not necessarily
 optimal computationally.
\iflong 
\paragraph{Our contributions}  
\fi 
In this paper, we establish several computational complexity characterizations
of model-checking problems restricted to flat counter systems in the presence of a rich  LTL-like
specification language with arithmetical constraints and past-time operators. 
Not only we provide an optimal complexity but also, we believe that our proof technique could be reused
for further extensions. 
Indeed, we combine three proof techniques: the general stuttering theorem~\cite{Kucera&Strejcek05}, 
the property of small integer solutions of equation systems~\cite{Borosh&Treybig76}
 (this latter technique is used 
\iflong
since~\cite{Rackoff78,Gurari&Ibarra81}) and the elimination of disjunctions in 
guards (see Section~\ref{section-disjunction}). 
\fi 
Let us be a bit more precise.

\noindent
We extend the \iflong general \fi  stuttering principle established in~\cite{Kucera&Strejcek05} for LTL (without past-time operators)
to Past LTL. 
\iflong
However, since this principle will be applied to path 
schemas, a fundamental structure
in flat counter systems, we do not aim at being optimal as soon as it will be helpful to 
establish the \np \ upper bounds.
A path schema is simply a finite alternation of path segments and simple loops (no repetition of edges)
and the principle states that satisfaction of an LTL formula requires only to take loops a number of times
that is linear in the temporal depth of the formula. 
This principle has been already used to establish that LTL model-checking over \emph{weak} Kripke structures 
is in \np~\cite{Kuhtz&Finkbeiner11} (weakness corresponds to flatness). 
\fi 
It is worth noting that another way to show a similar result would be to eliminate
past-time operators  thanks to Gabbay's Separation Theorem~\cite{Gabbay87} (preserving initial equivalence)
but the temporal depth of formulae might increase at least exponentially, 
which is a crucial parameter in our complexity
analysis. 
\noindent
We show that the model-checking problem restricted to flat 
counter systems in the presence
of LTL with past-time operators 
is in \np \ (Theorem~\ref{theorem-main}) by combining the above-mentioned
proof 
\iflong 
techniques (we call this problem $\mc{\PLTL[\counters]}{\flatcs}$). 
\fi 
Apart from the use of the general stuttering theorem (Theorem~\ref{theorem-stuttering}),
we  take advantage of the other properties stated for instance
in Lemma~\ref{lemma-constraint-system} (characterization of
runs by quantifier-free Presburger formulae)
and
Theorem~\ref{theorem-main-disjunction} (elimination of disjunctions in guards preserving
flatness). 
\iflong
Note that the loops  in runs are visited  a number of times that can be exponential in the
 worst case, but this does not
prevent us from  establishing the \np \ upper bound.
\fi 
\iflong 
We  also take advantage of the fact that 
model-checking ultimately periodic models with Past LTL is in 
\ptime~\cite{Laroussinie&Markey&Schnoebelen02} but
our main decision procedure
is not automata-based, unlike the approach from~\cite{Vardi97}. 
\fi 
In the paper, complexity results for fragments/subproblems are also considered.
For instance, we get a  sharp lower bound since we establish that 
the model-checking problem on path schemas
 with only 2 loops is already
\np-hard (see Lemma~\ref{lemma-constant-loops2}).
\iflong A summary table of results can be found in Section~\ref{section-conclusion}.
\fi 

\ifshort
\else
\fi

\section{Flat Counter Systems and its LTL Dialect}
\label{section-definitions}

We write $\nat$ [resp. $\rel$] to denote  the set of natural 
numbers [resp. integers]
and $\interval{i}{j}$ to denote  $\set{k \in \Zed: i \leq k \ {\rm and} \ k \leq j}$.
For $\vec{v}\in\Zed^n$, $\vec{v}[i]$ 
denotes the $i^{th}$ element of $\vec{v}$ for every $i\in \interval{1}{n}$. For some $n$-ary tuple $t$,
we write $\pi_j(t)$ to denote the $j^{th}$ element of $t$ ($j \leq n$). 
In the sequel, integers are encoded with a binary representation. 
 For a finite alphabet $\aalphabet$, $\aalphabet^*$ represents the set of finite words over $\aalphabet$, 
$\aalphabet^+$ the set of finite non-empty words over $\aalphabet$ and $\aalphabet^\omega$ the set of 
$\omega$-words over $\aalphabet$. For a finite word $\aword=\aletter_1\ldots \aletter_k$ over $\aalphabet$, 
we write $\length{\aword}$ to denote its \defstyle{length} $k$. 
For $0 \leq i < \length{\aword}$, $\aword(i)$ represents the $(i+1)$-th letter of the word, here $\aletter_{i+1}$.
\subsection{Counter Systems}
\iflong 
Counter constraints are defined below as a subclass of Presburger formulae whose free variables are understood as counters.
Such constraints are used to define guards in counter systems but also to define arithmetical constraints in temporal formulae.
\fi 
Let $\counters = \set{\acounter_1, \acounter_2, \ldots}$ be a countably infinite set
of \defstyle{counters} (variables interpreted over non-negative integers) and
$\varprop = \set{\avarprop_1, \avarprop_2, \ldots}$ be a countable infinite set of propositional
variables (abstract properties about program points). 
We write $\counters_n$ to denote \iflong the restriction of $\counters$ to \fi 
$\set{\acounter_1, \acounter_2, \ldots, 
\acounter_n}$. 
\iflong
\begin{definition}[Guards]
\label{definition-counter-constraints} 
The set $\guards(\counters_n)$ of \defstyle{guards} (arithmetical constraints on counters
in $\counters_n$)
is defined inductively as follows:
$$
\begin{array}{rcl}
\aterm & ::= & \afactor.\acounter~\mid~ \aterm + \aterm\\
\aguard & ::= & \aterm \sim \aconstant~  \mid~\aguard \wedge \aguard~ \mid~ \aguard \vee \aguard 
\end{array}
$$
where $\acounter \in \counters_n$, $\afactor \in \rel$, $\aconstant \in \nat$ and $\sim \in \set{=,\leq,\geq,<,>}$.
\end{definition}
\fi 
\iflong
Note that such guards 
\fi 
are closed under negations (but negation is not a
logical connective) and the truth constants $\top$ and $\perp$ can be easily defined
too. 
Given $\aguard \in \guards(\counters_n)$ and 
a 
vector $\vect{v} \in \nat^n$, we  say that $\avect$ satisfies $\aguard$, written 
$\vect{v} \models \aguard$, if the formula obtained by replacing each  $\acounter_i$ 
by $\vec{v}[i]$ holds. 
%
\begin{definition}[Counter system] For  \iflong a natural 
number \fi $n \geq 1$, a 
\iflong $n$-dim counter 
system  (shortly a counter system)
\fi 
$\asys$ is a tuple 
$\tuple{\states,\counters_n,\edges,\alabelling}$ 
\iflong
where:
\begin{itemize}
\itemsep 0 cm 
\iflong
\item $\states$ is a finite set of \defstyle{control states}. 
\item $\alabelling: \states \rightarrow \powerset{\varprop}$ is a \defstyle{labelling function}.
\else
\item $\states$ is a finite set of \defstyle{control states}; 
$\alabelling: \states \rightarrow \powerset{\varprop}$ is a \defstyle{labelling function}.
\fi 
\item $\edges \subseteq \states \times \guards(\counters_n) \times \Zed^n 
      \times \states$ is a finite set of edges labeled by guards and updates of the counter values
      (\defstyle{transitions}).
\end{itemize}
\fi 
\end{definition}
For  $\anedge=\tuple{\astate,\aguard,\anupdate,\astate'}$ in $\edges$, we  use the following notations:
\begin{itemize}
\itemsep 0 cm
\item $\source{\anedge}=\astate$, 
\item $\target{\anedge}=\astate'$, 
\item $\guard{\anedge}=\aguard$,
\item $\update{\anedge}=\anupdate$.
\end{itemize}
As usual, to a counter system $\asys=\tuple{\states,\counters_n,\edges,\alabelling}$, 
we associate a labeled transition system $\transsysof{\asys}=\tuple{\confs,\trans}$ 
where $\confs=\states \times \nat^n$ is the set of \defstyle{configurations} and 
$\trans \subseteq \confs \times \Delta \times \confs$ is the \defstyle{transition relation} defined by: 
$\triple{\pair{\astate}{\avect}}{\delta}{\pair{\astate'}{\avect'}} \in \trans$ 
(also written  $\tuple{\astate,\avect} \labtrans{\delta} \tuple{\astate',\avect'}$) iff 
\iflong
the conditions below 
are satisfied:
\begin{itemize}
\itemsep 0 cm 
\item $\astate=\source{\anedge}$ and $\astate'=\target{\anedge}$,
\item $\avect \models \guard{\anedge}$ and $\avect'=\avect + \update{\anedge}$.
\end{itemize}
\fi
\iflong Note that in
\fi 
 such a transition system, 
the counter values are non-negative since $\confs=\states \times \nat^n$.
We extend the transition relation $\trans$ to finite words of transitions in $\edges^+$ as 
follows. For each $\aword=\anedge_1 \anedge_2 \ldots \anedge_{\alpha} \in \edges^+$, 
we have $\pair{\astate}{\avect} \labtrans{\aword} \pair{\astate'}{\avect'}$ 
if there are $\aconf_0,\aconf_1,\ldots,\aconf_{\alpha+1} \in \confs$ such that 
$\aconf_i \labtrans{\anedge_i} \aconf_{i+1}$ for all $i \in \interval{0}{\alpha}$, 
$\aconf_0=(\astate,\avect)$ and $\aconf_{\alpha+1}= \pair{\astate'}{\avect'}$. We say that an 
$\omega$-word $\aword \in \edges^\omega$ 
is \defstyle{fireable} in $\asys$ from a configuration $\aconf_0 \in \states 
\times \nat^n$ if for all finite prefixes $\aword'$ of $\aword$ there exists a configuration 
$\aconf \in \states \times \nat^n$ such that $\aconf_0 \labtrans{\aword'} \aconf$. We write 
$\wordof{\aconf_0}$ to denote the set  of $\omega$-words (\defstyle{labels}) which are fireable from $\aconf_0$ in $\asys$.

Given \iflong an initial \fi  configuration $\aconf_0 \in \states \times \nat^n$, a \defstyle{run}
$\arun$ starting from $\aconf_0$ in $\asys$ is an infinite path in the associated 
transition system $\transsysof{\asys}$ denoted as:
\iflong
$$
\arun := \aconf_0 \labtrans{\anedge_0} \cdots \labtrans{\anedge_{\alpha-1}} 
\aconf_{\alpha} \labtrans{\anedge_{\alpha}} \cdots 
$$
\fi 
where $\aconf_i \in \states \times \nat^n$ and 
$\anedge_i \in \edges$ for all $i \in \nat$. 
Let $\wordof{\arun}$ be the $\omega$-word $\anedge_0 \anedge_1 \ldots$ 
associated to the run $\arun$. 
Note that by definition we have $\wordof{\arun} \in \wordof{\aconf_0}$. 
When $\aregexp$ is an $\omega$-regular expression over the finite alphabet 
$\edges$ and $\aconf_0$ is an initial configuration, $\wordof{\aregexp,\aconf_0}$ 
is defined as the set of labels of infinite runs $\arun$ starting at $\aconf_0$ 
such that $\wordof{\arun}$ belongs to the language defined by $\aregexp$. So
$\wordof{\aregexp,\aconf_0} \subseteq \wordof{\aconf_0}$.

We  say that a counter system is \defstyle{flat} if every node in the 
underlying graph belongs to at most one simple cycle (a cycle being simple if no edge is 
repeated twice in it) \cite{Comon&Jurski98}. In a flat counter system, simple cycles can be organized as a DAG
where two simple cycles are in the relation whenever there is path between
a node of the first cycle and a node of the second cycle.
We denote by $\flatcs$ the class of flat counter systems.
 
\iflong
Below, we present the control graph of a flat counter system (guards and updates are omitted).
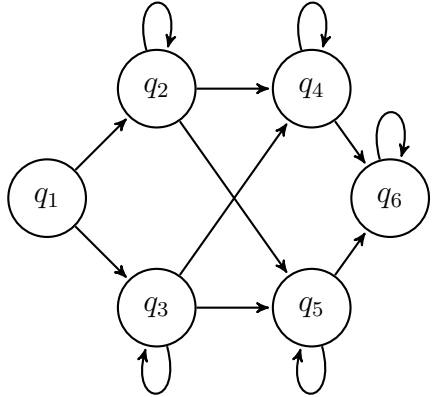
\begin{center}
\scalebox{1}{
\begin{tikzpicture}[->,>=stealth',shorten >=1pt,
node distance=1cm, thick,auto,bend angle=60]

  \tikzstyle{every state}=[fill=white,draw=black,text=black]
  \node[state] (q1) [below]    {$\astate_1$};
  \node[state] (q2) [above right= of q1]    {$\astate_2$};
  \node[state] (q3) [below right= of q1]    {$\astate_3$};
  \node[state] (q4) [right= of q2]    {$\astate_4$};
  \node[state] (q5) [right= of q3]    {$\astate_5$};
  \node[state] (q6) [right=3.5cm  of q1]    {$\astate_6$};

  \path[->] (q1) edge node {} (q2);
  \path[->] (q1) edge node {} (q3);
  \path[->] (q2) edge node {} (q4);
  \path[->] (q2) edge node {} (q5);
  \path[->] (q3) edge node {} (q5);
  \path[->] (q3) edge node {} (q4);
  \path[->] (q4) edge node {} (q6);
  \path[->] (q5) edge node {} (q6);
  \path[->] (q2) edge [loop above] node {} (q2);
  \path[->] (q4) edge [loop above] node {} (q4);
  \path[->] (q6) edge [loop above] node {} (q6);
  \path[->] (q3) edge [loop below] node {} (q3);
  \path[->] (q5) edge [loop below] node {} (q5);
               
\end{tikzpicture}
}
\end{center}
\fi
\iflong
A \defstyle{Kripke structure} $\asys$  is a tuple 
 $\tuple{\states,\edges,\alabelling}$ where $\edges \subseteq \states \times \states$ and $\alabelling$
is labelling. It can be viewed as a degenerate form of counter systems without counters (in the sequel,
we take the freedom to see them as counter systems). All standard notions on counter systems naturally apply
to Kripke structures too (configuration, run, flatness, etc.). In the sequel, we shall also investigate the complexity
of model-checking problems on flat Kripke structures (such a class is denoted by $\flatks$). 
\fi 
\subsection{Linear-Time Temporal Logic with Past and Arithmetical Constraints}
\iflong
Model-checking problem for Past LTL over finite state systems is 
known to be \pspace-complete~\cite{Sistla&Clarke85}. In spite of this nice feature,
a propositional variable $\avarprop$ only represents an abstract property 
about the current configuration of the system.
A more satisfactory solution is to include in the logical language the possibility to express
directly constraints between variables of the program, and doing so refining the standard abstraction made with propositional variables. 
 When the variables are typed, they may be interpreted in some
specific domain like integers, strings and so on; reasoning in such theories can be performed
thanks to satisfiability modulo theories proof techniques, see e.g.,~\cite{Ghilardietal08}
in which SMT solvers are used for model-checking infinite-state systems.
Hence, the basic idea behind the design of the logic  $\PLTL[\counters]$ is to refine
the language of atomic formulae and to allow  comparisons of counter values.
Similar motivations can be found in the introduction of concrete domains in description 
logics, that are logic-based formalisms for knowledge representation~\cite{Baader&Hanschke91,Lutz01b}.
We define below a version of linear-time temporal logic LTL dedicated to counter systems
in which the atomic formulae are linear constraints and the temporal operators
are those of LTL. Note that capacity constraints 
from~\cite{Dixon&Fisher&Konev07} are arithmetical constraints different from those
defined below. 
\fi 
\iflong

The formulae of the logic $\PLTL[\counters]$ are defined as follows:
$$
\begin{array}{lcl}
\aformula & ::= &  \avarprop~ \mid~ \aguard \mid~ \neg \aformula~ \mid~ 
\aformula \wedge \aformula~ \mid~ 
\aformula \vee \aformula~ \mid~ \mynext  \aformula~ \mid~ \aformula \until \aformula~ \mid~ \myprevious  \aformula~ \mid~ \aformula \since \aformula~  
\end{array}
$$
where $\avarprop \in \varprop$ and $\aguard \in \guards(\counters_n)$ for some $n$.
\fi  
We may use the standard abbreviations $\sometimes$, $\always$,  $\always^{-1}$  etc. For instance,
the formula $\always \sometimes (\acounter_1 + 2 \geq \acounter_2)$
states that infinitely often the value of counter 1 plus 2 is greater than the value of counter 2.
The past-time operators $\since$ and $\myprevious$ do not 
add  expressive power to the logic 
\iflong itself~\cite{Gabbay87},
\fi 
but it is known that it helps a lot to express properties succinctly,
see e.g.~\cite{Laroussinie&Schnoebelen00,Laroussinie&Markey&Schnoebelen02}. 
The temporal depth of $\aformula$, written $\td{\aformula}$, is defined
as the maximal number of imbrications of temporal operators in $\aformula$. 
Restriction of $\PLTL[\counters]$ to atomic formulae from $\varprop$ only is written
$\PLTL[\emptyset]$, it corresponds to the standard version of LTL with past-time operators. 
Models of $\PLTL[\counters]$ are essentially abstractions of runs
from counter systems, i.e. $\omega$-sequences $\amodel: \Nat \rightarrow \powerset{\varprop} \times
\Nat^{\counters}$. Given a model $\amodel$ and a position $i \in \Nat$, 
the satisfaction relation $\models$ for $\PLTL[\counters]$ is defined as follows 
(Boolean clauses are omitted):
\iflong
$$
\begin{array}{rcl}
\amodel,i \models~ \avarprop & ~\defeq~ & \avarprop \in \pi_1(\amodel(i))\\
\amodel,i \models~ \aguard & ~\defeq~ & \avect_i \models \aguard \ {\rm where} \ 
\avect_i(\acounter_j) =  \pi_2(\amodel(i))(\acounter_j) \\
\amodel,i \models~ \mynext \aformula & ~\defeq~ & \amodel,i+1 \models~ \aformula \\
\amodel,i \models~ \aformula_1 \until \aformula_2 & ~\defeq~ & 
 \amodel,j \models~ \aformula_2 \mbox{ for some }  i \leq j ~ \\
 & & \mbox{ such that } ~ \amodel,k \models~ \aformula_1 \mbox{ for all } i \leq k <j \\
\amodel,i \models~ \myprevious \aformula & ~\defeq~ & i>0 \mbox{ and } \amodel,i-1 \models~ \aformula \\
\amodel,i \models~ \aformula_1 \since \aformula_2 & ~\defeq~ & 
 \amodel,j \models~ \aformula_2 \mbox{ for some }  0 \leq j \leq i ~ \\
 & & \mbox{ such that } ~ \amodel,k \models~ \aformula_1 \mbox{ for all }  j < k \leq i \\
\end{array}
$$
\fi
 
Given \iflong a counter system \fi  $\tuple{\states, \counters_n, \edges, \alabelling}$
and a run $\arun:= \pair{\astate_0}{\avect_0} \labtrans{\anedge_0} 
\cdots \labtrans{\anedge_{p-1}} \pair{\astate_p}{\avect_p} \labtrans{\anedge_p} \cdots$,
we consider the model $\amodel_{\arun}: \Nat \rightarrow  \powerset{\varprop} \times
\Nat^{\counters}$ such that $\pi_1(\amodel_{\arun}(i)) \egdef \alabelling(\astate_i)$
and $\pi_2(\amodel_{\arun}(i))(\acounter_j) \egdef \avect_i[j]$ for all $j \in \interval{1}{n}$ and  
all $i \in \Nat$. 
Note that  $\pi_2(\amodel_{\arun}(i))(\acounter_j)$ is arbitrary for $j \not \in \interval{1}{n}$.
As expected, we extend the satisfaction relation to runs so that  $\arun, i \models \aformula$
$\equivdef$ $\amodel_{\arun}, i \models \aformula$ whenever $\aformula$ is built from counters
in $\counters_n$. 

\iflong
The verification problem we are interested in is the model-checking problem 
for $\PLTL[\counters]$ over counter systems, written $\mc{\logicfrag}{\csfrag}$, 
where $\logicfrag$ is a fragment of $\PLTL[\counters]$ and $\csfrag$ is a class of counter systems.
$\mc{\logicfrag}{\csfrag}$ is defined as follows:
\begin{description}
\itemsep 0 cm 
\item[Input:] A counter system $\asys \in \csfrag$, a configuration $\aconf_0$ and a formula $\aformula \in \logicfrag$;
\item[Output:] Does there exist a run $\arun$ starting from $\aconf_0$ in $\asys$ such that $\arun,0 \models \aformula$?
\end{description}
If the answer is "yes", we will write $\asys,\aconf_0 \models \aformula$. It is known that for the full class of counter systems, the \iflong model-checking \fi problem is undecidable; this is due to the 
fact that  reachability of a control state is undecidable for counter systems
 manipulating at 
\iflong least two counters~\cite{minsky-computation-67}. 
\else  least two counters.
\fi 
On the other hand, some 
restrictions can lead to 
decidability of this problem. This is the case for flat counter systems, for whom it is 
proved in \cite{demri-model-10} that the model-checking problem of some temporal logic more expressive than 
$\PLTL[\counters]$ is decidable. Unfortunately the decision 
procedure proposed in~\cite{demri-model-10} involves an exponential reduction to the 
satisfiability problem for some formulae of the Presburger arithmetic and as a consequence has a  high complexity.
\fi 
\begin{theorem}\cite{demri-model-10,Kuhtz&Finkbeiner11}
$\mc{\PLTL[\counters]}{\flatcs}$ can be solved in $4$\exptime.\\
$\mc{\PLTL[\emptyset]}{\flatks}$ restricted to formulae with temporal operators $\until$,$\mynext$
is \np-complete.
\end{theorem}
\iflong
The main goal of this work is to show that we can have a much better upper bounded for 
$\mc{\PLTL[\counters]}{\flatcs}$ and to give in fact the precise complexity of this problem
and for related fragments. 
\fi 
%

\section{Stuttering Theorem for $\PLTL[\emptyset]$}
\label{sec:stuttering}
\iflong

Stuttering of finite words or single letters has been instrumental to show \iflong several \fi 
results about the expressive power of $\PLTL[\emptyset]$ fragments, see e.g.~\cite{Peled&Wilke97,Kucera&Strejcek05};
for instance, $\PLTL[\emptyset]$ restricted to the temporal operator $\until$ characterizes \iflong exactly \fi 
the class of formulae defining  classes of models invariant under stuttering. 
This is refined in~\cite{Kucera&Strejcek05} for $\PLTL[\emptyset]$ restricted to $\until$ and 
$\mynext$, by taking into account not only the $\until$-depth but also the $\mynext$-depth of formulae
and by introducing a principle of stuttering that involves both letter stuttering and word stuttering. 
In this section, we establish another substantial generalization that involves the full logic $\PLTL[\emptyset]$ (with its past-time
temporal operators). Roughly speaking, we show that if
$\amodel_1 \mathbf{s}^M \amodel_2, 0 \models \aformula$ where 
$\amodel_1 \mathbf{s}^M \amodel_2$ is a $\PLTL[\emptyset]$ model ($\amodel_1,\mathbf{s}$ being finite words), 
$\aformula \in \PLTL[\emptyset]$,
$\tempdepth{\aformula} \leq N$ and $M \geq 2N+1$,
then $\amodel_1 \mathbf{s}^{2N+1} \amodel_2, 0 \models \aformula$ (and other related properties). 
Hence,  if there is a run 
\begin{enumerate}
\itemsep 0 cm
\item[(a)] satisfying a path schema $\aschema$ (see Section~\ref{section-path-schemas}) 
and,
\item[(b)] verifying
a $\PLTL[\emptyset]$ formula $\aformula$,
\end{enumerate}
 then there is a run satisfying (a), (b) and each loop is visited
at most $2 \times \tempdepth{\aformula}+ 5$ times, leading to an \np \ upper bound (see
Proposition~\ref{proposition-iter-flatks}).
This extends a result without past-time operators~\cite{Kuhtz&Finkbeiner11}.
Moreover, this turns out to be a key property (Theorem~\ref{theorem-stuttering})
to establish the \np \ upper bound even in the
presence of 
counters (but additional work needs to be done). 
Note that Theorem~\ref{theorem-stuttering} below is interesting for its own sake, independently of our
investigation on flat counter systems. 


Given $M,M',N \in \Nat$, we write $M \equivrel{N} M'$ iff
$\mathtt{Min}(M,N) = \mathtt{Min}(M',N)$. Given 
$\aword= \aword_1 \awordbis^M \aword_2, \aword'=\aword_1 \awordbis^{M'} \aword_2 \in \aalphabet^{\omega}$
and $i,i' \in \Nat$, we define an equivalence relation $\pair{\aword}{i} \equivrel{N} \pair{\aword'}{i'}$
(implicitly parameterized by $\aword_1$, $\aword_2$ and $\awordbis$) such that 
$\pair{\aword}{i} \equivrel{N} \pair{\aword'}{i'}$ means that
the number of copies of $\awordbis$ before position $i$ and 
the number of copies of $\awordbis$ before position $i'$
are related by $~\equivrel{N}~$ and the same applies for the number of copies after the 
positions. Moreover, if $i$ and $i'$ occur in the part where $\awordbis$ is repeated, then they correspond to 
identical positions in $\awordbis$. More formally,
$\pair{\aword}{i} \equivrel{N} \pair{\aword'}{i'}$ $\equivdef$ $M \equivrel{2N} M'$ and one of the following conditions
holds true: 
\begin{enumerate}
\itemsep 0 cm
\item $i,i' < \length{\aword_1} + N \cdot \length{\awordbis}$ and $i=i'$. 
\item $i \geq \length{\aword_1} +  (M-N) \cdot \length{\awordbis}$ and $i' \geq \length{\aword_1} +  (M'-N) \cdot \length{\awordbis}$ and $(i - i') = (M - M') \cdot \length{\awordbis}$.
\item $\length{\aword_1} + N \cdot \length{\awordbis} \leq i < \length{\aword_1} +  (M-N) \cdot \length{\awordbis}$ and $\length{\aword_1} + N \cdot \length{\awordbis} \leq i' < \length{\aword_1} +  (M'-N) \cdot \length{\awordbis}$ and $|i - i'| = 0 \mod \length{\awordbis}$.
\end{enumerate}

Figure~\ref{fig:stut-rel} presents two words 
$\aword$ and $\aword'$ over the alphabet $\aalphabet = \set{\Box, \blacksquare}$
such that $\aword$ is of the form $\aword_1 (\Box \blacksquare)^7 \aword_2$ and 
$\aword'$ is of the form $\aword_1 (\Box \blacksquare)^8 \aword_2$. 
The relation $\equivrel{3}$ is represented by edges between positions: 
each edge from positions $i$ of $\aword$ to positions $i'$ of $\aword'$
represents the fact that $\pair{\aword}{i} \equivrel{3} \pair{\aword'}{i'}$.
\begin{figure}
  \begin{center}
    \begin{tikzpicture}[shorten >=1pt,node distance=0.4mm,auto,>=stealth',
vertex1/.style={minimum size=2pt,draw,rectangle},vertex2/.style={minimum size=2pt,fill,draw,rectangle}]


  \draw (0,0.2) -- node {$w_1$} (1.4,0.2);
  \draw (0,0) -- (0,0.4);
  \draw (1.4,0) -- (1.4,0.4);
  \node (st) at (1.4,0.2) {};
  \node           (d1) [right=of st] {};
  \node[vertex1]  (q1)  [right=of d1] {};
  \node[vertex2]  (q2)  [right=of q1] {};
  \node           (d11) [right=of q2] {};                
  \node[vertex1]  (q3)  [right=of d11] {};
  \node[vertex2]  (q4)  [right=of q3] {};
  \node           (d12) [right=of q4] {};
  \node[vertex1]  (q5)  [right=of d12] {};
  \node[vertex2]  (q6)  [right=of q5] {};
  \node           (d13) [right=of q6] {};
  \node  (endn)  [right=of d13] {$|$};
  \node           (d14) [right=of endn] {};
  \node[vertex1]  (q7)  [right=of d14] {};
  \node[vertex2]  (q8)  [right=of q7] {};
  \node           (d15) [right=of q8] {};
  \node[vertex1]  (q9)  [right=of d15] {};
  \node[vertex2]  (q10)  [right=of q9] {};
  \node           (d16) [right=of q10] {};
  \node  (startn)  [right=of d16] {$|$};
  \node           (d17) [right=of startn] {};
  \node[vertex1]  (q11)  [right=of d17] {};
  \node[vertex2]  (q12)  [right=of q11] {};
  \node           (d18) [right=of q12] {};
  \node[vertex1]  (q13)  [right=of d18] {};
  \node[vertex2]  (q14)  [right=of q13] {};
  \node           (d19) [right=of q14] {};
  \node[vertex1]  (q15)  [right=of d19] {};
  \node[vertex2]  (q16)  [right=of q15] {};
  \node           (d10) [right=of q16] {};
\node (end) [right=of d10] {$|$};

           (q3)   edge  node  {} (q4)
           (q5)   edge  node  {} (q6)
          (q7)   edge  node  {} (q8)
          (q9)   edge  node  {} (q10)
          (q11)   edge  node  {} (q12)
          (q13)   edge  node  {} (q14)
          (q15)   edge  node  {} (q16);
\draw (11.6,0.2) -- node {$w_2$} (13,0.2);
\draw (13,0) -- (13,0.4);


  \draw (0,2.2) -- node {$w_1$} (1.4,2.2);
  \draw (0,2) -- (0,2.4);
  \draw (1.4,2) -- (1.4,2.4);
  \node (st0) at (1.4,2.2) {};
  \node           (d2) [right=of st0] {};
  \node[vertex1]  (q01)  [right=of d2] {};
  \node[vertex2]  (q02)  [right=of q01] {};
  \node           (d21)  [right=of q02] {};
  \node[vertex1]  (q03)  [right=of d21] {};
  \node[vertex2]  (q04)  [right=of q03] {};
  \node           (d22)  [right=of q04] {};
  \node[vertex1]  (q05)  [right=of d22] {};
  \node[vertex2]  (q06)  [right=of q05] {};
  \node           (d23)  [right=of q06] {};
  \node  (endn)  [right=of d23] {$|$};
  \node           (d24)  [right=of endn] {};
  \node  (q07)  [right=of d24] {};
  \node[vertex1]  (q08)  [right=of q07] {};
  \node[vertex2]  (q09)  [right=of q08] {};
  \node  (q010)  [right=of q09] {};
  \node           (d25)  [right=of q010] {};
  \node           (d26)  [right=of d25] {};
  \node  (startn)  [right=of d26] {$|$};
  \node           (d27)  [right=of startn] {};
  \node[vertex1]  (q011)  [right=of d27] {};
  \node[vertex2]  (q012)  [right=of q011] {};
  \node           (d28)  [right=of q012] {};
  \node[vertex1]  (q013)  [right=of d28] {};
  \node[vertex2]  (q014)  [right=of q013] {};
  \node           (d29)  [right=of q014] {};
  \node[vertex1]  (q015)  [right=of d29] {};
  \node[vertex2]  (q016)  [right=of q015] {};
  \node           (d20)  [right=of q016] {};

\node (end) [right=of d20] {$|$};

           (q03)   edge  node  {} (q04)
           (q05)   edge  node  {} (q06)
           (q08)   edge  node  {} (q09)
           (q011)   edge  node  {} (q012)
           (q013)   edge  node  {} (q014)
          (q015)   edge  node  {} (q016);
\draw (11.6,2.2) -- node {$w_2$} (13,2.2);
\draw (13,2) -- (13,2.4);


\path[-]  (q01)   edge  node  {} (q1)
           (q02)   edge  node  {} (q2)
           (q03)   edge  node  {} (q3)
           (q04)   edge  node  {} (q4)
           (q05)   edge  node  {} (q5)
           (q06)   edge  node  {} (q6)
           (q08)   edge  node  {} (q7)
           (q08)   edge  node  {} (q9)
           (q09)   edge  node  {} (q8)
           (q09)   edge  node  {} (q10)
           (q011)   edge  node  {} (q11)
           (q012)   edge  node  {} (q12)
           (q013)   edge  node  {} (q13)
           (q014)   edge  node  {} (q14)
           (q015)   edge  node  {} (q15)
           (q016)   edge  node  {} (q16);

  
\end{tikzpicture}
  \end{center}
  \caption{Two words $\aword$, $\aword'$ with $\awordbis = \Box \blacksquare$ and the relation $\equivrel{3}$}
  \label{fig:stut-rel}
\end{figure}


In order to prove our stuttering theorem for $\PLTL[\emptyset]$, we need to express some properties concerning the relation $\equivrel{}$ whose proofs can be found in the subsequent subsections. Let 
$\aword= \aword_1 \awordbis^M \aword_2, \aword'=\aword_1 \awordbis^{M'} \aword_2 \in \aalphabet^{\omega}$, 
$i,i' \in \Nat$ and $N \geq 2$ such that $M,M' \geq 2N+1$ and $\pair{\aword}{i} \equivrel{N} \pair{\aword'}{i'}$. We can show the following properties:

\begin{description}
\itemsep 0 cm
\item[(Claim 1)]  $\pair{\aword}{i} \equivrel{N-1} \pair{\aword'}{i'}$ and $\aword(i) = \aword'(i')$.
\item[(Claim 2)]  $i,i'>0$ implies $\pair{\aword}{i-1} \equivrel{N-1} \pair{\aword'}{i'-1}$.
\item[(Claim 3)]  $\pair{\aword}{i+1} \equivrel{N-1} \pair{\aword'}{i'+1}$
\item[(Claim 4)]  For all $j \geq i$, there is $j' \geq i'$ such that
                  $\pair{\aword}{j} \equivrel{N-1} \pair{\aword'}{j'}$ and 
                  for all $k' \in \interval{i'}{j'-1}$, there is
                  $k \in  \interval{i}{j-1}$ such that 
                  $\pair{\aword}{k} \equivrel{N-1} \pair{\aword'}{k'}$.
\item[(Claim 5)] For all $j \leq i$, there is $j' \leq i'$ 
      such that  $\pair{\aword}{j} \equivrel{N-1} \pair{\aword'}{j'}$ and 
       for all $k' \in \interval{j'-1}{i'}$, there is
            $k \in  \interval{j-1}{i}$ such that $\pair{\aword}{k} \equivrel{N-1} \pair{\aword'}{k'}$.
\end{description}

We now state our stuttering theorem for $\PLTL[\emptyset]$ that is tailored for our future needs.
\begin{theorem}[Stuttering]
\label{theorem-stuttering}
Let $\amodel = \amodel_1 \mathbf{s}^M \amodel_2, \amodel'= \amodel_1  \mathbf{s}^{M'} 
\amodel_2 \in 
(\powerset{\varprop})^{\omega}$
and  $i,i' \in \Nat$ such that $N \geq 2$, $M,M' \geq 2N+1$ and
$\pair{\amodel}{i} \equivrel{N} \pair{\amodel'}{i'}$. Then,
for every  $\PLTL[\emptyset]$ formula $\aformula$ with $\tempdepth{\aformula} \leq N$, we have
$\amodel, i \models \aformula$ iff $\amodel', i \models \aformula$. 
\end{theorem}

\begin{proof} (sketch) The proof is by structural induction on the formula but first we need to establish 
properties

By way of example, let us 
present the induction step for subformulae of the form $\aformulabis_1 \until \aformulabis_2$. 
We show that $\amodel, i \models \aformulabis_1 \until \aformulabis_2$ implies 
 $\amodel', i' \models \aformulabis_1 \until \aformulabis_2$. 
Suppose there is  $j \geq i$ such that $\amodel, j \models \aformulabis_2$ and 
for every $k \in \interval{i}{j-1}$, we have $\amodel, k \models \aformulabis_1$. 
There is $j' \geq i'$ satisfying (Claim 4).  
Since $\tempdepth{\aformulabis_1}, \tempdepth{\aformulabis_2} \leq N-1$, 
by (IH), we have  $\amodel', j' \models \aformulabis_2$.
Moreover, for every $k' \in \interval{i'}{j'-1}$, there is
                  $k \in  \interval{i}{j-1}$ such that 
                  $\pair{\aword}{k} \equivrel{N-1} \pair{\aword'}{k'}$ and 
by (IH), we have $\amodel', k' \models \aformulabis_1$ 
for every $k' \in \interval{i'}{j'-1}$. 
Hence, $\amodel', i' \models \aformulabis_1 \until \aformulabis_2$.
\end{proof}

An alternative proof consists in using 
Ehrenfeucht-Fra\"{\i}ss{\'e} games~\cite{Etessami&Wilke00}. This will not provide necessarily
a shorter proof and it requires to use properties from the games in~\cite{Etessami&Wilke00}.
The forthcoming proofs for claims are self-contained. 

\subsection{A  Zone Classification for Proving (Claim 1) -- (Claim 5)}

For the proofs of (Claim 1) -- (Claim 5), 
the positions of each word $\aword$ of the form 
$\aword =\aword_1 \awordbis^M \aword_2 \in \aalphabet^{\omega}$  ($\aword_1 \in \aalphabet^*$, 
$\awordbis \in \aalphabet^+$ and $\aword_2 \in \aalphabet^{\omega}$)
with $M >  2N$ are partitionned into five zones (A, B, C, D and E).
We also assume that $N \geq 2$. 
Indeed, given that $\pair{\aword}{i} \equivrel{N} \pair{\aword'}{i'}$,
we shall proceed  by a case analysis on the positions $i$ and $i'$ 
depending on which zones $i$ and $i'$ belong to. 
The definition of zones is illustrated on Figure~\ref{fig:stut-zones} and here is the formal
characterization: 
\begin{itemize}
\itemsep 0 cm 
\item Zone A corresponds to the set of  positions $i \in \nat$ such that 
      $0 \leq i < \length{\aword_1} + (N-1) \cdot \length{\awordbis}$.
\item Zone B corresponds to the set of positions $i \in \nat$ such that 
      $\length{\aword_1} + (N-1) \cdot \length{\awordbis} \leq i < \length{\aword_1} + N \cdot \length{\awordbis}$.
\item Zone C corresponds to the set of  positions $i \in \nat$ such that $\length{\aword_1} + N \cdot \length{\awordbis} \leq i < \length{\aword_1} + (M-N) \cdot\length{\awordbis} $.
\item Zone D corresponds to the set of positions $i \in \nat$ such that 
      $\length{\aword_1} + (M-N) \cdot \length{\awordbis} \leq i < \length{\aword_1} + (M-(N-1)) \cdot \length{\awordbis} $.
\item Zone E corresponds to the set of positions $i \in \nat$ such that 
      $\length{\aword_1} + (M-(N-1)) \cdot \length{\awordbis} \leq i$.
\end{itemize}

\begin{figure}
  \begin{center}
    \begin{tikzpicture}[shorten >=1pt,node distance=0.4mm,auto,>=stealth',
vertex1/.style={minimum size=2pt,draw,rectangle},vertex2/.style={minimum size=2pt,fill,draw,rectangle}]

\filldraw[fill=black!20, draw=black!20] (-0.2,-0.2) rectangle +(3.9,0.8)
                                        (3.8,-0.2) rectangle +(0.9,0.8)
                                        (5.3,-0.2) rectangle +(2.3,0.8)
                                        (8.2,-0.2) rectangle +(0.9,0.8)
                                        (9.2,-0.2) rectangle +(3.9,0.8);
  \node (zoneA) at (2,0.8) {A};
  \node (zoneB) at (4.25,0.8) {B};
  \node (zoneC) at (6.4,0.8) {C};
  \node (zoneD) at (8.7,0.8) {D};
  \node (zoneE) at (11.2,0.8) {E};

  \draw (0,0.2) -- node {} (1.4,0.2);
  \draw (0,0) -- (0,0.4);
  \draw (1.4,0) -- (1.4,0.4);
  \node (st) at (1.4,0.2) {};
  \node           (d1)  [right=of st] {};
  \node[vertex1]  (q1)  [right=of d1] {};
  \node[vertex2]  (q2)  [right=of q1] {};
  \node           (d2)  [right=of q2] {};
  \node[vertex1]  (q3)  [right=of d2] {};
  \node[vertex2]  (q4)  [right=of q3] {};
  \node           (d3)  [right=of q4] {};
  \node[vertex1]  (q5)  [right=of d3] {};
  \node[vertex2]  (q6)  [right=of q5] {};
  \node           (d4)  [right=of q6] {};
  \node  (endn)  [right=of d4] {$|$};
  \node           (d5)  [right=of endn] {};
  \node[vertex1]  (q8)  [right=of d5] {};
  \node[vertex2]  (q9)  [right=of q8] {};
  \node           (d6)  [right=of q9] {};
  \node[vertex1]  (q10)  [right=of d6] {};
  \node[vertex2]  (q11)  [right=of q10] {};
  \node           (d7)  [right=of q11] {};
  \node  (startn)  [right=of d7] {$|$};
  \node           (d8)  [right=of startn] {};
  \node[vertex1]  (q13)  [right=of d8] {};
  \node[vertex2]  (q14)  [right=of q13] {};
  \node           (d9)  [right=of q14] {};
  \node[vertex1]  (q15)  [right=of d9] {};
  \node[vertex2]  (q16)  [right=of q15] {};
  \node           (d0)  [right=of q16] {};
  \node[vertex1]  (q17)  [right=of d0] {};
  \node[vertex2]  (q18)  [right=of q17] {};
  \node           (da)  [right=of q18] {};

\node (end) [right=of da] {$|$};

          (q3)   edge  node  {} (q4)
          (q5)   edge  node  {} (q6)
          (q8)   edge  node  {} (q9)
          (q10)   edge  node  {} (q11)
          (q13)   edge  node  {} (q14)
          (q15)   edge  node  {} (q16)
          (q17)   edge  node  {} (q18);
\draw (11.6,0.2) -- (13,0.2);
\draw (13,0) -- (13,0.4);

\end{tikzpicture}
  \end{center}
  \caption{The five zones for  $\aword_1 (\Box \blacksquare)^8 \aword_2$ with $N = 3$ and $\awordbis = 
  \Box \blacksquare$}
  \label{fig:stut-zones}
\end{figure}

Note that the definition of zones depends on the value $N$ (taken from $\equivrel{N}$) and also
on $\awordbis$, $\aword_1$ and $\aword_2$. In the sequel, we may index the zones 
by $N$ (providing e.g., A$_N$, B$_N$ etc.) 
when it is useful
to make explicit from which relation $\equivrel{N}$ the definition of zones is made.
Moreover, we may use a prime (providing for instance, A$'_N$, B$'_N$ etc.) to refer to zones 
for $\aword'$. 
So, the relation $\equivrel{N}$ can be redefined as follows when 
$M, M' > 2N$: 
$\pair{\aword}{i} \equivrel{N} \pair{\aword'}{i'}$ $\equivdef$ ($M \equivrel{2N} M'$ and)
 one of the conditions
holds true: 
\begin{enumerate}
\itemsep 0 cm
\item $i = i'$ and either ($i \in {\rm A}_N$ and $i' \in {\rm A}_N'$) or
       ($i \in {\rm B}_N$ and $i' \in {\rm B}_N'$).   
\item   $(i - i') = (M - M') \length{\awordbis}$ and either 
        ($i \in {\rm D}_N$ and $i' \in {\rm D}_N'$) or
       ($i \in {\rm E}_N$ and $i' \in {\rm E}_N'$).
\item $i \in {\rm C}_N$,  $i' \in {\rm C}_N'$ and 
      $|i - i'| = 0 \mod \length{\awordbis}$.
\end{enumerate}

\subsection{Proof of (Claim 1)} 

Before the proof, let us recall what is (Claim 1). 
Let 
$\aword= \aword_1 \awordbis^M \aword_2, \aword'=\aword_1 \awordbis^{M'} \aword_2 \in \aalphabet^{\omega}$, 
$i,i' \in \Nat$ and $N \geq 2$ such that $M,M' \geq 2N+1$ and $\pair{\aword}{i} \equivrel{N} \pair{\aword'}{i'}$.
\begin{description}
\itemsep 0 cm
\item[(Claim 1)]  $\pair{\aword}{i} \equivrel{N-1} \pair{\aword'}{i'}$; $\aword(i) = \aword'(i')$.
\end{description}

\begin{proof}
Let us first prove that $\pair{\aword}{i} \equivrel{N-1} \pair{\aword'}{i'}$.
Without any loss of generality, we can assume that $M \geq M'$.
Since $N > N-1$, it is obvious that $M \equivrel{2(N-1)} M'$.
\begin{itemize}
\itemsep 0 cm
\item If $i< \length{\aword_1}+(N-1) \cdot \length{\awordbis}$
      \commzone{$i$ is Zone A$_N$}, then 
       $i=i'$. Hence either ($i \in {\rm A}_{N-1}$, $i' \in {\rm A}_{N-1}'$ and $i = i'$)
      or   ($i \in {\rm B}_{N-1}$, $i' \in {\rm B}_{N-1}'$ and $i = i'$). Hence, 
      $\pair{\aword}{i} \equivrel{N-1} \pair{\aword'}{i'}$.
\item If $i \geq \length{\aword_1} +  (M-(N-1)) \cdot \length{\awordbis}$ \commzone{$i$ is in zone E$_N$} then 
 $i=i'+(M-M') \cdot \length{\awordbis}$ and 
 $i' \geq \length{\aword_1} +  (M'-(N-1))\cdot \length{\awordbis}$ \commzone{$i'$ is in zone E$_N'$}.
So, either ($i$ is in zone E$_{N-1}$  and $i'$ is in zone E$_{N-1}'$) 
    or  ($i$ is in zone D$_{N-1}$  and $i'$ is in zone D$_{N-1}'$).
Since $i=i'+(M-M') \cdot \length{\awordbis}$, we conclude that  $\pair{\aword}{i} \equivrel{N-1} \pair{\aword'}{i'}$.
\item If $\length{\aword_1} + (N-1) \cdot \length{\awordbis} \leq i < \length{\aword_1} +  
      N \cdot \length{\awordbis}$ \commzone{$i$ is in Zone B$_N$} then $i = i'$.
      Hence, $i \in {\rm C}_{N-1}$,  $i' \in {\rm C}_{N-1}'$ and 
      $|i - i'| = 0 \mod \length{\awordbis}$. Hence,  $\pair{\aword}{i} \equivrel{N-1} \pair{\aword'}{i'}$.

\item If $\length{\aword_1} + N \cdot \length{\awordbis} \leq i < \length{\aword_1} +  
      (M-N) \cdot \length{\awordbis}$ \commzone{$i$ in Zone C$_N$}, then 
      $\length{\aword_1} + N \cdot \length{\awordbis} \leq i' < \length{\aword_1} +  
      (M'-N) \cdot \length{\awordbis}$ \commzone{$i'$ is in Zone C$'_N$} and $|i - i'| = 0 \mod \length{\awordbis}$. 
      Consequently, $i$ is in Zone C$_{N-1}$, $i'$ is in Zone C$'_{N-1}$ and $|i - i'| = 0 \mod \length{\awordbis}$. 
      This entails that $\pair{\aword}{i} \equivrel{N-1} \pair{\aword'}{i'}$.

\item If $\length{\aword_1} + (M-N) \cdot \length{\awordbis} \leq i < \length{\aword_1} +  
(M-(N-1)) \cdot \length{\awordbis}$ \commzone{$i$ in Zone D$_N$}, then $i'$ is in Zone D$'_N$ and  
$i=i'+(M-M') \cdot \length{\awordbis}$.  Consequently,
$i$ is in Zone C$_{N-1}$, $i'$ is in Zone C$'_{N-1}$ and $|i - i'| = 0 \mod \length{\awordbis}$. 
This also entails that $\pair{\aword}{i} \equivrel{N-1} \pair{\aword'}{i'}$.
\end{itemize}

As far as the second property is concerned, it is also clear that $\aword(i) = \aword'(i')$, because 
either $i$ and $i'$ are at the same position in the word $\aword_1$ or $\aword_2$ either they are pointing some positions in the portions of the word which belong to $\awordbis^+$ and since their difference will be such that
 $|i - i'| = 0 \mod \length{\awordbis}$, it is easy to see that $i$ and $i'$ will point at the same position in  $\awordbis$.
\end{proof}

\subsection{Proof of (Claim 2)} 
Before the proof, let us recall what is (Claim 2).
Let 
$\aword= \aword_1 \awordbis^M \aword_2, \aword'=\aword_1 \awordbis^{M'} \aword_2 \in \aalphabet^{\omega}$, 
$i,i' \in \Nat$ and $N \geq 2$ such that  $M,M' \geq 2N+1$ and $\pair{\aword}{i} \equivrel{N} \pair{\aword'}{i'}$.
\begin{description}
\itemsep 0 cm
\item[(Claim 2)]  $i,i'>0$ implies $\pair{\aword}{i-1} \equivrel{N-1} \pair{\aword'}{i'-1}$.
\end{description}

\begin{proof}Without any loss of generality, we can assume that $M \geq M'$.
Since $N > N-1$, it is obvious that $M \equivrel{2(N-1)} M'$.
\begin{itemize}
\itemsep 0 cm
\item If $i< \length{\aword_1}+(N-1) \cdot \length{\awordbis}$
      \commzone{$i$ is Zone A$_N$}, then 
       $i=i'$. Hence, $i-1 \in {\rm A}_{N-1}$, $i'-1 \in {\rm A}_{N-1}'$ and $i-1 = i'-1$. 
       So, 
      $\pair{\aword}{i-1} \equivrel{N-1} \pair{\aword'}{i'-1}$.
\item If $i \geq \length{\aword_1} +  (M-(N-1)) \cdot \length{\awordbis}$ \commzone{$i$ is in zone E$_N$} then 
 $i=i'+(M-M') \cdot \length{\awordbis}$ and 
 $i' \geq \length{\aword_1} +  (M'-(N-1)) \cdot \length{\awordbis}$ \commzone{$i'$ is in zone E$_N'$}.
So, either ($i-1$ is in zone E$_{N-1}$, $i'-1$ is in zone E$_{N-1}'$ and $i-1=i'-1+(M-M') \cdot \length{\awordbis}$) 
    or  ($i-1$ is in zone D$_{N-1}$  and $i'-1$ is in zone D$_{N-1}'$ and $i-1=i'-1+(M-M') \cdot \length{\awordbis}$)
    or  ($i-1$ is in zone C$_{N-1}$  and $i'-1$ is in zone C$_{N-1}'$ and 
         $|(i-1) - (i'-1)| = 0 \mod \length{\awordbis}$).
   We conclude that  $\pair{\aword}{i-1} \equivrel{N-1} \pair{\aword'}{i'-1}$.
\item If $\length{\aword_1} + (N-1) \cdot \length{\awordbis} \leq i < \length{\aword_1} +  
      N \cdot \length{\awordbis}$ \commzone{$i$ is in Zone B$_N$} then $i = i'$.
      Hence, either ($i-1 \in {\rm C}_{N-1}$,  $i'-1 \in {\rm C}_{N-1}'$ and 
      $|(i-1) - (i'-1)| = 0 \mod \length{\awordbis}$) or 
      ($i-1 \in {\rm B}_{N-1}$,  $i'-1 \in {\rm B}_{N-1}'$ and 
      $i-1 = i'-1$). Hence,  $\pair{\aword}{i-1} \equivrel{N-1} \pair{\aword'}{i'-1}$.

\item If $\length{\aword_1} + N \cdot \length{\awordbis} \leq i < \length{\aword_1} +  
      (M-N) \cdot \length{\awordbis}$ \commzone{$i$ in Zone C$_N$}, then 
      $\length{\aword_1} + N \cdot \length{\awordbis} \leq i' < \length{\aword_1} +  
      (M'-N) \cdot \length{\awordbis}$ \commzone{$i'$ is in Zone C$'_N$} and $|i - i'| = 0 \mod \length{\awordbis}$. 
      Consequently, $i-1$ is in Zone C$_{N-1}$, $i'-1$ is in Zone C$'_{N-1}$ and 
      $|(i-1) - (i'-1)| = 0 \mod \length{\awordbis}$.
      This entails that $\pair{\aword}{i-1} \equivrel{N-1} \pair{\aword'}{i'-1}$.

\item If $\length{\aword_1} + (M-N) \cdot \length{\awordbis} \leq i < \length{\aword_1} +  (M-(N-1)) \cdot \length{\awordbis}$ \commzone{$i$ in Zone D$_N$}, then $i'$ is in Zone D$'_N$ and  
$i=i'+(M-M') \cdot \length{\awordbis}$.  Consequently,
$i-1$ is in Zone C$_{N-1}$, $i'-1$ is in Zone C$'_{N-1}$ and $|(i-1) - (i'-1)| = 0 \mod \length{\awordbis}$. 
This  entails that $\pair{\aword}{i-1} \equivrel{N-1} \pair{\aword'}{i'-1}$.
\end{itemize}
\end{proof}

\subsection{Proof of (Claim 3)} 
The proof proceeds in a similar way as the proof of Claim 2. For completeness shake the proof is provided in the Appendix.
\cut{
Before the proof, let us recall what is (Claim 3).
Let 
$\aword= \aword_1 \awordbis^M \aword_2, \aword'=\aword_1 \awordbis^{M'} \aword_2 \in \aalphabet^{\omega}$, 
$i,i' \in \Nat$ and $N \geq 2$ such that  $M,M' \geq 2N+1$ and $\pair{\aword}{i} \equivrel{N} \pair{\aword'}{i'}$.
\begin{description}
\itemsep 0 cm
\item[(Claim 3)]  $\pair{\aword}{i+1} \equivrel{N-1} \pair{\aword'}{i'+1}$.
\end{description}

\begin{proof} 
The proof is similar to the proof for (Claim 2). Nevertheless, full proof is provided below
for the sake of completeness.
Without any loss of generality, we can assume that $M \geq M'$.
Since $N > N-1$, it is obvious that $M \equivrel{2(N-1)} M'$.
\begin{itemize}
\itemsep 0 cm
\item If $i< \length{\aword_1}+(N-1) \cdot \length{\awordbis}$
      \commzone{$i$ is Zone A$_N$}, then 
       $i=i'$. Hence either ($i+1 \in {\rm A}_{N-1}$, $i'+1 \in {\rm A}_{N-1}'$ and $i+1 = i'+1$)
      or   ($i+1 \in {\rm B}_{N-1}$, $i'+1 \in {\rm B}_{N-1}'$ and $i+1 = i'+1$) or 
       ($i+1 \in {\rm C}_{N-1}$, $i'+1 \in {\rm C}_{N-1}'$ and $i - i' = 0$). Hence, 
      $\pair{\aword}{i+1} \equivrel{N-1} \pair{\aword'}{i'+1}$.
\item If $i \geq \length{\aword_1} +  (M-(N-1)) \cdot \length{\awordbis}$ \commzone{$i$ is in zone E$_N$} then 
 $i=i'+(M-M') \cdot \length{\awordbis}$ and 
 $i' \geq \length{\aword_1} +  (M'-(N-1)) \cdot \length{\awordbis}$ \commzone{$i'$ is in zone E$_N'$}.
So, either ($i+1$ is in zone E$_{N-1}$  and $i'+1$ is in zone E$_{N-1}'$) 
    or  ($i+1$ is in zone D$_{N-1}$  and $i'+1$ is in zone D$_{N-1}'$).
Since $i+1=i'+1+(M-M') \cdot \length{\awordbis}$, we conclude that  $\pair{\aword}{i+1} \equivrel{N-1} \pair{\aword'}{i'+1}$.
\item If $\length{\aword_1} + (N-1) \cdot \length{\awordbis} \leq i < \length{\aword_1} +  
      N \cdot \length{\awordbis}$ \commzone{$i$ is in Zone B$_N$} then $i = i'$.
      Hence, $i+1 \in {\rm C}_{N-1}$,  $i'+1 \in {\rm C}_{N-1}'$ and 
      $|(i+1) - (i'+1)| = 0 \mod \length{\awordbis}$. Hence,  $\pair{\aword}{i+1} \equivrel{N-1} \pair{\aword'}{i'+1}$.

\item If $\length{\aword_1} + N \cdot \length{\awordbis} \leq i < \length{\aword_1} +  
      (M-N) \cdot \length{\awordbis}$ \commzone{$i$ in Zone C$_N$}, then 
      $\length{\aword_1} + N \cdot \length{\awordbis} \leq i' < \length{\aword_1} +  
      (M'-N) \cdot \length{\awordbis}$ \commzone{$i'$ is in Zone C$'_N$} and $|i - i'| = 0 \mod \length{\awordbis}$. 
      Consequently, $i+1$ is in Zone C$_{N-1}$, $i'+1$ is in Zone C$'_{N-1}$ and 
      $|(i+1) - (i'+1)| = 0 \mod \length{\awordbis}$.
      This entails that $\pair{\aword}{i+1} \equivrel{N-1} \pair{\aword'}{i'+1}$.

\item If $\length{\aword_1} + (M-N) \cdot \length{\awordbis} \leq i < \length{\aword_1} +  (M-(N-1)) \cdot \length{\awordbis}$ \commzone{$i$ in Zone D$_N$}, then $i'$ is in Zone D$'_N$ and  
$i=i'+(M-M') \cdot \length{\awordbis}$.  Consequently,
either ($i+1$ is in Zone C$_{N-1}$, $i'+1$ is in Zone C$'_{N-1}$ and $|(i+1) - (i'+1)| = 0 \mod \length{\awordbis}$)
or ($i+1$ is in Zone D$_{N-1}$, $i'+1$ is in Zone D$'_{N-1}$ and 
$i+1=i'+1 +(M-M') \cdot \length{\awordbis}$). 
This also entails that $\pair{\aword}{i+1} \equivrel{N-1} \pair{\aword'}{i'+1}$.
\end{itemize}
\end{proof}

}

\subsection{Proof of (Claim 4)} 

Before providing the detailed proof, 
we give a concrete example on Figure~\ref{fig:stut-rel-n-1}. On this example,
 we assume that the top word $\aword$ and the bottom word $\aword'$ and their respective positions 
$i$ and $i'$ are such that $\pair{\aword}{i} \equivrel{3} \pair{\aword'}{i'}$. 
We want to illustrate (Claim 4) and for this matter, we choose a position $j$ in 
$\aword$. Now observe that according to the zone classification, 
$j$ is in the Zone C of the word $\aword$ and furthermore it is not possible to find 
a $j' > i'$ in the Zone C of the word $\aword'$ such that $j$ and $j'$ points on the same position of 
the word $\awordbis$. That is why we need to consider at this stage not the relation $\equivrel{3}$ but 
instead $\equivrel{2}$. In fact, as shown on the bottom of Figure~\ref{fig:stut-rel-n-1}, we can find for $j$, 
a position $j'$ in $\aword'$ such that $\pair{\aword}{j} \equivrel{2} \pair{\aword'}{j'}$ 
(take $j=j'$) and this figure also shows that for all $i' \leq k \leq j'$, $\pair{\aword}{k} 
\equivrel{2} \pair{\aword'}{k}$.
\begin{figure}
  \begin{center}
    \input{stut-rel-n-1.tikz}
  \end{center}
  \caption{Relation between $\equivrel{N}$ and $\equivrel{N-1}$}
  \label{fig:stut-rel-n-1}
\end{figure}

Before the proof, let us recall what is (Claim 4). 
Let 
$\aword= \aword_1 \awordbis^M \aword_2, \aword'=\aword_1 \awordbis^{M'} \aword_2 \in \aalphabet^{\omega}$, 
$i,i' \in \Nat$ and $N \geq 2$ such that  $M,M' \geq 2N+1$ and $\pair{\aword}{i} \equivrel{N} \pair{\aword'}{i'}$.
We can show the following properties:
\begin{description}
\itemsep 0 cm
\item[(Claim 4)]  For all $j \geq i$, there is $j' \geq i'$ such that
                  $\pair{\aword}{j} \equivrel{N-1} \pair{\aword'}{j'}$ and 
                  for all $k' \in \interval{i'}{j'-1}$, there is
                  $k \in  \interval{i}{j-1}$ such that 
                  $\pair{\aword}{k} \equivrel{N-1} \pair{\aword'}{k'}$.
\end{description}

\begin{proof} We proceed by a case analysis on the positions  $i$ and $j$.
Without any loss of generality, we can assume that $M \geq M'$.
\begin{itemize}
\itemsep 0 cm 
\item If $i \geq \length{\aword_1}+(M-N)\cdot \length{\awordbis}$ 
      \commzone{$i$ is in Zone D or E} then $j \geq \length{\aword_1}+(M-N)\cdot \length{\awordbis}$ \commzone{$j$ is in Zone D or E} and $i' \geq \length{\aword_1} +  (M'-N)\cdot \length{\awordbis}$ \commzone{$i'$ is Zone D or E} and $i=i'+(M-M')\cdot \length{\awordbis}$. We define $j'=j-(M-M')\cdot \length{\awordbis}$. Then it is clear that 
$j' \geq i'$ and  $\pair{\aword}{j} \equivrel{N} \pair{\aword'}{j'}$. By 
(Claim 1), we get $\pair{\aword}{j} \equivrel{N-1} \pair{\aword'}{j'}$.
Let $k' \in \interval{i'}{j'-1}$ and let $k=k'+(M-M')\cdot \length{\awordbis}$, then we have that 
$k \in \interval{i}{j-1}$ and also $\pair{\aword}{k} \equivrel{N} \pair{\aword'}{k'}$, 
hence by (Claim 1), $\pair{\aword}{k} \equivrel{N-1} \pair{\aword'}{k'}$.
\item If $i < \length{\aword_1}+N\cdot \length{\awordbis}$ \commzone{$i$ is in Zone A or B} then $i'< \length{\aword_1}+N\cdot \length{\awordbis}$ \commzone{$i'$ is in Zone A or B} and $i=i'$ and we have the following possibilities for 
the position $j \geq i$:
\begin{itemize}
\itemsep 0 cm 
\item If $j < \length{\aword_1}+N \cdot \length{\awordbis}$ \commzone{$j$ is in Zone A or B}, then 
let $j'=j$. Consequently we have $\pair{\aword}{j} \equivrel{N} \pair{\aword'}{j'}$ and by (Claim 1) we 
get  $\pair{\aword}{j} \equivrel{N-1} \pair{\aword'}{j'}$. 
Let $k' \in \interval{i'}{j'-1}$ and $k=k'$. Then we have that $k \in \interval{i}{j-1}$ 
and also 
$\pair{\aword}{k} \equivrel{N} \pair{\aword'}{k'}$ and 
by (Claim 1), $\pair{\aword}{k} \equivrel{N-1} \pair{\aword'}{k'}$.
\item If $\length{\aword_1}+N \cdot \length{\awordbis} \leq j < \length{\aword_1}+(M-N) 
      \cdot \length{\awordbis}$ \commzone{$j$ is in Zone C}, then let 
       $\ell=(j-(\length{\aword_1}+N \cdot \length{\awordbis})) \mod \length\awordbis$ 
       ($\ell$ the relative position of $j$ in the word $\awordbis$ it belongs to). 
       Consequently $0 \leq \ell < \length{\awordbis}$. 
       Let $j'= \length{\aword_1} + N \cdot \length{\awordbis}+\ell$ (we choose $j'$ at the same relative 
       position of $j$ in the first word $\awordbis$ of the Zone C). 
       Then $\length{\aword_1}+N \cdot \length{\awordbis} \leq j' < \length{\aword_1}+(M'-N) \cdot \length{\awordbis}$ \commzone{$j'$ is in Zone C} (because $(M'-N)>0$) and $|j - j'| = 0 \mod \length{\awordbis}$. 
We deduce that $\pair{\aword}{j} \equivrel{N} \pair{\aword'}{j'}$ and by 
(Claim 1)  we get $\pair{\aword}{j} \equivrel{N-1} \pair{\aword'}{j'}$.
Then let $k' \in \interval{i'}{j'-1}$ and let $k=k'$. 
Then we have that $k \in \interval{i}{j-1}$. Furthermore, 
 if $k' < \length{\aword_1}+N \cdot \length{\awordbis}$ \commzone{$k'$ is in Zone A or B} 
 we obtain $\pair{\aword}{k} \equivrel{N} \pair{\aword'}{k'}$ and by (Claim 1),
 $\pair{\aword}{k} \equivrel{N-1} \pair{\aword'}{k'}$. 
 Moreover, if $\length{\aword_1}+N \cdot \length{\awordbis} \leq k'$ 
 \commzone{$k'$ is in Zone C} then $k$ is in Zone C and 
 $|k - k'| = 0 \mod \length{\awordbis}$ since  $k=k'$.
 So, $\pair{\aword}{k} \equivrel{N} \pair{\aword'}{k'}$ and by (Claim 1),
 $\pair{\aword}{k} \equivrel{N-1} \pair{\aword'}{k'}$. 
\item If $ \length{\aword_1}+(M-N) \cdot \length{\awordbis} \leq j$ 
 \commzone{$j$ is in Zone E or D}, let $j'=j-(M-M')\cdot \length{\awordbis}$. 
 Then, we have $ \length{\aword_1}+(M'-N) \cdot \length{\awordbis} \leq j'$ \commzone{$j'$ is in Zone D or E} 
 and we deduce that $\pair{\aword}{j} \equivrel{N} \pair{\aword'}{j'}$ and 
 by (Claim 1) we get $\pair{\aword}{j} \equivrel{N-1} \pair{\aword'}{j'}$. 
 Then let $k' \in \interval{i'}{j'-1}$. 
 If $k' < \length{\aword_1}+N \cdot \length{\awordbis}$ \commzone{$k'$ is in Zone A or B}, 
 for $k=k'$, we obtain 
  $\pair{\aword}{k} \equivrel{N} \pair{\aword'}{k'}$ 
  and by (Claim 1),
 $\pair{\aword}{k} \equivrel{N-1} \pair{\aword'}{k'}$. 
 If $k' \geq \length{\aword_1}+(M'-N) \cdot \length{\awordbis}$ 
 \commzone{$k'$ is in Zone D or E}, we choose $k=k'+(M-M') \cdot \length{\awordbis}$ and here also we deduce 
 $\pair{\aword}{k} \equivrel{N} \pair{\aword'}{k'}$ 
  and by (Claim 1),
 $\pair{\aword}{k} \equivrel{N-1} \pair{\aword'}{k'}$.  
 If ${\aword_1}+N \cdot \length{\awordbis} \leq k' < \length{\aword_1}+(M'-N) \cdot \length{\awordbis}$ 
 \commzone{$k'$ is in Zone C}, let $\ell=(k'-(\length{\aword_1}+N \cdot \length{\awordbis})) \mod \length{\awordbis}$ 
 ($\ell$ is the relative position of $k'$ in the word $\awordbis$ it belongs to) and let 
 $k=\length{\aword_1}+N \cdot \length{\awordbis}+\ell$ ($k$ is placed at the same relative position  
 of $k'$ in the first word $\awordbis$ of the Zone C). Then we have 
 ${\aword_1}+N \cdot \length{\awordbis} \leq k < \length{\aword_1}+(M-N) \cdot \length{\awordbis}$ and $|k-k'| =0 
 \mod \length{\awordbis}$ which allows to deduce that 
  $\pair{\aword}{k} \equivrel{N} \pair{\aword'}{k'}$ and by (Claim 1),
 $\pair{\aword}{k} \equivrel{N-1} \pair{\aword'}{k'}$.  
\end{itemize}
\item If $\length{\aword_1}+N\cdot \length{\awordbis} \leq i < \length{\aword_1}+(M-N)\cdot \length{\awordbis}$ 
      \commzone{$i$ is Zone C} then $\length{\aword_1}+N\cdot \length{\awordbis} \leq i' < 
      \length{\aword_1}+(M'-N)\cdot \length{\awordbis}$ \commzone{$i'$ is in Zone C} and $|i-i'|=0 \mod 
      \length{\awordbis}$. Let $\ell=(i-(\length{\aword_1}+N\cdot \length{\awordbis})) \mod \length{\awordbis}$ 
      (the relative position of $i$ in the word $\awordbis$). We have the following 
      possibilities for the position  $j \geq i$:
\begin{itemize}
\item If $j-i < \length{\awordbis}-\ell+\length{\awordbis}$ ($j$ is either in the same word 
$\awordbis$ as $i$ or in the next word $\awordbis$), then $j < \length{\aword_1}+(M-(N-1))\cdot \length{\awordbis}$ \commzone{$j$ 
is in Zone C or D}. We define $j'=i'+(j-i)$ and we have that 
$\length{\aword_1}+N\cdot \length{\awordbis} \leq j' < \length{\aword_1}+(M'-(N-1))\cdot \length{\awordbis}$ \commzone{$j'$ is in 
Zone C or D} and since $|i-i'|=0 \mod \length{\awordbis}$, we deduce  $|j-j'|=0 \mod \length{\awordbis}$. From this we 
obtain  $\pair{\aword}{j} \equivrel{N-1} \pair{\aword'}{j'}$. Let $k' \in 
\interval{i'}{j'-1}$ and $k=i+k'-i'$. We have then that 
$k \in \interval{i}{j-1}$ and $\length{\aword_1}+N\cdot \length{\awordbis} \leq k' < \length{\aword_1}+(M'-(N-1))\cdot \length{\awordbis}$ and $\length{\aword_1}+N\cdot \length{\awordbis} \leq k < \length{\aword_1}+(M-(N-1))\cdot \length{\awordbis}$. 
Since $|i-i'|=0 \mod \length{\awordbis}$, we also have $|k-k'|=0 \mod \length{\awordbis}$. 
Consequently $\pair{\aword}{k} \equivrel{N-1} \pair{\aword'}{k'}$.
\item If $j-i \geq \length{\awordbis}-\ell+\length{\awordbis}$ 
($j$ is neither in the same word $\awordbis$ as $i$ nor in the next word $\awordbis$) and 
$j  \geq \length{\aword_1}+(M-N)\cdot \length{\awordbis}$ \commzone{$j$ is in Zone E or  D}. 
Let $j'=j-(M-M')\cdot \length{\awordbis}$
then $j' \geq \length{\aword_1}+(M'-N)\cdot \length{\awordbis}$ \commzone{$j'$ is in Zone E or  D} 
and consequently $\pair{\aword}{j} \equivrel{N} \pair{\aword'}{j'}$ and 
by (Claim 1)  we get $\pair{\aword}{j} \equivrel{N-1} \pair{\aword'}{j'}$. 
Then let $k' \in \interval{i'}{j'-1}$. If $k'\geq \length{\aword_1}+(M'-N)\cdot \length{\awordbis}$ \commzone{$k'$ is in Zone D or E}, then let $k=k'+(M-M')\cdot \length{\awordbis}$; we have in this case that 
$k  \geq \length{\aword_1}+(M-N)\cdot \length{\awordbis}$ and this allows us to deduce that 
$\pair{\aword}{k} \equivrel{N-1} \pair{\aword'}{k'}$. Now assume $k'< \length{\aword_1}+(M'-N)\cdot \length{\awordbis}$ \commzone{$k'$ is in Zone C} and $k'-i' < \length{\awordbis} -\ell$ ($k'$ and $i'$ are in the same word $\awordbis$), then 
let $k=i+k'-i'$. In this case we have $k < \length{\aword_1}+(M-N)\cdot \length{\awordbis}$ \commzone{$k$ is in Zone C} and since $|i-i'|=0 \mod \length{\awordbis}$, we also have $|k-k'|=0 \mod \length{\awordbis}$, whence 
$\pair{\aword}{k} \equivrel{N-1} \pair{\aword'}{k'}$. 
Now assume $k'< \length{\aword_1}+(M'-N)\cdot \length{\awordbis}$ \commzone{$k'$ is in Zone C} and $k'-i' \geq \length{\awordbis} -\ell$  ($k'$ and $i'$ are not in the same word $\awordbis$). 
We denote by $\ell'=(k'-(\length{\aword_1}+N\cdot \length{\awordbis})) \mod \length{\awordbis}$ the relative 
position of $k'$ in $\awordbis$ and let $k=i+(\length{\awordbis}-\ell)+\ell'$ 
($k$ and $k'$ occur in the same position in $\awordbis$ but $k$ occurs in the word  $\awordbis$
just after the word $\awordbis$ in which $i$ belongs to)
Then $k \in \interval{i}{j-1}$ (because $\ell' < \length{\awordbis}$ and $j-i \geq \length{\awordbis}-\ell+\length{\awordbis}$) and $k < \length{\aword_1}+(M-(N-1))\cdot \length{\awordbis}$ (because $i+(\length{\awordbis}-\ell) <\length{\aword_1}+(M-N)\cdot \length{\awordbis}$ and $\ell' < \length{\awordbis}$) and $|k-k'|=0 \mod \length{u}$ ($k$ and $k'$ are both pointing on the $\ell'$-th position in word $\awordbis$). This allows us to deduce that $\pair{\aword}{k} \equivrel{N-1} \pair{\aword'}{k'}$.
\item If $j-i \geq \length{\awordbis}-\ell+\length{\awordbis}$ 
($j$ is neither in the same word $\awordbis$ as $i$ nor in the next word $\awordbis$) and $j < \length{\aword_1}+(M-N)\cdot \length{\awordbis}$ \commzone{$j$ is in Zone C}. 
Then let $\ell' = (j-(\length{\aword_1} + N \cdot \length{\awordbis})) \mod \length{\awordbis}$ the relative position 
of $j$ in  $\awordbis$. We choose $j'=i'+(\length{\awordbis}-\ell)+\ell'$ 
($j$ and $j'$ occur in the same position in $\awordbis$ but $j'$ occurs in the word  $\awordbis$
just after the word $\awordbis$ in which $i'$ belongs to)
We have then that $j' < \length{\aword_1}+(M'-(N-1))\cdot \length{\awordbis}$ \commzone{$j'$ is in Zone C or D} (because $i'+(\length{\awordbis}-\ell) <\length{\aword_1}+(M-N)\cdot \length{\awordbis}$ and $\ell' < \length{\awordbis}$) and $|j-j'|=0 \mod \length{\awordbis}$ ($j$ and $j'$ are both pointing on the $\ell'$-th position in word $\awordbis$), hence 
$\pair{\aword}{j} \equivrel{N-1} \pair{\aword'}{j'}$. Let $k' \in \interval{i'}{j'-1}$. If 
$k'-i' < \length{\awordbis} -\ell$ ($k'$ and $i'$ are in the same word $\awordbis$), then let $k=i+k'-i'$. 
In this case we have $k < \length{\aword_1}+(M-N)\cdot \length{\awordbis}$ \commzone{$k$ is in Zone C} and 
since $|i-i'|=0 \mod \length{\awordbis}$, we also have $|k-k'|=0 \mod \length{\awordbis}$, hence 
$\pair{\aword}{k} \equivrel{N-1} \pair{\aword'}{k'}$. If $k'-i' \geq \length{\awordbis} -\ell$ ($k'$ and $i'$ are not in the same word $\awordbis$), then $j'-k' < \ell'$ and let $k=j-j'-k'$. In this case we have $k < \length{\aword_1}+(M-N)\cdot \length{\awordbis}$ \commzone{$k$ is in Zone C} and since $|j-j'|=0 \mod \length{\awordbis}$, we also have $|k-k'|=0 \mod \length{\awordbis}$, hence 
 $\pair{\aword}{k} \equivrel{N-1} \pair{\aword'}{k'}$.
\end{itemize}
\end{itemize}
\end{proof}

\subsection{Proof of (Claim 5)} 
The proof proceeds in a similar way as the proof of Claim 4. For completeness shake the proof is provided in the Appendix.
\cut{Before the proof, let us recall what is (Claim 5). 
Let 
$\aword= \aword_1 \awordbis^M \aword_2, \aword'=\aword_1 \awordbis^{M'} \aword_2 \in \aalphabet^{\omega}$, 
$i,i' \in \Nat$ and $N \geq 2$ such that  $M,M' \geq 2N+1$ and $\pair{\aword}{i} \equivrel{N} \pair{\aword'}{i'}$.
\begin{description}
\itemsep 0 cm
\item[(Claim 4)] for all $j \leq i$, there is $j' \leq i'$ 
      such that  $\pair{\aword}{j} \equivrel{N-1} \pair{\aword'}{j'}$ and 
       for all $k' \in \interval{j'-1}{i'}$, there is
            $k \in  \interval{j-1}{i}$ such that $\pair{\aword}{k} \equivrel{N-1} \pair{\aword'}{k'}$.
\end{description}

\begin{proof}
The proof is similar to the proof for (Claim 4) by looking backward instead of looking forward
(still there are slight differences because past is finite). Nevertheless, full proof is provided below
for the sake of completeness. We proceed by a case analysis on the positions  $i$ and $j$.
Without any loss of generality, we can assume that $M \geq M'$.
\begin{itemize}
\itemsep 0 cm 
\item If $i < \length{\aword_1}+N\cdot \length{\awordbis}$ \commzone{$i$ is in Zone A or B} then $j< \length{\aword_1}+N\cdot \length{\awordbis}$ \commzone{$j$ is in Zone A or B} and $i' < \length{\aword_1} + N\cdot \length{\awordbis}$ 
\commzone{$i'$ is in Zone A or B} and $i=i'$. We define $j'=j$. Then it is clear that $j' < i'$ and
 $\pair{\aword}{j} \equivrel{N} \pair{\aword'}{j'}$. By (Claim 1), we get 
$\pair{\aword}{j} \equivrel{N-1} \pair{\aword'}{j'}$. 
Let $k' \in \interval{j'-1}{i'}$ and let $k=k'$, then we have that 
$k \in \interval{j-1}{i}$ and also $\pair{\aword}{k} \equivrel{N} \pair{\aword'}{k'}$, 
hence by  (Claim 1), $\pair{\aword}{k} \equivrel{N-1} \pair{\aword'}{k'}$.
%

%
\item If $i \geq \length{\aword_1}+(M-N)\cdot \length{\awordbis}$ \commzone{$i$ is Zone D or E} then $i'\geq \length{\aword_1}+(M'-N)\cdot \length{\awordbis}$ \commzone{$i'$ is in Zone D or E} 
and $i=i'+(M-M')\cdot \length{\awordbis}$ and we have the following possibilities for the position  $j \leq i$:
\begin{itemize}
\itemsep 0 cm 
\item If $j \geq \length{\aword_1}+(M-N) \cdot \length{\awordbis}$ \commzone{$j$ is in Zone D or E}, then let $j'=j-(M-M')\cdot \length{\awordbis}$. Consequently, we have $\pair{\aword}{j} \equivrel{N} \pair{\aword'}{j'}$ and 
by (Claim 1) we get  $\pair{\aword}{j} \equivrel{N-1} \pair{\aword'}{j'}$. 
Let $k' \in \interval{j'-1}{i'}$ and $k=k'+(M-M')\cdot \length{\awordbis}$. Then we have that 
$k \in \interval{j-1}{i}$ and also 
 $\pair{\aword}{k} \equivrel{N} \pair{\aword'}{k'}$. By (Claim 1), 
 $\pair{\aword}{k} \equivrel{N-1} \pair{\aword'}{k'}$.
\item If $\length{\aword_1}+N \cdot \length{\awordbis} \leq j < \length{\aword_1}+(M-N) \cdot \length{\awordbis}$ \commzone{$j$ is in Zone C}, then let $\ell=(j-(\length{\aword_1}+N \cdot \length{\awordbis})) \mod \length\awordbis$ ($\ell$ is the relative position of $j$ in the word $\awordbis$ it belongs to). Consequently $0 \leq \ell < \length{\awordbis}$. 
Let $j'= \length{\aword_1} + (M'-N) \cdot \length{\awordbis}- (\length{\awordbis}-\ell)$ 
($j'$ is at the same position as $j$ in the last word $\awordbis$ of the Zone C). Then $\length{\aword_1}+N \cdot \length{\awordbis} \leq j' < \length{\aword_1}+(M'-N) \cdot \length{\awordbis}$ \commzone{$j'$ is in Zone C} 
(because $(M'\geq 2N+1$) and $|j - j'| = 0 \mod \length{\awordbis}$ (they are at the same position in the word $\awordbis$). We deduce that $\pair{\aword}{j} \equivrel{N} \pair{\aword'}{j'}$ and by (Claim 1) we get 
$\pair{\aword}{j} \equivrel{N-1} \pair{\aword'}{j'}$.
Then let $k' \in \interval{j'-1}{i'}$ and let $k=k'+(M-M') \cdot \length{\awordbis}$. Then we have that 
$k \in \interval{j-1}{i}$. Furthermore, if $k' \geq \length{\aword_1}+(M'-N) \cdot \length{\awordbis}$ 
\commzone{$k'$ is in Zone D or E} then $k \geq \length{\aword_1}+(M-N) \cdot \length{\awordbis}$ 
\commzone{$k$ is in Zone D or E} and we obtain $\pair{\aword}{k} \equivrel{N} \pair{\aword'}{k'}$ and by (Claim 1), 
$\pair{\aword}{k} \equivrel{N-1} \pair{\aword'}{k'}$. 
Moreover, if $k'< \length{\aword_1}+(M'-N) \cdot \length{\awordbis}$  then necessarily 
$\length{\aword_1}+N \cdot \length{\awordbis} \leq k'$ \commzone{$k'$ is in Zone C} (because $j'<k'$) 
and  $|k - k'| = 0 \mod \length{\awordbis}$ (because $k=k'+(M-M') \cdot \length{\awordbis}$). 
Whence, $k$ is in Zone C and
$\pair{\aword}{k} \equivrel{N} \pair{\aword'}{k'}$. By (Claim 1), we obtain 
 $\pair{\aword}{k} \equivrel{N-1} \pair{\aword'}{k'}$.
\item If $j< \length{\aword_1}+N \cdot \length{\awordbis}$ \commzone{$j$ is in Zone A or B}, let $j'=j$. 
We have then $j' < \length{\aword_1}+N\cdot \length{\awordbis}$ \commzone{$j'$ is in Zone A or B}. 
We deduce 
that $\pair{\aword}{j} \equivrel{N} \pair{\aword'}{j'}$ and by 
(Claim 1) we get  $\pair{\aword}{j} \equivrel{N-1} \pair{\aword'}{j'}$. 
Then let $k' \in \interval{j'-1}{i'}$. If $k' < \length{\aword_1}+N \cdot \length{\awordbis}$ \commzone{$k'$ is in Zone A}, 
for $k=k'$, we obtain $\pair{\aword}{k} \equivrel{N} \pair{\aword'}{k'}$ and by (Claim 1), 
$\pair{\aword}{k} \equivrel{N-1} \pair{\aword'}{k'}$. 
If $k' \geq \length{\aword_1}+(M'-N) \cdot \length{\awordbis}$ \commzone{$k'$ is in Zone D or E}, we choose 
$k=k'+(M-M') \cdot \length{\awordbis}$ and here also we deduce $\pair{\aword}{k}  \equivrel{N}  \pair{\aword'}{k'}$ 
and by (Claim 1),  $\pair{\aword}{k}  \equivrel{N-1}  \pair{\aword'}{k'}$. 
If ${\aword_1}+N \cdot \length{\awordbis} \leq k' < \length{\aword_1}+(M'-N) \cdot \length{\awordbis}$ 
\commzone{$k'$ is in Zone C}, let $\ell=(k'-(\length{\aword_1}+N \cdot \length{\awordbis})) \mod 
\length{\awordbis}$ ($\ell$ is the relative position of $k'$ in the word $\awordbis$ it belongs to) and 
let $k=\length{\aword_1}+N \cdot \length{\awordbis}+\ell$ ($k$ is at the same position of $k'$ in the first 
word of the zone C). Then we have ${\aword_1}+N \cdot \length{\awordbis} \leq k < \length{\aword_1}+(M-N) 
\cdot \length{\awordbis}$ \commzone{$k$ is in the Zone C} and $|k-k'|=0 \mod \length{\awordbis}$ which allows to 
deduce that $\pair{\aword}{k} \equivrel{N} \pair{\aword'}{k'}$ and by (Claim 1), 
$\pair{\aword}{k} \equivrel{N-1} \pair{\aword'}{k'}$.
\end{itemize}
\item If $\length{\aword_1}+N\cdot \length{\awordbis} \leq i < \length{\aword_1}+(M-N)\cdot \length{\awordbis}$ 
\commzone{$i$ in Zone C} then $\length{\aword_1}+N\cdot \length{\awordbis} \leq i' < \length{\aword_1}+(M'-N)\cdot \length{\awordbis}$ 
\commzone{$i'$ in Zone C} and $|i-i'|=0 \mod \length{\awordbis}$. Let $\ell=(i-(\length{\aword_1}+N\cdot \length{\awordbis})) 
\mod \length{\awordbis}$ (the relation position of $i$ in the word $\awordbis$ it belongs to). We 
have the following possibilities for the position $j \leq i$:
\begin{itemize}
\itemsep 0 cm 
\item If $i-j <\ell+\length{\awordbis}$ ($j$ is  in the same word $\awordbis$ as $i$ or in the previous word 
$\awordbis$) then $j \geq \length{\aword_1}+(N-1)\cdot \length{\awordbis}$ \commzone{$j$ is in Zone B or C}. We define 
$j'=i'-(i-j)$ and we have that $\length{\aword_1}+(N-1)\cdot \length{\awordbis} \leq j' < \length{\aword_1}+(M'-N)
\cdot \length{\awordbis}$ \commzone{$j'$ is in Zone B or C} and since $|i-i'|=0 \mod \length{\awordbis}$, 
we deduce  $|j-j'|=0 \mod \length{\awordbis}$. From this, we obtain  $\pair{\aword}{j} 
\equivrel{N-1} \pair{\aword'}{j'}$. Let $k' \in \interval{j'-1}{i'}$ and $k=i-(i'-k')$. 
We have then that $k \in \interval{j-1}{i}$ and $\length{\aword_1}+(N-1)\cdot \length{\awordbis} \leq k' < \length{\aword_1}+(M'-N)\cdot \length{\awordbis}$ \commzone{$k'$ is in Zone B or C} and $\length{\aword_1}+(N-1)\cdot \length{\awordbis} \leq k < \length{\aword_1}+(M-N)\cdot \length{\awordbis}$ \commzone{$k$ is in Zone B or C} and since $|i-i'|=0 \mod \length{\awordbis}$, we also have $|k-k'|=0 \mod \length{\awordbis}$. Consequently $\pair{\aword}{k} \equivrel{N-1} \pair{\aword'}{k'}$.
\item If $i-j \geq \ell+\length{\awordbis}$ ($j$ is neither in the same word $\awordbis$ as $i$ nor in the 
previous word $\awordbis$) and $j  < \length{\aword_1}+N\cdot \length{\awordbis}$ \commzone{$j$ is in zone A or B}. Let 
$j'=j$. So, $j' < \length{\aword_1}+N\cdot \length{\awordbis}$ and $\pair{\aword}{j} \equivrel{N} \pair{\aword'}{j'}$. 
By using (Claim 1) 
we get  $\pair{\aword}{j} \equivrel{N-1} \pair{\aword'}{j'}$. 
Then let $k' \in \interval{j'-1}{i'}$. If $k'< \length{\aword_1}+N\cdot \length{\awordbis}$ 
\commzone{$k'$ is in Zone A or B}, then let $k=k'$; we have in this case that 
$k  < \length{\aword_1}+N\cdot \length{\awordbis}$ and this allows us to deduce that 
$\pair{\aword}{k} \equivrel{N-1} \pair{\aword'}{k'}$. Now assume $k'\geq  \length{\aword_1}+N\cdot \length{\awordbis}$ 
\commzone{$k'$ is in Zone C} and $i'-k' \leq \ell$ ($k'$ and $i'$ are in the same word $\awordbis$), then let 
$k=i-(i'-k')$. In this case we have $k \geq \length{\aword_1}+N\cdot \length{\awordbis}$ \commzone{$k$ is in Zone C} and 
since $|i-i'|=0 \mod \length{\awordbis}$, we also have $|k-k'|=0 \mod \length{\awordbis}$, hence 
$(\aword,k) \equivrel{N-1} (\aword',k')$. Now assume $k'\geq \length{\aword_1}+N\cdot \length{\awordbis}$ 
\commzone{$k'$ is in Zone C} and $i'-k' > \ell$ ($k'$ and $i'$ are not in the same word $\awordbis$). We denote by $\ell'=(k'-(\length{\aword_1}+N\cdot \length{\awordbis})) \mod \length{\awordbis}$ the relation position of $k'$ in $\awordbis$ and let $k=i-\ell-(\length{\awordbis}-\ell')$ ($k$ is at the same position as $k'$ of $k$ in the word $\awordbis$  preceding the word $\awordbis$ $i$ belongs to). Then $k \in \interval{j-1}{i}$ 
(because $\length{\awordbis}-\ell'< \length{\awordbis}$ and $i-j \geq \ell+\length{\awordbis}$) and $k \geq \length{\aword_1}+(N-1)\cdot \length{\awordbis}$ (because $i+(\length{\awordbis}-\ell) \geq \length{\aword_1}+(M-N)\cdot \length{\awordbis}$ and $\length{\awordbis}-\ell < \length{\awordbis}$) and $|k-k'|=0 \mod \length{u}$ ($k$ and $k'$ are both pointing on the $\ell'$-th position in word $\awordbis$). This allows us to deduce that 
$\pair{\aword}{k} \equivrel{N-1} \pair{\aword'}{k'}$.
\item If $j-i \geq \ell+\length{\awordbis}$ ($j$ is neither in the same word $\awordbis$ as $i$ 
nor in the previous word $\awordbis$) 
and $j \geq \length{\aword_1}+N\cdot \length{\awordbis}$ \commzone{$j$ is in zone C}. 
Then let $\ell' = (j-(\length{\aword_1} + N \cdot \length{\awordbis})) \mod \length{\awordbis}$ the relative 
position of $j \in \awordbis$. We choose $j'=i'-\ell-(\length{\awordbis}-\ell')$ ($j'$ 
and $j$ are on the same position of $\awordbis$ but in the word $\awordbis$ 
precedent in the one to which $i$ belongs to). We have then that 
$j' \geq \length{\aword_1}+(N-1)\cdot \length{\awordbis}$ \commzone{$j'$ is zone B or C} (because 
$i'-\ell \geq \length{\aword_1}+N)\cdot \length{\awordbis}$ and $\length{\awordbis}-\ell' \leq \length{\awordbis}$) and 
$|j-j'|=0 \mod \length{\awordbis}$ ($j$ and $j'$ are both pointing on the $\ell'$-th position in word $\awordbis$), 
hence $\pair{\aword}{j} \equivrel{N-1} \pair{\aword'}{j'}$. Let $k' \in \interval{j'-1}{i'}$. 
If $i'-k' \leq \ell$ ($k'$ and $i'$ are in the same word $\awordbis$), then let $k=i-(i'-k')$. In this case we have $k \geq \length{\aword_1}+N\cdot \length{\awordbis}$ \commzone{$k$ is in Zone C} and since $|i-i'|=0 \mod \length{\awordbis}$, we also have 
$|k-k'|=0 \mod \length{\awordbis}$, hence $\pair{\aword}{k} \equivrel{N-1} \pair{\aword'}{k'}$. 
If $i'-k' > \ell$ ($k'$ and $i'$ are not in the same word $\awordbis$), 
then $k'-j' <\length{\awordbis}- \ell'$ and let $k=j+k'-j'$. In this case we have 
$\length{\aword_1}+N\cdot \length{\awordbis} \leq k < \length{\aword_1}+(M-N)\cdot \length{\awordbis}$ \commzone{$k$ is in Zone C} and 
since $|j-j'|=0 \mod \length{\awordbis}$, we also have $|k-k'|=0 \mod \length{\awordbis}$, whence
 $\pair{\aword}{k} \equivrel{N-1} \pair{\aword'}{k'}$.
\end{itemize}
\end{itemize}
\end{proof}

}

\cut{
\begin{lemma}\label{lemma-word}
Given $\aword,\aword' \in \edges^\omega$, a natural $N \geq 2$ and two positions $i,i' \in \nat$, if $(\aword,i) \equivrel{N} (\aword',i')$ then $(\aword, i) \equivrel{N-1} (\aword', i')$ and $\aword(i) = \aword'(i')$.
\end{lemma}
\ifshort
\else
\begin{proof}
Assume $(\aword,i) \equivrel{N} (\aword',i')$ then let $\awordbis \in \edges^+$, $\aword_1 \in \edges^\ast$, $\aword_2 \in \edges^\omega$ and $M,M' \in \nat$ as in the definition of the equivalence relation $\equivrel{N}$. 

First we will prove that $(\aword,i) \equivrel{N-1} (\aword',i')$. Since $N > N-1$, it is obvious that $M \equivrel{2(N-1)} M'$. Furthermore if $M=M'$ then $i=i'$ and consequently $(\aword,i) \equivrel{N-1} (\aword',i')$. If $M \neq M'$ with $M >M'$ (the case $M>M'$ can be treated similarly), we proceed by a case analysis on the position of $i$.:
\begin{itemize}
\item If $i< \length{\aword_1}+(N-1) \length{\awordbis}$ \commzone{$i$ is Zone A} then $i,i' < \length{\aword_1}+N\length{\awordbis}$ and  $i=i'$.
\item If $i \geq \length{\aword_1} +  (M-(N-1))\length{\awordbis}$ \commzone{$i$ is in zone E} then $i \geq \length{\aword_1} +  (M-N)\length{\awordbis}$ and $i' \geq \length{\aword_1} +  (M'-N)\length{\awordbis}$ \commzone{$i,i'$ are in Zone D or E} and $i=i'+(M-M')\length{\awordbis}$, hence we also have $i' \geq \length{\aword_1} +  (M'-(N-1))\length{\awordbis}$.
\item If $\length{\aword_1} + (N-1) \length{\awordbis} \leq i < \length{\aword_1} +  N\length{\awordbis}$ \commzone{$i$ is in Zone B} then $ i' < \length{\aword_1} +  N\length{\awordbis}$ and $i=i'$. So $\length{\aword_1} + (N-1) \length{\awordbis} \leq i' < \length{\aword_1} +  N\length{\awordbis}$ \commzone{$i'$ is also in Zone B} and $|i - i'| = 0 \mod \length{\awordbis}$.
\item If $\length{\aword_1} + N \length{\awordbis} \leq i < \length{\aword_1} +  (M-N)\length{\awordbis}$ \commzone{$i$ in Zone C}, then $\length{\aword_1} + N \length{\awordbis} \leq i' < \length{\aword_1} +  (M'-N)\length{\awordbis}$ \commzone{$i'$ is in Zone C} and $|i - i'| = 0 \mod \length{\awordbis}$. Consequently $\length{\aword_1} + N \length{\awordbis} \leq i' < \length{\aword_1} +  (M'-(N-1))\length{\awordbis}$ and  $|i - i'| = 0 \mod \length{\awordbis}$.
\item If $\length{\aword_1} + (M-N) \length{\awordbis} \leq i < \length{\aword_1} +  (M-(N-1))\length{\awordbis}$ \commzone{$i$ in Zone D}, then $\length{\aword_1} + (M'-N) \length{\awordbis} \leq i'$ and $i=i'+(M-M')\length{\awordbis}$.  Consequently $ i' < \length{\aword_1} +  (M'-(N-1))\length{\awordbis}$ \commzone{$i'$ is also in Zone E} and  $|i - i'| = 0 \mod \length{\awordbis}$.
\end{itemize}
From these different points, we deduce that $(\aword,i) \equivrel{N-1} (\aword',i')$.

It is also clear that $\aword(i) = \aword'(i')$, because either $i$ and $i'$ are at the same position in the word $\aword_1$ or $\aword_2$ either they are pointing some positions in the portions of the word which belong to $\awordbis^+$ and since their difference will be such that $|i - i'| = 0 \mod \length{\awordbis}$, it easy to see that $i$ and $i'$ will point at the same position in  $\awordbis$.
\end{proof}

\fi
\begin{lemma}\label{lemma-word-bis}
Given $\aword,\aword' \in \edges^\omega$, a natural $N \geq 2$ and two positions $i,i' \in \nat$, if $(\aword,i) \equivrel{N} (\aword',i')$ then the following properties hold:
\begin{itemize} 
\item[\textbf(i)] $(\aword,i+1) \equivrel{N-1} (\aword', i'+1)$.
\item[\textbf(ii)]  if $i,i'>0$ then $(\aword, i-1) \equivrel{N-1} (\aword', i'-1)$.
\item[\textbf(iii)] for all $j \geq i$, there is $j' \geq i'$ such that:
\begin{enumerate}
\item  $(\aword, j) \equivrel{N-1} (\aword',j')$,
\item and for all $k' \in \set{i',\ldots,j'-1}$, there is
            $k \in  \set{i,\ldots,j-1}$ such that 
             $(\aword, k) \equivrel{N-1} (\aword', k')$.
      \end{enumerate}
\item[\textbf(iv)] for all $j \leq i$, there
     is $j' \leq i'$ 
      such that:
      \begin{enumerate}
      \item  $(\aword, j) \equivrel{N-1}{} (\aword', j')$,
      \item and for all $k' \in \set{j'-1,\ldots,i'}$, there is
            $k \in  \set{j-1,\ldots,i}$ such that 
             $(\aword, k) \equivrel{N-1}{} (\aword', k')$.
      \end{enumerate}
\end{itemize}
\end{lemma}
\ifshort
\else
\begin{proof}
Assume $(\aword,i) \equivrel{N} (\aword',i')$ then let $\awordbis \in \edges^+$, $\aword_1 \in \edges^\ast$, $\aword_2 \in \edges^\omega$ and $M,M' \in \nat$ as in the definition of the equivalence relation $\equivrel{N}$. We will now prove the different properties. First note that since $N > N-1$, it is obvious that $M \equivrel{2(N-1)} M'$\\

\textbf{Property (i).} If $M=M'$ then $i=i'$ and $i+1=i'+1$ consequently $(\aword,i+1) \equivrel{N-1} (\aword',i'+1)$. If $M \neq M'$ with $M >M'$ (the case $M'>M$ can be treated similarly) we proceed by a case analysis on the position of $i$:
\begin{itemize}
\item If $i< \length{\aword_1}+(N-1) \length{\awordbis}$ \commzone{$i$ is in Zone A} then $i,i' < \length{\aword_1}+N\length{\awordbis}$ and $i=i'$ and hence $i'< \length{\aword_1}+(N-1) \length{\awordbis}$ \commzone{$i'$ is also in Zone A}. We can deduce that $i+1=i'+1$ and either $i'+1,i+1 < \length{\aword_1}+(N-1)\length{\awordbis}$ \commzone{$i+1$ and $i'+1$ are in zone A} or, since $\length{\awordbis}>0$, $\length{\aword_1}+(N-1)\length{\awordbis} \leq i'+1,i+1 < \length{\aword_1}+N \length{\awordbis}$ \commzone{$i+i$ and $i'+1$ are in zone B} and in the latter case we have effectively $|(i+1) - (i'+1)| = 0 \mod \length{\awordbis}$. Hence $(\aword,i+1) \equivrel{N-1} (\aword', i'+1)$.
\item If $i \geq \length{\aword_1} +  (M-N)\length{\awordbis}$ \commzone{$i$ is in Zone D or E} then $i' \geq \length{\aword_1} +  (M'-N)\length{\awordbis}$ \commzone{$i'$ is in zone D or E} and $i=i'+(M-M')\length{\awordbis}$. Consequently we have $i+1 \geq  \length{\aword_1} +  (M-N)\length{\awordbis}$ and $i'+1 \geq \length{\aword_1} +  (M'-N)\length{\awordbis}$ \commzone{$i+1$ and $i'+1$ are also in Zone D or E} and $i+1=i'+1+(M-M')\length{\awordbis}$. From this we deduce that $(\aword,i+1) \equivrel{N} (\aword', i'+1)$ and using Lemma \ref{lemma-word}, we obtain $(\aword,i+1) \equivrel{N-1} (\aword', i'+1)$.
\item If $\length{\aword_1}+N \length{\awordbis} \leq i < \length{\aword_1}+(M-N) \length{\awordbis}$ \commzone{$i$ is in Zone C}, then  $\length{\aword_1}+N \length{\awordbis} \leq i' <\length{\aword_1}+(M'-N) \length{\awordbis}$ \commzone{$i'$ is also in Zone C} and $|i - i'| = 0 \mod \length{\awordbis}$. Since $\length{\awordbis} >0$ we deduce that $\length{\aword_1}+(N-1) \length{\awordbis} \leq i+1 < \length{\aword_1}+(M-(N-1)) \length{\awordbis}$ and  $\length{\aword_1}+(N-1) \length{\awordbis} \leq i'+1 < \length{\aword_1}+(M'-(N-1)) \length{\awordbis}$ \commzone{$i+1$ and $i'+1$ are in Zone C or D} and $|(i+1) - (i'+1)| = 0 \mod \length{\awordbis}$, hence $(\aword,i+1) \equivrel{N-1} (\aword', i'+1)$.
\item If $\length{\aword_1}+(N-1) \length{\awordbis} \leq i < \length{\aword_1}+N \length{\awordbis}$ \commzone{$i$ is in Zone B}, then  $i,i' < \length{\aword_1}+ N \length{\awordbis}$ and $i=i'$ \commzone{$i$ and $i'$ are in Zone B}. Furthermore since $M \neq M'$ and $M \equivrel{2N} M'$, we have $M>2N$ and $M'>2N$, hence $(M-N) > N$ and $(M'-N) > N$. Using that $\length{\awordbis}>0$, this allows us to deduce that  $i+1 < \length{\aword_1}+(M-N)) \length{\awordbis}$ and $i'+1 < \length{\aword_1}+(M'-N-1) \length{\awordbis}$ \commzone{$i+1$ and $i'+1$ are in zone B or C} and since $|(i+1) - (i'+1)| = 0 \mod \length{\awordbis}$, we obtain $(\aword,i+1) \equivrel{N-1} (\aword', i'+1)$.
\end{itemize}

\textbf{Property (ii).} If $M=M'$ then $i=i'$ and $i-1=i'-1$ consequently $(\aword,i-1) \equivrel{N-1} (\aword',i'-1)$. If $M \neq M'$ with $M >M'$ (the case $M'>M$ can be treated similarly) we proceed by a case analysis on the position of $i$:
\begin{itemize}
\item If $i< \length{\aword_1}+N \length{\awordbis}$ \commzone{$i$ is in Zone A or B} then $i,i' < \length{\aword_1}+N\length{\awordbis}$ \commzone{$i$ and $i'$ are in Zone A or B} and  $i=i'$. Hence $i-1,i'-1 < \length{\aword_1}+N\length{\awordbis}$ \commzone{$i-1$ and $i'-1$ are in Zone A or B}.From this we deduce that $(\aword,i-1) \equivrel{N} (\aword', i'-1)$ and using Lemma \ref{lemma-word}, we obtain $(\aword,i-1) \equivrel{N-1} (\aword', i'-1)$.
\item If $i \geq \length{\aword_1} +  (M-(N-1))\length{\awordbis}$ \commzone{$i$ is in Zone E} then  $i=i'+(M-M')\length{\awordbis}$ and $i' \geq \length{\aword_1} +  (M'-(N-1))\length{\awordbis}$. We can deduce that $i-1=i'-1+(M-M')\length{\awordbis}$ and either ($i-1 \geq \length{\aword_1} +  (M-(N-1))\length{\awordbis}$  and $i'-1 \geq \length{\aword_1} +  (M-(N-1))\length{\awordbis}$) \commzone{$i-1$ and $i'-1$ are in Zone E} or, since $\length{\awordbis}>0$, ( $\length{\aword_1} +  (M-N)\length{\awordbis}\leq i-1 < \length{\aword_1} +  (M-(N-1))\length{\awordbis}$ and $\length{\aword_1} +  (M'-N)\length{\awordbis}\leq i'-1 < \length{\aword_1} +  (M'-(N-1))\length{\awordbis}$) \commzone{$i-1$ and $i'-1$ are in Zone D} and in the latter case, we have effectively $|(i-1)-(i'-1)|=0\mod \length{\awordbis}$. Hence $(\aword,i-1) \equivrel{N-1} (\aword', i'-1)$.
\item If $\length{\aword_1}+N \length{\awordbis} \leq i < \length{\aword_1}+(M-N) \length{\awordbis}$ \commzone{$i$ is in Zone C}, then  $\length{\aword_1}+N \length{\awordbis} \leq i' <\length{\aword_1}+(M'-N) \length{\awordbis}$ \commzone{$i'$ is in Zone C} and $|i - i'| = 0 \mod \length{\awordbis}$. Since $\length{\awordbis} >0$ we deduce that $\length{\aword_1}+(N-1) \length{\awordbis} \leq i-1 < \length{\aword_1}+(M-(N-1)) \length{\awordbis}$ and $\length{\aword_1}+(N-1) \length{\awordbis} \leq i'-1 < \length{\aword_1}+(M'-(N-1)) \length{\awordbis}$ \commzone{$i-1$ and $i'-1$ are in zone B or C (or D)} and $|(i-1) - (i'-1)| = 0 \mod \length{\awordbis}$, hence $(\aword,i-1) \equivrel{N-1} (\aword', i'-1)$.
\item If $\length{\aword_1}+(M-N) \length{\awordbis} \leq i < \length{\aword_1}+(M-(N-1)) \length{\awordbis}$ \commzone{$i$ in Zone D}, then  $ \length{\aword_1}+(M-N) \length{\awordbis} \leq i$ and $ \length{\aword_1}+(M'-N) \length{\awordbis} \leq i'$ and $i=i'+(M-M')\length{\awordbis}$. Furthermore since $M \neq M'$ and $M \equivrel{2N} M'$, we have $M>2N$ and $M'>2N$, hence $(M-N) > N$ and $(M'-N) > N$. Using that $\length{\awordbis}>0$, this allows us to deduce that  $ \length{\aword_1}+N \length{\awordbis}\leq i-1,i'-1$ \commzone{$i-1$ and $i'-1$ are in Zone C} and since $|(i+1) - (i'+1)| = 0 \mod \length{\awordbis}$, we obtain $(\aword,i-1) \equivrel{N-1} (\aword', i'-1)$.
\end{itemize}

\textbf{Property (iii).} If $M=M'$ then take for all $j \geq i$, $j'=j$, this will implies that $(\aword,j) \equivrel{N-1} (\aword'j')$ and for all $k' \in \set{i',\ldots,j'-1}$, take $k=k'$, hence we will have $k \in  \set{i,\ldots,j-1}$ and $(\aword, k) \equivrel{N-1} (\aword', k')$. Assume now $M \neq M'$ with $M >M'$ (the case $M'>M$ can be treated similarly). We proceed by a case analysis on the position  $i$ and $j$:
\begin{itemize}
\item If $i \geq \length{\aword_1}+(M-N)\length{\awordbis}$ \commzone{$i$ is in Zone D or E} then $j \geq \length{\aword_1}+(M-N)\length{\awordbis}$ \commzone{$j$ is in Zone D or E} and $i' \geq \length{\aword_1} +  (M'-N)\length{\awordbis}$ \commzone{$i'$ is zone D or E} and $i=i'+(M-M')\length{\awordbis}$. We define $j'=j-(M-M')\length{\awordbis}$. Then it is clear that $j' \geq i'$ and also that $(\aword,j) \equivrel{N} (\aword'j')$ and by Lemma \ref{lemma-word} we get $(\aword,j) \equivrel{N-1} (\aword'j')$. Let $k' \in \set{i',\ldots,j'-1}$ and let $k=k'+(M-M')\length{\awordbis}$, then we have that $k \in \set{i,\ldots,j-1}$ and also $(\aword,k) \equivrel{N} (\aword',k')$, hence by Lemma \ref{lemma-word}, $(\aword,k) \equivrel{N-1} (\aword',k')$.
\item If $i < \length{\aword_1}+N\length{\awordbis}$ \commzone{$i$ is in Zone A or B} then $i'< \length{\aword_1}+N\length{\awordbis}$ \commzone{$i'$ is in Zone A or B} and $i=i'$ and we have the following possibilities for the position of $j \geq i$:
\begin{itemize}
\item if $j < \length{\aword_1}+N \length{\awordbis}$ \commzone{$j$ is in Zone A or B}, then let $j'=j$. Consequently we have $(\aword,j) \equivrel{N} (\aword'j')$ and by Lemma \ref{lemma-word} we get $(\aword,j) \equivrel{N-1} (\aword'j')$. Let $k' \in \set{i',\ldots,j'-1}$ and $k=k'$. Then we have that $k \in \set{i,\ldots,j-1}$ and also $(\aword,k) \equivrel{N} (\aword',k')$ and by Lemma \ref{lemma-word}, $(\aword,k) \equivrel{N-1} (\aword',k')$.
\item if $\length{\aword_1}+N \length{\awordbis} \leq j < \length{\aword_1}+(M-N) \length{\awordbis}$ \commzone{$j$ is in Zone C}, then let $\ell=(j-(\length{\aword_1}+N \length{\awordbis})) \mod \length\awordbis$ ($\ell$ the relative position of $j$ in the word $\awordbis$ it belongs to). Consequently $0 \leq \ell < \length{\awordbis}$. Let $j'= \length{\aword_1} + N \length{\awordbis}+\ell$ (we chose $j'$ at the same relative position of $j$ in the first word $\awordbis$ of the zone C). Then $\length{\aword_1}+N \length{\awordbis} \leq j' < \length{\aword_1}+(M'-N) \length{\awordbis}$ \commzone{$j'$ is in Zone C} (because $(M'-N)>0$) and $|j - j'| = 0 \mod \length{\awordbis}$. We deduce that $(\aword,j) \equivrel{N} (\aword',j')$ and by Lemma \ref{lemma-word} we get $(\aword,j) \equivrel{N-1} (\aword'j')$.
Then let $k' \in \set{i',\ldots,j'-1}$ and let $k=k'$. Then we have that $k \in \set{i,\ldots,j-1}$. Furthermore if $k' < \length{\aword_1}+N \length{\awordbis}$ \commzone{$k'$ is in Zone A or B} we obtain $(\aword,k) \equivrel{N} (\aword',k')$ and by Lemma \ref{lemma-word}, $(\aword,k) \equivrel{N-1} (\aword',k')$. And if $\length{\aword_1}+N \length{\awordbis} \leq k'$ \commzone{$k'$ is in Zone C} then necessarily $k'< \length{\aword_1}+(M'-N) \length{\awordbis}$ (because $k'<j'$) and $k< \length{\aword_1}+(M-N) \length{\awordbis}$ (because $M>M'$) and also $|k - k'| = 0 \mod \length{\awordbis}$ (because $k=k'$). Hence $(\aword,k) \equivrel{N} (\aword', k')$ and Lemma \ref{lemma-word} tells us that $(\aword,k) \equivrel{N-1} (\aword', k')$.
\item if $ \length{\aword_1}+(M-N) \length{\awordbis} \leq j$ \commzone{$j$ is in Zone E}, let $j'=j-(M-M')\length{\awordbis}$. We have then $ \length{\aword_1}+(M'-N) \length{\awordbis} \leq j'$ \commzone{$j'$ is in Zone D or E} and we deduce that $(\aword,j) \equivrel{N} (\aword'j')$ and by Lemma \ref{lemma-word} we get $(\aword,j) \equivrel{N-1} (\aword'j')$. Then let $k' \in \set{i',\ldots,j'-1}$. If $k' < \length{\aword_1}+N \length{\awordbis}$ \commzone{$k'$ is in Zone A or B}, for $k=k'$, we obtain $(\aword,k) \equivrel{N} (\aword',k')$ and by Lemma \ref{lemma-word}, $(\aword,k) \equivrel{N-1} (\aword',k')$. If $k' \geq \length{\aword_1}+(M'-N) \length{\awordbis}$ \commzone{$k'$ is in Zone D or E}, we choose $k=k'+(M-M') \length{\awordbis}$ and here also we deduce $(\aword,k) \equivrel{N} (\aword',k')$ and by Lemma \ref{lemma-word}, $(\aword,k) \equivrel{N-1} (\aword',k')$. If ${\aword_1}+N \length{\awordbis} \leq k' < \length{\aword_1}+(M'-N) \length{\awordbis}$ \commzone{$k'$ is in Zone C}, let $\ell=(k'-(\length{\aword_1}+N \length{\awordbis})) \mod \length{\awordbis}$ ($\ell$ is the relative position of $k'$ in the word $\awordbis$ it belongs to) and let $k=\length{\aword_1}+N \length{\awordbis}+\ell$ ($k$ is placed at the same relative position  of $k'$ in the first word $\awordbis$ of the Zone C). Then we have ${\aword_1}+N \length{\awordbis} \leq k < \length{\aword_1}+(M-N) \length{\awordbis}$ and $|k-k'|=0 \mod \length{\awordbis}$ which allows to deduce that $(\aword,k) \equivrel{N} (\aword',k')$ and by Lemma \ref{lemma-word}, $(\aword,k) \equivrel{N-1} (\aword',k')$.
\end{itemize}
\item If $\length{\aword_1}+N\length{\awordbis} \leq i < \length{\aword_1}+(M-N)\length{\awordbis}$ \commzone{$i$ is zone C} then $\length{\aword_1}+N\length{\awordbis} \leq i' < \length{\aword_1}+(M'-N)\length{\awordbis}$ \commzone{$i'$ is in Zone C} and $|i-i'|=0 \mod \length{\awordbis}$. Let $\ell=(i-(\length{\aword_1}+N\length{\awordbis})) \mod \length{\awordbis}$ (the relative position of $i$ in the word $\awordbis$). We have the following possibilities for the position of $j \geq i$:
\begin{itemize}
\item If $j-i < \length{\awordbis}-\ell+\length{\awordbis}$ ($j$ is either in the same word $\awordbis$ as $i$ or in the next word $\awordbis$), then $j < \length{\aword_1}+(M-(N-1))\length{\awordbis}$ \commzone{$j$ is in Zone C or D}. We define $j'=i'+(j-i)$ and we have that $\length{\aword_1}+N\length{\awordbis} \leq j' < \length{\aword_1}+(M'-(N-1))\length{\awordbis}$ \commzone{$j'$ is in Zone C or D} and since $|i-i'|=0 \mod \length{\awordbis}$, we deduce  $|j-j'|=0 \mod \length{\awordbis}$. From this we obtain  $(\aword,j) \equivrel{N-1} (\aword',j')$. Let $k' \in \set{i',\ldots,j'-1}$ and $k=i+k'-i'$. We have then that $k \in \set{i,\ldots,j-1}$ and $\length{\aword_1}+N\length{\awordbis} \leq k' < \length{\aword_1}+(M'-(N-1))\length{\awordbis}$ and $\length{\aword_1}+N\length{\awordbis} \leq k < \length{\aword_1}+(M-(N-1))\length{\awordbis}$ and since $|i-i'|=0 \mod \length{\awordbis}$, we also have $|k-k'|=0 \mod \length{\awordbis}$. Consequently $(\aword,k) \equivrel{N-1} (\aword',k')$.
\item If $j-i \geq \length{\awordbis}-\ell+\length{\awordbis}$ ($j$ is not in the same word $\awordbis$ as $i$ neither in the next word $\awordbis$) and $j  \geq \length{\aword_1}+(M-N)\length{\awordbis}$ \commzone{$j$ is in Zone E}. Let $j'=j-(M-M')\length{\awordbis}$
then $j' \geq \length{\aword_1}+(M'-N)\length{\awordbis}$ \commzone{$j$ is in Zone E} and consequently $(\aword,j) \equivrel{N} (\aword',j')$ and using Lemma \ref{lemma-word} we get $(\aword,j) \equivrel{N-1} (\aword',j')$. Then let $k' \in \set{i',\ldots,j'-1}$. If $k'\geq \length{\aword_1}+(M'-N)\length{\awordbis}$ \commzone{$k'$ is in Zone D or E}, then let $k=k'+(M-M')\length{\awordbis}$; we have in this case that $k  \geq \length{\aword_1}+(M-N)\length{\awordbis}$ and this allows us to deduce that $(\aword,k) \equivrel{N-1} (\aword',k')$. Now assume $k'< \length{\aword_1}+(M'-N)\length{\awordbis}$ \commzone{$k'$ is in Zone C} and $k'-i' < \length{\awordbis} -\ell$ ($k'$ and $i'$ are in the same word $\awordbis$), then let $k=i+k'-i'$. In this case we have $k < \length{\aword_1}+(M-N)\length{\awordbis}$ \commzone{$k$ is in Zone C} and since $|i-i'|=0 \mod \length{\awordbis}$, we also have $|k-k'|=0 \mod \length{\awordbis}$, hence $(\aword,k) \equivrel{N-1} (\aword',k')$. Now assume $k'< \length{\aword_1}+(M'-N)\length{\awordbis}$ \commzone{$k'$ is in Zone C} and $k'-i' \geq \length{\awordbis} -\ell$  ($k'$ and $i'$ are not in the same word $\awordbis$). We denote by $\ell'=(k'-(\length{\aword_1}+N\length{\awordbis})) \mod \length{\awordbis}$ the relative position of $k'$ in $\awordbis$ and let $k=i+(\length{\awordbis}-\ell)+\ell'$ ($k$ is at the same position as $k'$ of $k$ in the word $\awordbis$ following the word $\awordbis$ $i$ belongs to). Then $k \in \set{i,\ldots,j-1}$ (because $\ell' < \length{\awordbis}$ and $j-i \geq \length{\awordbis}-\ell+\length{\awordbis}$) and $k < \length{\aword_1}+(M-(N-1))\length{\awordbis}$ (because $i+(\length{\awordbis}-\ell) <\length{\aword_1}+(M-N)\length{\awordbis}$ and $\ell' < \length{\awordbis}$) and $|k-k'|=0 \mod \length{u}$ ($k$ and $k'$ are both pointing on the $\ell'$-th position in word $\awordbis$). This allows us to deduce that $(\aword,k) \equivrel{N-1} (\aword',k')$.
\item If $j-i \geq \length{\awordbis}-\ell+\length{\awordbis}$ ($j$ is not in the same word $\awordbis$ as $i$ neither in the next word $\awordbis$) and $j < \length{\aword_1}+(M-N)\length{\awordbis}$ \commzone{$j$ is in Zone C}. Then let $\ell' = (j-(\length{\aword_1} + N \length{\awordbis})) \mod \length{\awordbis}$ the relative position of $j \in \awordbis$. We choose $j'=i'+(\length{\awordbis}-\ell)+\ell'$ ($j'$ is in the same position of $j$ in the word $\awordbis$ following the word $\awordbis$ $i$ belongs to). We have then that $j' < \length{\aword_1}+(M'-(N-1))\length{\awordbis}$ \commzone{$j'$ is in Zone C or D} (because $i'+(\length{\awordbis}-\ell) <\length{\aword_1}+(M-N)\length{\awordbis}$ and $\ell' < \length{\awordbis}$) and $|j-j'|=0 \mod \length{\awordbis}$ ($j$ and $j'$ are both pointing on the $\ell'$-th position in word $\awordbis$), hence $(\aword,j) \equivrel{N-1} (\aword',j')$. Let $k' \in \set{i',\ldots,j'-1}$. If $k'-i' < \length{\awordbis} -\ell$ ($k'$ and $i'$ are in the same word $\awordbis$), then let $k=i+k'-i'$. In this case we have $k < \length{\aword_1}+(M-N)\length{\awordbis}$ \commzone{$k$ is in Zone C} and since $|i-i'|=0 \mod \length{\awordbis}$, we also have $|k-k'|=0 \mod \length{\awordbis}$, hence $(\aword,k) \equivrel{N-1} (\aword',k')$. If $k'-i' \geq \length{\awordbis} -\ell$ ($k'$ and $i'$ are not in the same word $\awordbis$), then $j'-k' < \ell'$ and let $k=j-j'-k'$. In this case we have $k < \length{\aword_1}+(M-N)\length{\awordbis}$ \commzone{$k$ is in Zone C} and since $|j-j'|=0 \mod \length{\awordbis}$, we also have $|k-k'|=0 \mod \length{\awordbis}$, hence $(\aword,k) \equivrel{N-1} (\aword',k')$.

\end{itemize}
\end{itemize}
\textbf{Property (iv).} If $M=M'$ then take for all $j \geq i$, $j'=j$, this will implies that $(\aword,j) \equivrel{N-1} (\aword'j')$ and for all for all $k' \in \set{j'-1,\ldots,i'}$, take $k=k'$, hence we will have $k \in  \set{j-1,\ldots,i}$ and $(\aword, k) \equivrel{N-1} (\aword', k')$. Assume now $M \neq M'$ with $M >M'$ (the case $M'>M$ can be treated similarly). We proceed by a case analysis on the position  $i$ and $j$:
\begin{itemize}
\item If $i < \length{\aword_1}+N\length{\awordbis}$ \commzone{$i$ is in Zone A or B} then $j< \length{\aword_1}+N\length{\awordbis}$ \commzone{$j$ is in Zone A or B} and $i' < \length{\aword_1} + N\length{\awordbis}$ \commzone{$i'$ is in Zone A or B} and $i=i'$. We define $j'=j$. Then it is clear that $j' < i'$ and also that $(\aword,j) \equivrel{N} (\aword'j')$ and by Lemma \ref{lemma-word} we get $(\aword,j) \equivrel{N-1} (\aword'j')$. Let $k' \in \set{j'-1,\ldots,i'}$ and let $k=k'$, then we have that $k \in \set{j-1,\ldots,i}$ and also $(\aword,k) \equivrel{N} (\aword',k')$, hence by Lemma \ref{lemma-word}, $(\aword,k) \equivrel{N-1} (\aword',k')$.
%

%
\item If $i \geq \length{\aword_1}+(M-N)\length{\awordbis}$ \commzone{$i$ is Zone D or E} then $i'\geq \length{\aword_1}+(M'-N)\length{\awordbis}$ \commzone{$i'$ is in Zone E} and $i=i'+(M-M')\length{\awordbis}$ and we have the following possibilities for the position of $j \leq i$:
\begin{itemize}
\item if $j \geq \length{\aword_1}+(M-N) \length{\awordbis}$ \commzone{$j$ is in Zone D or E}, then let $j'=j-(M-M')\length{\awordbis}$. Consequently we have $(\aword,j) \equivrel{N} (\aword'j')$ and by Lemma \ref{lemma-word} we get $(\aword,j) \equivrel{N-1} (\aword'j')$. Let $k' \in \set{j'-1,\ldots,i'}$ and $k=k'+(M-M')\length{\awordbis}$. Then we have that $k \in \set{j-1,\ldots,i}$ and also $(\aword,k) \equivrel{N} (\aword',k')$ and by Lemma \ref{lemma-word}, $(\aword,k) \equivrel{N-1} (\aword',k')$.
\item if $\length{\aword_1}+N \length{\awordbis} \leq j < \length{\aword_1}+(M-N) \length{\awordbis}$ \commzone{$j$ is in Zone C}, then let $\ell=(j-(\length{\aword_1}+N \length{\awordbis})) \mod \length\awordbis$ ($\ell$ is the relative position of $j$ in the word $\awordbis$ it belongs to). Consequently $0 \leq \ell < \length{\awordbis}$. Let $j'= \length{\aword_1} + (M'-N) \length{\awordbis}- (\length{\awordbis}-\ell$ ($j'$ is at the same position as $j$ in the last word $\awordbis$ of the Zone C). Then $\length{\aword_1}+N \length{\awordbis} \leq j' < \length{\aword_1}+(M'-N) \length{\awordbis}$ \commzone{$j'$ is in Zone C} (because $(M'-M)>N$) and $|j - j'| = 0 \mod \length{\awordbis}$ (they are at the same position in the word $\awordbis$). We deduce that $(\aword,j) \equivrel{N} (\aword',j')$ and by Lemma \ref{lemma-word} we get $(\aword,j) \equivrel{N-1} (\aword'j')$.
Then let $k' \in \set{j'-1,\ldots,i}$ and let $k=k'+(M-M') \length{\awordbis}$. Then we have that $k \in \set{j-1,\ldots,i}$. Furthermore if $k' \geq \length{\aword_1}+(M'-N) \length{\awordbis}$ \commzone{$k'$ is in Zone D or E} then $k \geq \length{\aword_1}+(M-N) \length{\awordbis}$ \commzone{$k$ is in Zone D or E} and we obtain $(\aword,k) \equivrel{N} (\aword',k')$ and by Lemma \ref{lemma-word}, $(\aword,k) \equivrel{N-1} (\aword',k')$. And if $k'< \length{\aword_1}+(M'-N) \length{\awordbis}$  then necessarily $ \length{\aword_1}+N \length{\awordbis} \leq k'$ \commzone{$k'$ is in Zone C} (because $j'<k'$) and $\length{\aword_1}+N \length{\awordbis}\leq k< \length{\aword_1}+(M-N) \length{\awordbis}$ (because $M>M'$) and also $|k - k'| = 0 \mod \length{\awordbis}$ (because $k=k'+(M-M') \length{\awordbis}$). Hence $(\aword,k) \equivrel{N} (\aword', k')$ and Lemma \ref{lemma-word} tells us that $(\aword,k) \equivrel{N-1} (\aword', k')$.
\item if $j< \length{\aword_1}+N \length{\awordbis}$ \commzone{$j$ is in Zone A}, let $j'=j$. We have then $j' < \length{\aword_1}+N\length{\awordbis}$ \commzone{$j'$ is in Zone A} and we deduce that $(\aword,j) \equivrel{N} (\aword'j')$ and by Lemma \ref{lemma-word} we get $(\aword,j) \equivrel{N-1} (\aword'j')$. Then let $k' \in \set{j'-1,\ldots,i'}$. If $k' < \length{\aword_1}+N \length{\awordbis}$ \commzone{$k'$ is in Zone A}, for $k=k'$, we obtain $(\aword,k) \equivrel{N} (\aword',k')$ and by Lemma \ref{lemma-word}, $(\aword,k) \equivrel{N-1} (\aword',k')$. If $k' \geq \length{\aword_1}+(M'-N) \length{\awordbis}$ \commzone{$k'$ is in Zone D or E}, we choose $k=k'+(M-M') \length{\awordbis}$ and here also we deduce $(\aword,k) \equivrel{N} (\aword',k')$ and by Lemma \ref{lemma-word}, $(\aword,k) \equivrel{N-1} (\aword',k')$. If ${\aword_1}+N \length{\awordbis} \leq k' < \length{\aword_1}+(M'-N) \length{\awordbis}$ \commzone{$k'$ is in Zone C}, let $\ell=(k'-(\length{\aword_1}+N \length{\awordbis})) \mod \length{\awordbis}$ ($\ell$ is the relative position of $k'$ in the word $\awordbis$ it belongs to) and let $k=\length{\aword_1}+N \length{\awordbis}+\ell$ ($k$ is at the same position of $k'$ in the first word of the zone C). Then we have ${\aword_1}+N \length{\awordbis} \leq k < \length{\aword_1}+(M-N) \length{\awordbis}$ \commzone{$k$ is in the Zone C} and $|k-k'|=0 \mod \length{\awordbis}$ which allows to deduce that $(\aword,k) \equivrel{N} (\aword',k')$ and by Lemma \ref{lemma-word}, $(\aword,k) \equivrel{N-1} (\aword',k')$.
\end{itemize}
\item If $\length{\aword_1}+N\length{\awordbis} \leq i < \length{\aword_1}+(M-N)\length{\awordbis}$ \commzone{$i$ in Zone C} then $\length{\aword_1}+N\length{\awordbis} \leq i' < \length{\aword_1}+(M'-N)\length{\awordbis}$ \commzone{$i'$ in Zone C} and $|i-i'|=0 \mod \length{\awordbis}$. Let $\ell=(i-(\length{\aword_1}+N\length{\awordbis})) \mod \length{\awordbis}$ (the relation position of $i$ in the word $\awordbis$ it belongs to). We have the following possibilities for the position of $j \leq i$:
\begin{itemize}
\item If $i-j <\ell+\length{\awordbis}$ ($j$ is  in the same word $\awordbis$ as $i$ or in the previous word $\awordbis$) then $j \geq \length{\aword_1}+(N-1)\length{\awordbis}$ \commzone{$j$ is in Zone B or C}. We define $j'=i'-(i-j)$ and we have that $\length{\aword_1}+(N-1)\length{\awordbis} \leq j' < \length{\aword_1}+(M'-N)\length{\awordbis}$ \commzone{$j'$ is in Zone B or C} and since $|i-i'|=0 \mod \length{\awordbis}$, we deduce  $|j-j'|=0 \mod \length{\awordbis}$. From this we obtain  $(\aword,j) \equivrel{N-1} (\aword',j')$. Let $k' \in \set{j'-1,\ldots,i}$ and $k=i-(i'-k')$. We have then that $k \in \set{j-1,\ldots,i}$ and $\length{\aword_1}+(N-1)\length{\awordbis} \leq k' < \length{\aword_1}+(M'-N)\length{\awordbis}$ \commzone{$k'$ is in Zone B or C} and $\length{\aword_1}+(N-1)\length{\awordbis} \leq k < \length{\aword_1}+(M-N)\length{\awordbis}$ \commzone{$k$ is in Zone B or C} and since $|i-i'|=0 \mod \length{\awordbis}$, we also have $|k-k'|=0 \mod \length{\awordbis}$. Consequently $(\aword,k) \equivrel{N-1} (\aword',k')$.
\item If $i-j \geq \ell+\length{\awordbis}$ ($j$ is not in the same word $\awordbis$ as $i$ neither in the previous word $\awordbis$) and $j  < \length{\aword_1}+N\length{\awordbis}$ \commzone{$j$ is in zone A or B}. Let $j'=j$ then $j' < \length{\aword_1}+N\length{\awordbis}$ and consequently $(\aword,j) \equivrel{N} (\aword',j')$ and using Lemma \ref{lemma-word} we get $(\aword,j) \equivrel{N-1} (\aword',j')$. Then let $k' \in \set{j'-1,\ldots,i'}$. If $k'< \length{\aword_1}+N\length{\awordbis}$ \commzone{$k'$ is i Zone A or B}, then let $k=k'$; we have in this case that $k  < \length{\aword_1}+N\length{\awordbis}$ and this allows us to deduce that $(\aword,k) \equivrel{N-1} (\aword',k')$. Now assume $k'\geq  \length{\aword_1}+N\length{\awordbis}$ \commzone{$k'$ is in Zone C} and $i'-k' \leq \ell$ ($k'$ and $i'$ are in the same word $\awordbis$), then let $k=i-(i'-k')$. In this case we have $k \geq \length{\aword_1}+N\length{\awordbis}$ \commzone{$k$ is in Zone C} and since $|i-i'|=0 \mod \length{\awordbis}$, we also have $|k-k'|=0 \mod \length{\awordbis}$, hence $(\aword,k) \equivrel{N-1} (\aword',k')$. Now assume $k'\geq \length{\aword_1}+N\length{\awordbis}$ \commzone{$k'$ is in Zone C} and $i'-k' > \ell$ ($k'$ and $i'$ are not in the same word $\awordbis$). We denote by $\ell'=(k'-(\length{\aword_1}+N\length{\awordbis})) \mod \length{\awordbis}$ the relation position of $k'$ in $\awordbis$ and let $k=i-\ell-(\length{\awordbis}-\ell')$ ($k$ is at the same position as $k'$ of $k$ in the word $\awordbis$  preceding the word $\awordbis$ $i$ belongs to). Then $k \in \set{j-1,\ldots,i}$ (because $\length{\awordbis}-\ell'< \length{\awordbis}$ and $i-j \geq \ell+\length{\awordbis}$) and $k \geq \length{\aword_1}+(N-1)\length{\awordbis}$ (because $i+(\length{\awordbis}-\ell) \geq \length{\aword_1}+(M-N)\length{\awordbis}$ and $\length{\awordbis}-\ell < \length{\awordbis}$) and $|k-k'|=0 \mod \length{u}$ ($k$ and $k'$ are both pointing on the $\ell'$-th position in word $\awordbis$). This allows us to deduce that $(\aword,k) \equivrel{N-1} (\aword',k')$.
\item If $j-i \geq \ell+\length{\awordbis}$ ($j$ is not in the same word $\awordbis$ as $i$ neither in the previous word $\awordbis$) and $j \geq \length{\aword_1}+N\length{\awordbis}$ \commzone{$j$ is in zone C}. Then let $\ell' = (j-(\length{\aword_1} + N \length{\awordbis})) \mod \length{\awordbis}$ the relative position of $j \in \awordbis$. We choose $j'=i'-\ell-(\length{\awordbis}-\ell')$ ($j'$ is in the same position of $j$ in the word $\awordbis$ preceding the word $\awordbis$ $i$ belongs to). We have then that $j' \geq \length{\aword_1}+(N-1)\length{\awordbis}$ \commzone{$j'$ is zone B or C} (because $i'-\ell \geq \length{\aword_1}+N)\length{\awordbis}$ and $\length{\awordbis}-\ell' \leq \length{\awordbis}$) and $|j-j'|=0 \mod \length{\awordbis}$ ($j$ and $j'$ are both pointing on the $\ell'$-th position in word $\awordbis$), hence $(\aword,j) \equivrel{N-1} (\aword',j')$. Let $k' \in \set{j'-1,\ldots,i'}$. If $i'-k' \leq \ell$ ($k'$ and $i'$ are in the same word $\awordbis$), then let $k=i-(i'-k')$. In this case we have $k \geq \length{\aword_1}+N\length{\awordbis}$ \commzone{$k$ is in Zone C} and since $|i-i'|=0 \mod \length{\awordbis}$, we also have $|k-k'|=0 \mod \length{\awordbis}$, hence $(\aword,k) \equivrel{N-1} (\aword',k')$. If $i'-k' > \ell$ ($k'$ and $i'$ are not in the same word $\awordbis$), then $k'-j' <\length{\awordbis}- \ell'$ and let $k=j+k'-j'$. In this case we have $ \length{\aword_1}+N\length{\awordbis} \leq k < \length{\aword_1}+(M-N)\length{\awordbis}$ \commzone{$k$ is in Zone C} and since $|j-j'|=0 \mod \length{\awordbis}$, we also have $|k-k'|=0 \mod \length{\awordbis}$, hence $(\aword,k) \equivrel{N-1} (\aword',k')$.

\end{itemize}
\end{itemize}

\end{proof}

\fi
We can now state our stuttering theorem for $\PLTL[\cdot]$.
\begin{theorem}[Stuttering theorem]
\label{theorem-stuttering}
Let  $\asys$ be a Kripke structure and $N \in \nat$.  If $\arun$ and $\arun'$ are two runs of $\asys$  and $i,i' \in \nat$ two positions such that $(\wordof{\arun},i) \equivrel{N} (\wordof{\arun'},i')$ then for all $\PLTL[\cdot]$ formulae $\aformula$ of temporal depth at most $N$, we have $\arun,i \models \aformula$ if and only if $\arun',i \models \aformula$.
\end{theorem}

\begin{proof}
The proof of this lemma is done by induction on the height of  $\aformula$ using the result of Lemmas \ref{lemma-word} and \ref{lemma-word-bis}. 
\end{proof}
}
\fi 

\section{Fundamental Structures: Minimal Path Schemas}
\label{section-path-schemas}
\newcommand{\lengthpathschema}[1]{{\rm len}(#1)}
In this section, we introduce  the notion of a fundamental structure for flat
counter systems, namely path schemas. Indeed, every flat counter system
can be decomposed into a finite set of minimal path schemas and there are only
an exponential number of them. So, all 
 our nondeterministic
algorithms \iflong to solve model-checking problems \fi  
on flat counter systems have a 
preliminary step that first guesses  a minimal path schema. 

\subsection{Minimal Path Schemas}
\label{section-minimal-path-schemas}

Let $\asys=\tuple{\states,\counters_n,\edges,\alabelling}$ be a flat counter system. 
A \defstyle{path segment} $\aseg$ of $\asys$ is a finite sequence of transitions
from $\edges$  such that $\target{\aseg(i)}=\source{\aseg(i+1)}$ for all 
$0 \leq i < \length{\aseg}-1$. We write $\first{\aseg}$ [resp.
 $\last{\aseg}$] to denote the first [resp. last] control state of a 
path segment, in other words $\first{\aseg}=\source{\aseg(0)}$ and 
$\last{\aseg}=\target{\aseg(\length{\aseg}-1)}$. We also write $\effect{\aseg}$ to denote 
the sum vector 
$\sum_{0 \leq i < \length{\aseg}} \update{\aseg(i)}$ representing the total
effect of the  updates along the path segment. A path segment $\aseg$ is said to be \defstyle{simple} if $\length{\aseg} >0$  and for all $0 \leq i,j < \length{\aseg}$, 
$\aseg(i)=\aseg(j)$ implies $i=j$ (no repetition of transitions). 
A \defstyle{loop} is a simple path segment $\aseg$ such that 
$\first{\aseg}=\last{\aseg}$. 
\iflong If a path segment is not a loop it is called a \emph{non-loop segment}. \fi 
A \defstyle{path schema} $\aschema$ is an $\omega$-regular expression 
built over 
\iflong the alphabet of transitions \fi 
such that its language represents an overapproximation of the set of labels obtained from infinite runs 
following the transitions of $\aschema$. 
\iflong 
More precisely, 
a path schema $\aschema$ is
of the form $\aseg_1 \aloop_1^+ \aseg_2 \aloop_2^+ \ldots \aseg_k \aloop_k^{\omega}$ 
verifying the following conditions:
\begin{enumerate}
\itemsep 0 cm
\item $\aloop_1$, \ldots, $\aloop_k$ are loops,
\item $\aseg_1 \aloop_1 \aseg_2 \aloop_2 \ldots \aseg_k \aloop_k$ 
is a path segment.
\end{enumerate}
\fi

We write $\lengthpathschema{\aschema}$ to denote
$\length{\aseg_1 \aloop_1 \aseg_2 \aloop_2 \ldots \aseg_k \aloop_k}$  and
$\nbloops{\aschema}$ as its number $k$ of loops.  Let $\languageof{\aschema}$ denote the set of infinite words in $\edges^\omega$ which belong to the language defined
by $\aschema$. Note that some elements of $\languageof{\aschema}$ may not 
correspond to any run because of constraints on counter values.  Given
$\aword \in \languageof{\aschema}$, we write $\loopsof{\aschema}{\aword}$ to denote the unique tuple in $(\nat \setminus \set{0})^{k-1}$ such that
$\aword=\aseg_1 \aloop_1^{\loopsof{\aschema}{\aword}[1]}\aseg_2 \aloop_2^{\loopsof{\aschema}{\aword}[2]}\ldots
\aseg_k\aloop_k^\omega$. So,  for every $i \in \interval{1}{k-1}$, 
$\loopsof{\aschema}{\aword}[i]$ is the number of times the loop $\aloop_i$ is taken.  Then, for a configuration $\aconf_0$, the set $\loopsof{\aschema}{\aconf_0}$ is
the set of vectors $\set{\loopsof{\aschema}{\aword} \in (\nat \setminus \set{0})^{k-1}
\mid \aword \in \wordof{\aschema,\aconf_0}}$. Finally, we  say that a run $\arun$ 
starting in a configuration $\aconf_0$ \defstyle{respects} a path schema $\aschema$ if $\wordof{\arun} \in \wordof{\aschema,\aconf_0}$ and for such a run, we write
$\loopsof{\aschema}{\arun}$ to denote $\loopsof{\aschema}{\wordof{\arun}}$.  Note that by definition, if $\arun$ respects $\aschema$, then each loop $\aloop_i$ is
visited at least once, and the last one infinitely.

So far, a flat counter system may have an infinite set of path schemas.
\iflong To see this, it is sufficient to unroll loops in path segments. \fi 
However, we can impose minimality conditions on path schemas without sacrificing
completeness. 
%
%
\iflong
A path schema $\aseg_1 \aloop_1^+ \aseg_2 \aloop_2^+ \ldots \aseg_k \aloop_k^{\omega}$ is 
\defstyle{minimal}
whenever  
\begin{enumerate}
\itemsep 0 cm
\item $\aseg_1 \cdots \aseg_k$ is  either the empty word or a simple non-loop segment,
\item $\aloop_1$, \ldots, $\aloop_k$ are loops with disjoint sets of transitions.
\end{enumerate} 
\fi

\begin{lemma}
\label{lemma-schemata-finite}
Given a flat counter system $\asys=\tuple{\states,\counters_n,\edges,\alabelling}$,
the total number of minimal path schemas of $\asys$ is finite and is smaller than 
$\card{\edges}^{(2 \times \card{\edges})}$. 
\end{lemma}
%
%
This is a simple consequence of the fact that in a minimal path schema, each transition occurs 
at most twice. 
In Figure~\ref{figure-example-minimal-path-schemas}, we present a flat counter system $\asys$ with a unique
counter  and one of its
minimal path schemas. Each transition $\atransition_i$ labelled by $+i$ corresponds to a transition with
the guard $\top$ and the update value $+i$. The minimal path schema shown in  
Figure~\ref{figure-example-minimal-path-schemas} corresponds
to the $\omega$-regular expression 
$\atransition_1  (\atransition_2 \atransition_3)^+ 
\atransition_4 \atransition_5 
(\atransition_6 \atransition_5)^{\omega}$. 
\iflong
In order to avoid confusions between path schemas and flat counter systems
that look like path schemas, simple loops in the representation are labelled by  \fbox{$\omega$} or \fbox{$\geq 1$}
depending whether the simple loop is the last one or not.
\fi
 Note that in the representation of
path schemas, a state may occur several times, as it is the case for $\astate_3$
(this cannot occur in the representation of counter systems). 
\begin{figure}
\begin{center}
\scalebox{0.8}{
\begin{tikzpicture}[->,>=stealth',shorten >=1pt,
node distance=1cm, thick,auto,bend angle=60]

  \tikzstyle{every state}=[fill=white,draw=black,text=black]
  \node[state] (q0) [below]    {$\astate_0$};
  \node[state] (q1) [right= of q0]    {$\astate_1$};
  \node[state] (q2) [above= of q1]    {$\astate_2$};
  \node[state] (q3) [right=2cm of q1]    {$\astate_3$};
  \node[state] (q4) [above= of q3]    {$\astate_4$};

  
  \node[state] (q02) [right=2.2cm of q3]    {$\astate_0$};
  \node[state] (q12) [right= of q02]    {$\astate_1$};
  \node[state] (q22) [above= of q12]    {$\astate_2$};
  \node[state] (q32) [right= of q12]    {$\astate_3$};
  \node[state] (q42) [right=1.4cm of q32]    {$\astate_4$};
  \node[state] (q52) [above= of q42]    {$\astate_3$};
  \node (q62) [above=0.3cm of q12] {{\small \fbox{$\mathbf{\geq 1}$}}};
  \node (q72) [above=0.3cm of q42] {{\small \fbox{$\mathbf{\omega}$}}};

  \path[->] (q0) edge node {{\small $+1$}} (q1);
  \path[->] (q1) edge [bend left] node [right] {{\small $+2$}} (q2);
  \path[->] (q2) edge [bend left] node  [right] {{\small $+3$}} (q1);
  \path[->] (q1) edge node  [above] {{\small $+4$}} (q3);
  \path[->] (q3) edge  [bend left] node  [left] {{\small $+5$}} (q4);
  \path[->] (q4) edge  [bend left] node  [right] {{\small $+6$}} (q3);     

 
  \path[->] (q02) edge node {{\small $+1$}} (q12);
  \path[->] (q12) edge  [bend left] node {{\small $+2$}} (q22);
  \path[->] (q22) edge  [bend left] node {{\small $+3$}} (q12);
  \path[->] (q12) edge  node [below] {{\small $+4$}} (q32);
  \path[->] (q32) edge node  [above] {{\small $+5$}} (q42);
  \path[->] (q42) edge  [bend left] node  [left] {{\small $+5$}} (q52);
  \path[->] (q52) edge  [bend left] node  [right] {{\small $+6$}} (q42);

\end{tikzpicture}
}
\end{center}
\caption{A flat counter system and one of its minimal path schemas}
\label{figure-example-minimal-path-schemas}
\end{figure}
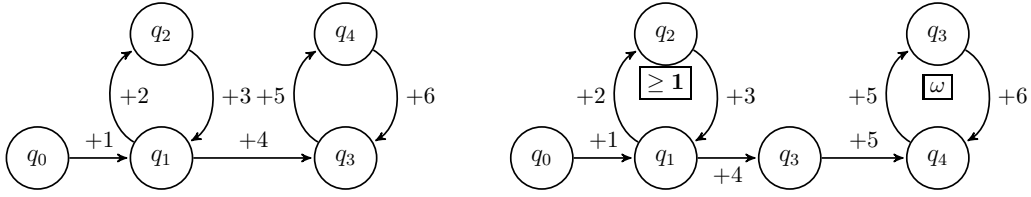
\iflong
Minimal path schemas play a crucial role in the sequel, mainly because of the properties
stated below.%
\begin{lemma} Let $\aschema$ be a path schema. There is a minimal path schema $\aschema'$ such that
every run respecting $\aschema$ respects $\aschema'$ too. 
\end{lemma}
The proof of the above lemma is by an easy verification. Indeed, whenever a maximal
number of copies of a simple loop is identified
as a factor of  $\aseg_1 \aloop_1 \cdots \aseg_k \aloop_k$, this factor is replaced by the simple loop unless
it is already present in the path schema.
\fi 

Finally, the conditions imposed on the structure of path schemas implies the following corollary which states that 
the number of minimal path schemas for a given flat counter system is at most exponential in the size of the 
system (see similar statements in~\cite{Leroux&Sutre05}).
\begin{corollary}
\label{corollary-ps-cover-runs} 
Given a flat counter system $\asys$ and a configuration $\aconf_0$,
there is a finite set of minimal path schemas $\aset$ of cardinality at most 
$\card{\edges}^{(2 \times \card{\edges})}$ 
such that
$\wordof{\aconf_0} = \wordof{\bigcup_{\aschema \in \aset} \aschema, \aconf_0}$. 
\end{corollary}


\subsection{Complexity Results}

We write $\cps$ [resp. $\kps$] to denote the class of path schemas from counter systems
[resp. the class of path schemas from Kripke structures].  
\iflong
As a first step towards our main result, we consider the model-checking 
problem for $\PLTL[\emptyset]$ over a path schema for a flat Kripke structure. We write 
$\mc{\PLTL[\emptyset]}{\kps}$ to denote the problem  defined  below:
\begin{description}
\itemsep 0 cm 
\item[Input:] A flat Kripke structure $\asys$, a path schema $\aschema$ of 
$\asys$, a configuration $\aconf_0$ and a formula $\aformula$ of $\PLTL[\emptyset]$;
\item[Output:] Does there exist a run $\arun$ starting from $\aconf_0$ which 
respects $\aschema$ and such that $\arun,0 \models \aformula$?
\end{description}
If the answer is "yes", we will write $\aschema,\aconf_0 \models \aformula$.
\fi 
Let $\arun$ and $\arun'$ be  runs respecting  $\aschema$.
For $\alpha \geq 0$, we write $\arun \equiv_{\alpha} \arun'$ $\equivdef$
for every $i \in \interval{1}{\nbloops{\aschema}-1}$, we have
$\mathtt{Min}(\loopsof{\aschema}{\arun}[i], \alpha) = \mathtt{Min}(\loopsof{\aschema}{\arun'}[i], \alpha)$. 
We state below a result concerning the runs of 
flat counter systems \iflong (including flat Kripke structures) \fi  when respecting the same path schema.

\begin{proposition}
\label{proposition-iter-flatks}
Let  $\asys$ be a flat counter system, $\aschema$ be a path schema, 
and $\aformula \in \PLTL[\emptyset]$.
For all runs $\arun$ and $\arun'$ respecting $\aschema$ 
such that  
$\arun \equiv_{ 2 \tempdepth{\aformula} +5} \arun'$, we have $\arun, 0 \models
\aformula$ iff  $\arun', 0 \models
\aformula$. 
\end{proposition}
This property can be proved by applying Theorem~\ref{theorem-stuttering} 
repeatedly in order to get rid of the unwanted iterations of the loops. 
Our algorithm for  $\mc{\PLTL[\emptyset]}{\kps}$ takes advantage of a result 
from~\cite{Laroussinie&Markey&Schnoebelen02}
for model-checking ultimately periodic models with formulae from Past LTL. 
An \defstyle{ultimately periodic path} 
is an infinite word in $\edges^\omega$ of the form $\awordbis \awordter^\omega$ were  
$\awordbis \awordter$ is a path segment and consequently $\first{\awordter} = \last{\awordter}$. 
More generally,  an \defstyle{ultimately periodic} word over \iflong the alphabet \fi  $\aalphabet$
is an $\omega$-word in $\aalphabet^{\omega}$ that can be written as
$\aword \cdot \awordbis^{\omega}$ where $\aword$ is the \defstyle{prefix} and
$\awordbis$ is the \defstyle{loop}.
According to~\cite{Laroussinie&Markey&Schnoebelen02}, 
given an ultimately periodic path $\aword$,  and a formula 
$\aformula \in \PLTL[\emptyset]$, the problem of checking whether there exists a run $\arun$  
such that $\wordof{\arun}=\aword$ and $\arun,0 \models \aformula$ is in $\ptime$ (a tighter bound of \nc~can be obtained by combining 
results from~\cite{Kuhtz10} and Theorem~\ref{theorem-stuttering}). 
Observe that $\arun$ is unique if such a run exists.

\begin{lemma}
\label{lemma-modelcheckps-np}
$\mc{\PLTL[\emptyset]}{\kps}$ is in \np.
\end{lemma}
The proof is a consequence of Proposition~\ref{proposition-iter-flatks} and~\cite{Laroussinie&Markey&Schnoebelen02}. 
Indeed, given $\aformula \in \PLTL[\emptyset]$ and $\aschema = \aseg_1 \aloop_1^+ \aseg_2 \aloop_2^+ \ldots \aseg_k \aloop_k^\omega$, first guess 
$\vect{m} \in \interval{1}{2 \tempdepth{\aformula}+5}^{k-1}$
and check whether $\arun, 0 \models \aformula$ where $\arun$ is the obvious ultimately periodic word such that 
$\wordof{\arun}= \aseg_1 \aloop_1^{\vect{m}[1]} \aseg_2 \aloop_2^{\vect{m}[2]} \ldots \aseg_k \aloop_k^\omega$. 
Since  $\vect{m}$ is of polynomial size and  $\arun, 0 \models \aformula$
can be checked in polynomial time by~\cite{Laroussinie&Markey&Schnoebelen02}, we get the \np \ upper bound.

From \cite{Kuhtz&Finkbeiner11}, we have the lower bound for $\mc{\PLTL[\emptyset]}{\kps}$.
%
%
\begin{lemma}
\label{lemma-nphardness} \cite{Kuhtz&Finkbeiner11}
$\mc{\PLTL[\emptyset]}{\kps}$ is \np-hard even if restricted to  $\mynext$ and $\eventually$. 
\end{lemma}
For the sake of completeness, we provide the proof presented in \cite{Kuhtz&Finkbeiner11} adapted to our context.
\ifshort
\else

\begin{proof}
The proof is by reduction from the SAT problem and it is included for the sake of being
self-contained. 
Let $\abigformula$ be a Boolean formula built over the propositional variables
$AP = \set{\avarprop_{1},\cdots,\avarprop_{n}}$. We build a path schema $\aschema$ and a formula $\aformulabis$
such that $\abigformula$ is satisfiable iff there is a run respecting $\aschema$ and satisfying $\aformulabis$. 
The path schema $\aschema$ is the one described in Figure~\ref{fig:nphardness} so that the
truth of the propositional variable $\avarprop_i$ is encoded by the fact that the loop containing
$\astate_i$ is visited twice, otherwise it is visited once.
\begin{figure}
  \begin{center}
    \begin{tikzpicture}[shorten >=1pt,node distance=1.1cm,auto,>=stealth']

  \node[state] (q_1)                  {};
  \node (a1)[above=0.3cm of q_1]  {\fbox{{\bf \tiny  $\geq 1$}}};
  \node[state]         (q_2)  [right=of q_1]  {};
  \node (a2)[above=0.3cm of q_2]  {\fbox{{\bf \tiny  $\geq 1$}}}; 
  \node[state]         (q_3)  [right=of q_2]  {};
   \node (a3)[above=0.3cm of q_3]  {\fbox{{\bf \tiny  $\geq 1$}}};
  \node[state]         (q_4)  [above=of q_1]  {$\astate_1$};
  \node[state]         (q_5)  [above=of q_2]  {$\astate_2$};
  \node[state]         (q_6)  [above=of q_3]  {$\astate_3$};
  \node[state]         (q_n)  [right=of q_3]  {};
   \node (a4)[above=0.3cm of q_n]  {\fbox{{\bf \tiny  $\geq 1$}}};
  \node[state]         (q_n+1)[above=of q_n]  {$\astate_n$};
  \node[state]         (q_f)  [right=of q_n]{};
   \node (a5)[above=0.2cm of q_f]  {\fbox{{\bf \small $\omega$}}};
  \path[->] (q_1)   edge [bend left]  node  {} (q_4)
                    edge              node  {} (q_2)
            (q_2)   edge [bend left]  node  {} (q_5)
                    edge              node  {} (q_3)
            (q_3)   edge [bend left]  node  {} (q_6)
            (q_4)   edge [bend left]  node  {} (q_1)
            (q_5)   edge [bend left]  node  {} (q_2)
            (q_6)   edge [bend left]  node  {} (q_3)
            (q_n)   edge [bend left]  node  {} (q_n+1)
                    edge              node  {} (q_f)
            (q_n+1) edge [bend left]  node  {} (q_n);
  \path[->] (q_f)  edge [out=125,in=55,loop]  node  {} (q_f);
  \draw[dashed] (4.9,0) -- (5.9,0);
\end{tikzpicture}
  \end{center}
  \caption{A simple path schema $\aschema$}
  \label{fig:nphardness}
\end{figure}
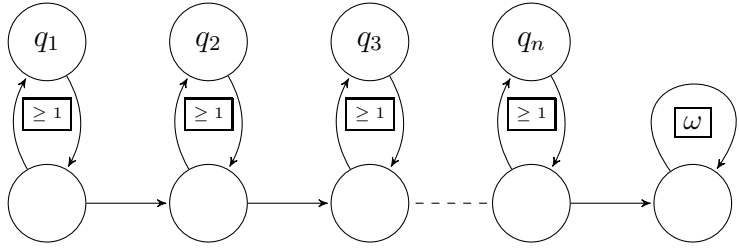
The formula $\aformulabis$ is defined as a conjunction $\aformulabis_{1 \vee 2} \wedge 
\aformulabis_{truth}$ where  $\aformulabis_{1 \vee 2}$ states that each loop is visited at most twice
and $\aformulabis_{truth}$ establishes the correspondence between the truth of $p_i$ and the number of times
the loop containing $\astate_i$ is visited. 
Formula  $\aformulabis_{1 \vee 2}$ is equal to
$[\bigwedge_{i}(\always(\astate_{i}\wedge \mynext\mynext \astate_{i}\Rightarrow 
\mynext\mynext\mynext\always\neg \astate_{i}))]
$
whereas $\aformulabis_{truth}$ is defined from $\abigformula$ by replacing each occurrence
of $\avarprop_i$ by $\eventually (\astate_{i}\wedge \mynext\mynext \astate_{i})$. 

Let us check the correctness of the reduction. Let $v: AP \rightarrow \set{\top, \perp}$
be a valuation satisfying  $\abigformula$. Let us consider the run  $\arun$ respecting
$\aschema$
such that $\loopsof{\aschema}{\arun}[i] \egdef 2$ if $v(\avarprop_i) = \top$, otherwise 
 $\loopsof{\aschema}{\arun}[i] \egdef 1$ for all $i \in \interval{1}{n}$. 
It is easy to check that $\arun, 0 \models \aformulabis$.
Conversely, if there is a run $\arun$  respecting $\aschema$ such that
$\arun, 0 \models \aformula$, the valuation $v$ satisfies $\abigformula$ where
for all $i \in \interval{1}{n}$, we have $v(\avarprop_i) = \top$ $\equivdef$  $\loopsof{\aschema}{\arun}[i] = 2$.
\end{proof}

\fi 
\iflong
The \np-completeness of $\mc{\PLTL[\emptyset]}{\kps}$ can then be deduced from the two previous lemmas.
%
%
\fi 
We also consider the case where we restrict the class of path schemas by bounding the number of allowed loops. Hence, for a fixed $k \in \Nat$, we write  $\mc{\PLTL[\emptyset]}{\kps(k)}$ to denote the restriction
of  $\mc{\PLTL[\emptyset]}{\kps}$ to path schemas with at most $k$ loops. 
Note that when $k$ is fixed, the number of
 ultimately periodic paths $\aword$ in $\languageof{P}$ such that
each loop (except the last one) is visited is at most  $2 \tempdepth{\aformula} + 5$ times
is bounded by $(2 \tempdepth{\aformula} + 5)^k$, which is polynomial in the size of
the input (because $k$ is fixed). From these considerations, we deduce the following result.
%

\begin{theorem}
\label{theorem-kps-fixed} 
$\mc{\PLTL[\emptyset]}{\kps}$ is \np-complete. \\
Given a fixed  $k \in \nat$, $\mc{\PLTL[\emptyset]}{\kps(k)}$ is in \ptime.
\end{theorem}
Note that it can be proved that $\mc{\PLTL[\emptyset]}{\kps(n)}$ is in \nc, hence giving a tighter upper bound for the problem. This can be obtained by observing that we can run the \nc~algorithm from~\cite{Kuhtz10}~for model checking $\PLTL[\emptyset]$ over ultimately periodic paths 
parallelly on $(2 \tempdepth{\aformula} + 5)^k$ (polynomially many) 
different paths.

Now, we present  how to solve $\mc{\PLTL[\emptyset]}{\flatks}$ using Lemma \ref{lemma-modelcheckps-np}. From Lemma~\ref{lemma-schemata-finite}, we know that the number of minimal path schemas in a 
flat Kripke structure $\asys=\tuple{\states,\edges,\alabelling}$ 
is finite and the length of a minimal path schema is 
at most $2 \times \card{\edges}$.
Hence, for solving the model-checking problem for a state 
$\astate$ and a $\PLTL[\emptyset]$ formula $\aformula$, a possible algorithm consists in 
choosing non-deterministically a minimal path schema $\aschema$  starting at $\astate_0$ of the given initial configuration $\aconf_0$ and then 
apply the algorithm used to establish Lemma~\ref{lemma-modelcheckps-np}. 
This new algorithm would be in \np. 
Furthermore, thanks to  Corollary~\ref{corollary-ps-cover-runs}, we 
know that if there exists a run $\arun$ of $\asys$ such that $\arun,0 \models \aformula$ then there 
exists a minimal path schema $\aschema$ such that $\arun$ respects $\aschema$. 
Consequently there is an algorithm in \np \  to solve $\mc{\PLTL[\emptyset]}{\flatks}$ and \np-hardness can be established as a variant of the proof of Lemma~\ref{lemma-nphardness}.

\begin{theorem}
$\mc{\PLTL[\emptyset]}{\flatks}$ is \np-complete.
\end{theorem}
%
%
%
\subsection{Some lower bounds in the presence of counters}

We will now provide some complexity lower bounds when considering path schemas over counter systems. 
%
%
As for path schemas from Kripke structures, we use $\cps(k)$ to denote the class of path schemas obtained from flat counter systems with number of loops 
bounded by $k$. Surprisingly, in the presence of counters, bounding the number of loops entails \np-hardness. 
\begin{lemma}
\label{lemma-constant-loops2}
For $k \geq 2$, 
$\mc{\PLTL[\counters]}{\cps(k)}$ is \np-hard.
\end{lemma}
\ifshort
\else 
The proof is by reduction from SAT  
and it is less  straightforward than the proof for Lemma~\ref{lemma-nphardness} or 
the reduction presented in~\cite{Kuhtz&Finkbeiner11} when path schemas are involved. Indeed, 
we cannot  encode the nondeterminism in the structure itself and 
 the structure has only a constant number of loops. Actually, we cannot use a separate loop for 
each counter; the reduction is done by encoding the nondeterminism in the
(possibly exponential) number of times  a single loop is taken, and then using its binary encoding  
as an assignment for the propositional variables.

\begin{proof}
The proof is by reduction from the  problem SAT.  Let $\abigformula$ be a
Boolean formula built over the propositional variables in $
\set{\avarprop_{1},\cdots,\avarprop_{n}}$. We build a path schema
$\aschema \in \cps(2)$, an initial configuration (all counters will be equal to
zero)
and a formula $\aformulabis$ such that
$\abigformula$ is satisfiable iff there is a run respecting $\aschema$
and starting at the initial configuration such that
it satisfies $\aformulabis$.  The path schema $\aschema$ is the one
described in Figure~\ref{fig:constnphardness}; it has one internal loop
and a second loop that is visited infinitely. 
The guard $\acounter_1 \leq 2^n$ enforces that the first loop is visited $\alpha$
times with $\alpha \in \interval{1}{2^n}$, which corresponds to guess
a propositional valuation such that the truth
value of the propositional variable $\avarprop_i$ is $\top$
whenever
the $i$th bit of $\alpha -1$ is equal to 1. 
When $\alpha-1$ is encoded in binary with $n$ bits, we assume the first bit is the most significant
bit. Note that the internal
loop has to be visited at least once since $\aschema$ is a path schema. 

\begin{figure}
  \begin{center}
    \tikzset{every loop/.style={min distance=15mm,looseness=10}}
\begin{tikzpicture}[shorten >=1pt,node distance=3.1cm,auto,>=stealth']

  \node[state] (q0)                  {$\astate_0$};
  \node[state]         (q1)  [right=of q0]  {$\astate_1$};
  \path[->] (q0)  edge [out=125,in=55,loop]  node  {$\acounter_1 \leq 2^n$, {\tiny $\left(
\begin{array}{c}
1 \\ 1 \\ \vdots \\ 1 \\
\end{array}
\right)$}} node [swap]{\fbox{$\geq 1$}} (q0)
                  edge              node  {$\top$, {\tiny $\left(
\begin{array}{c}
0 \\ 0 \\ \vdots \\ 0 \\
\end{array}
\right)$}}  (q1)
            (q1)   edge [out=125,in=55,loop] node  {$\top$, {\tiny $\left(
\begin{array}{c}
2^n \\ 2^{n-1} \\  \vdots \\ 2^1 \\
\end{array}
\right)$}} node [swap]{\fbox{$\omega$}} ();
\end{tikzpicture}
  \end{center}
  \caption{Path schema $\aschema$}
  \label{fig:constnphardness}
\end{figure}

Since the logical language does not allow to access to the 
$i$th bit of a counter value, we  simulate the test by arithmetical constraints
in the formula when the second loop of the path schema is visited. 
For every $\alpha \in  \interval{1}{2^n}$ and every $i \in \interval{1}{n}$, we write
$\alpha^i_{u}$ to denote the value in $\interval{0}{2^i-1}$ corresponding to the
$i-1$ first bits of $\alpha-1$. Similarly, we write $\alpha^i_{d}$ to denote the value in
$\interval{0}{2^{n+1-i}-1}$ corresponding to the $(n+1-i)$ last bits of $\alpha-1$. 
Observe that $ \alpha - 1 = \alpha_u^i \times 2^{n-i+1} + \alpha^i_d$.
\iflong
One can show that the propositions below
are equivalent:
\begin{enumerate}
\itemsep 0 cm
\item $i$th bit of $\alpha-1$ is 1,
\item there is some $k \geq 0$ such that $k \times 2^{n+1-i} + (\alpha-1) \in
      \interval{2^n + 2^{n-i}}{2^n + 2^{n+1-i}-1}$. 
\end{enumerate} 
\else
One can show that (1.) $i$th bit of $\alpha-1$ is 1 iff
(2.) there is some $k \geq 0$ such that $k \times 2^{n+1-i} + (\alpha-1) \in
      \interval{2^n + 2^{n-i}}{2^n + 2^{n+1-i}-1}$. 
\fi 
Actually, we shall show that $k$ is unique and the only possible value is $2^{i-1} - \alpha^i_{u}$.
Before showing the equivalence between (1.) and (2.), we can observe that condition (2.) 
can be expressed by the formula
$
\eventually (\astate_1 \wedge
((\acounter_i - 1) \geq 2^n + 2^{n-i}) \wedge 
((\acounter_i - 1) \leq 2^n + 2^{n-i+1} -1))
$.

First, note that  $\interval{2^n + 2^{n-i}}{2^n + 2^{n+1-i}-1}$ contains $2^{n-i}$ distinct values
and therefore satisfaction of (2.) implies unicity of $k$ since $2^{n+1-i} > 2^{n-i}$. 
Second, $i$th bit of $\alpha-1$ is equal to 1 iff $\alpha^i_d \in \interval{2^{n-i}}{2^{n+1-i}-1}$.
Now, observe that $(2^{i-1} - \alpha^i_{u})  2^{n+1-i} + (\alpha-1) = 2^n + \alpha_d^i$.
So, if (1.), then  $\alpha^i_d \in \interval{2^{n-i}}{2^{n+1-i}-1}$ and consequently
$2^n + \alpha_d^i \in \interval{2^n + 2^{n-i}}{2^n + 2^{n+1-i}-1}$. So, 
there is some $k \geq 0$ such that $k \times 2^{n+1-i} + (\alpha-1) \in
      \interval{2^n + 2^{n-i}}{2^n + 2^{n+1-i}-1}$ (take $k = 2^{i-1} - \alpha^i_{u}$). 
Now, suppose that (2.) holds true. There is $k \geq 0$ such that $k \times 2^{n+1-i} + (\alpha-1) \in
      \interval{2^n + 2^{n-i}}{2^n + 2^{n+1-i}-1}$. 
So, $k \times 2^{n+1-i} + (\alpha-1) - 2^n \in  \interval{2^{n-i}}{2^{n+1-i}-1}$
 and
therefore $k \times 2^{n+1-i} + \alpha_d^i - (2^{i-1} - \alpha_u^i) \times 2^{n+1-i}
\in \interval{2^{n-i}}{2^{n+1-i}-1}$. Since the expression denotes a non-negative
value, we have $k \geq  (2^{i-1} - \alpha_u^i)$ (indeed $\alpha_d^i < 2^{n+1-i}$) 
and since it denotes a value less or equal to
$2^{n+1-i}-1$, we have $k \leq  (2^{i-1} - \alpha_u^i)$. Consequently, $k = 2^{i-1} - \alpha_u^i$
and therefore $\alpha_d^i \in \interval{2^{n-i}}{2^{n+1-i}-1}$, which is precisely equivalent
to the fact that the $i$th bit of $\alpha-1$ is equal to 1. 

The formula $\aformulabis$ is defined from $\abigformula$ by
replacing each occurrence of $\avarprop_i$ by 
$
\eventually (\astate_1 \wedge
((\acounter_i - 1) \geq 2^n + 2^{n-i}) \wedge 
((\acounter_i - 1) \leq 2^n + 2^{n-i+1} -1))
$. Intuitively, $\aschema$ contains one counter by propositional variable 
and all the counters hold the same value after
the first loop. Next, in the second loop, 
we check that the $i$th bit of $\alpha -1$ is one by incrementing $\acounter_i$
by $2^{n+1-i}$. We had to consider $n$ counters since the increments differ.
In order to check whether the $i$th  bit of counter $\acounter_i$ is one, we add
repeatedly  $2^{n+1-i}$ to the counter. Note that this ensures that 
 the bits at positions $i$ to $n$ remains the same for the
counter whereas the counter is incremented till its value is greater
or equal to $2^n$. Eventually, we may deduce that the counter value
will  belong to  $\interval{2^n +2^{n-i}}{2^n + 2^{n-i+1} -1}$.  
This
is explained Table~\ref{table:constnphardness} with $n=4$.
\begin{center}
  \begin{table}
  {\scriptsize
  \begin{tabular}{|l||l|l|l|l|l|l|l|l|l|l|l|l|l|l|l|l}
    \hline
     &0&1&2&3&4&5&6&7&8&9&10&11&12&13&14&15 \\
    \hline
    \hline
    & & & & & & & & & & & & & & & &  \\
    \hline
    $p_1$&0&0&0&0&0&0&0&0&1&1&1&1&1&1&1&1\\
    \hline
    $p_2$&0&0&0&0&1&1&1&1&0&0&0&0&1&1&1&1\\
    \hline
    $p_3$&0&0&1&1&0&0&1&1&0&0&1&1&0&0&1&1\\
    \hline
    $p_4$&0&1&0&1&0&1&0&1&0&1&0&1&0&1&0&1\\
    \hline
  \end{tabular}
 }
  \end{table}
  \begin{table}
  {\scriptsize
  \begin{tabular}{l|l|l|l|l|l|l|l|l|l|l|l|l|l|l|l|l|}
    \hline
     &16&17&18&19&20&21&22&23&24&25&26&27&28&29&30&31 \\
    \hline
    \hline
    &1&1&1&1&1&1&1&1&1&1&1&1&1&1&1&1 \\
    \hline
    &0&0&0&0&0&0&0&0&1&1&1&1&1&1&1&1\\
    \hline
    &0&0&0&0&1&1&1&1&0&0&0&0&1&1&1&1\\
    \hline
    &0&0&1&1&0&0&1&1&0&0&1&1&0&0&1&1\\
    \hline
    &0&1&0&1&0&1&0&1&0&1&0&1&0&1&0&1\\
    \hline
  \end{tabular}
  \caption{Table showing the effect of last loop for 4 variables}
  \label{table:constnphardness}
 }
  \end{table}

 \end{center}
Let us check the correctness of the reduction. 
Let $v: \set{\avarprop_1, \ldots,\avarprop_n}
\rightarrow
\set{\top, \perp}$.  be a valuation satisfying $\abigformula$. Let us
consider the run $\arun$ respecting $\aschema$ such that the first
loop is taken $\alpha =(v(p_1)v(p_2)\cdots v(p_n))_2 + 1$ times and the initial counter values are all equal to
zero. $\top$ is read as $1$, $\perp$ as $0$ and $(v(p_1)v(p_2)\cdots v(p_n))_2$ denotes the value of the
natural number made of $n$ bits in binary encoding. 
Hence, for every $i \in \interval{1}{n}$, 
the counter $\acounter_i$ contains the value $\alpha$ after the
first loop. As noted earlier, $v(p_i)=1$ implies that adding $2^{n-i+1}$
repeatedly to $\acounter_i$ in the last loop, we will hit 
$\interval{2^n+2^{n-i}}{2^n+2^{n-i+1}-1}$.  Hence,
the formula $
\eventually (\astate_1 \wedge
((\acounter_i - 1) \geq 2^n + 2^{n-i}) \wedge 
((\acounter_i - 1) \leq 2^n + 2^{n-i+1} -1))
$ will be
satisfied by $\arun$ iff $v(p_i)=1$.  It is easy to check thus, that
$\arun, 0 \models \aformulabis$.  Conversely, if there is a run
$\arun$ respecting $\aschema$ such that $\arun, 0 \models \aformula$ and the initial counter values are all equal to
zero,
the valuation $v$ satisfies $\aformula$ where for all $i \in
\interval{1}{n}$, we have $v(p_i)$ iff the $i^{th}$ bit in the binary
encoding of $\alpha-1$ is 1, where $\alpha$ is the number of times the first loop is taken.
\end{proof}

\fi

We will now see that also simple properties on flat counter systems, such as reachability can be proved to be already \np-hard.
\newcommand{\reach}[1]{{\rm REACH}(#1)}
\newcommand{\conf}[1]{conf_0(#1)}
First we note that, a path schema in $\cps$ can also be seen as a 
flat counter system with the additional condition
of taking each loop at least once. For any state $\astate$, we write 
 $\conf{\astate}$ to denote the configuration $\pair{\astate}{\tuple{0,\cdots,0}}$ (all counter values
are equal to zero). 
The reachability problem $\reach{\csfrag}$ for a class of counter system $\csfrag$ is defined as:
\begin{description}
\itemsep 0 cm 
\item[Input:] A counter system $\asys \in \csfrag$ and two states $\astate_0$ and $\astate_f$;
\item[Output:] Does there exist a finite run 
from $\conf{\astate_0}$ to  $\conf{\astate_f}$?
\end{description}
%
%
We have then the following result concerning the lower bound of reachability in flat counter systems and path schemas from flat counter systems.
\begin{lemma}
\label{lemma-reachability}
$\reach{\cps}$ and $\reach{\flatcs}$ are \np-hard.
\end{lemma}
\begin{proof}
 The proofs are by reduction from the SAT problem.
Using the fact that $\cps$ is a special and constrained $\flatcs$, we will only prove \np-hardness of 
$\reach{\cps}$ and hence,
as a corollary, have the result for $\reach{\flatcs}$. 
Let $\abigformula$ be a Boolean formula built over the propositional variables
$AP = \set{p_{1},\cdots,p_{n}}$. We build a path schema $\aschema$
such that $\abigformula$ is satisfiable iff there is a run respecting $\aschema$ starting with the configuration 
$\conf{\astate_0}$ visits the configuration $\conf{\astate_f}$. 
The path schema $\aschema$ is the one described in Figure~\ref{fig:reachnphardness} so that the
truth of the propositional variable $p_i$ is encoded by the fact that the loop incrementing
$\acounter_i$ is visited at least twice.
\begin{figure}
  \begin{center}
    \scalebox{0.8}{
    \begin{tikzpicture}[shorten >=1pt,node distance=0.65cm,auto,>=stealth']

  \node[state]         (q_1)   at (0,0)       {$\astate_0$};
  \node[state]         (q_2)  [right=of q_1]  {};
  \node[state]         (q_3)  [right=of q_2]  {};
  \node[state]         (q_4)  [right=of q_3]  {};
  \node[state]         (q_5)  [right=of q_4]  {};
  \node[state]         (q_6)  [right=of q_5]  {};
  \node[state]         (q_n)  [right=of q_6]  {};
  \node[state]         (q_f)  [right=of q_n]  {$\astate_f$};
  \path[->] (q_1)   edge [loop]  node {{\tiny \fbox{$\mathbf{\geq 1}$}}} node [swap] {{\tiny $\left(
						\begin{array}{c}
							1 \\ 0 \\ \vdots \\ 0 \\
						\end{array}
						\right)$}} ()
                    edge         node  {} (q_2)
            (q_2)   edge [loop]  node {{\tiny \fbox{$\mathbf{\geq 1}$}}} node [swap] {{\tiny $\left(
						\begin{array}{c}
							0 \\ 1 \\ \vdots \\ 0 \\
						\end{array}
						\right)$}} ()
            (q_3)   edge [loop]  node {{\tiny \fbox{$\mathbf{\geq 1}$}}} node [swap] {{\tiny $\left(
						\begin{array}{c}
							0 \\ 0 \\ \vdots \\ 1 \\
						\end{array}
						\right)$}} ()
	            edge         node [swap] {$\aguard$,{\tiny $\left(
						\begin{array}{c}
							0 \\ 0 \\ \vdots \\ 0 \\
						\end{array}
						\right)$}} (q_4)
            (q_4)   edge         node  {} (q_5)
            (q_5)   edge [loop]  node {{\tiny \fbox{$\mathbf{\geq 1}$}}} node [swap] {{\tiny $\left(
						\begin{array}{c}
							-1 \\ 0 \\ \vdots \\ 0 \\
						\end{array}
						\right)$}}() 
	            edge         node  {} (q_6)
            (q_6)   edge [loop]  node {{\tiny \fbox{$\mathbf{\geq 1}$}}} node [swap] {{\tiny $\left(
						\begin{array}{c}
							0 \\ -1 \\ \vdots \\ 0 \\
						\end{array}
						\right)$}} ()
            (q_n)   edge [loop]  node {{\tiny \fbox{$\mathbf{\geq 1}$}}} node [swap] {{\tiny $\left(
						\begin{array}{c}
							0 \\ 0 \\ \vdots \\ -1 \\
						\end{array}
						\right)$}} ()
                    edge         node  {} (q_f)
            (q_f)   edge [loop]  node  {{\small \fbox{$\mathbf{\omega}$}}} ();
  \draw[dashed] (2.2,0) -- (2.9,0);
  \draw[dashed] (9,0) -- (9.7,0);
\end{tikzpicture}
    }
  \end{center}
  \caption{A simple path schema}
  \label{fig:reachnphardness}
\end{figure}
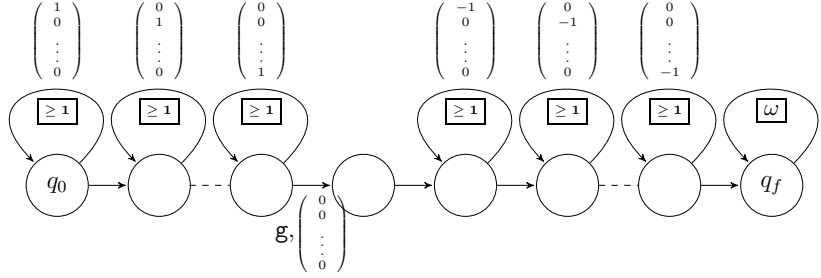
The guard $\aguard$ is defined as a formula that establishes the correspondence between the truth value of $p_i$ and the number of times
the loop incrementing $\acounter_i$ is visited. It is defined from $\abigformula$ by replacing each occurrence
of $p_i$ by $\acounter_{i}\geq 2$. Note that, since the $i^{th}$ and $(n+i)^{th}$ loops perform the complementary operation
on the same counters, both of the loops can be taken equal number of times. 

Let us check the correctness of the reduction. Let $v: AP \rightarrow \set{\top, \perp}$
be a valuation satisfying  $\abigformula$. Let us consider the run  $\arun$ respecting
$\aschema$
such that $\loopsof{\aschema}{\arun}[i] = k$ and $\loopsof{\aschema}{\arun}[n+i] = k$ for some $k\geq 2$, if $v(p_i) = \top$, otherwise 
 $\loopsof{\aschema}{\arun}[i] = 1$ and $\loopsof{\aschema}{\arun}[n+i] = 1$ for all $i \in \interval{1}{n}$. 
It is easy to check that the guard $\aguard$ is satisfied by the run and taking $i^{th}$ loop and $(n+i)^{th}$ loop equal number times
ensures resetting the counter values to zero. Hence the configuration $\conf{\astate_f}$ is reachable.
Conversely, if there is a run $\arun$  respecting $\aschema$ and starting with configuration $\conf{\astate_0}$ such that the configuration
$\conf{\astate_f}$ is reachable, then the guard $\aguard$ ensures that the valuation $v$ satisfies $\abigformula$ where
for all $i \in \interval{1}{n}$, we have $v(p_i) = \top$ $\equivdef$  $\loopsof{\aschema}{\arun}[i] \geq 2$.
\end{proof}

\iflong

\section{Characterizing Infinite Runs with a System of Equations}
\label{section-characterization}

In this section, we show how to build a system of equations from a path schema $\aschema$ and
a configuration $\aconf_0$
such that the system of equations encodes the set of all runs respecting $\aschema$
from $\aconf_0$. This can be done for path schemas without disjunctions in guards
that satisfy an additional \defstyle{validity} property. 
A path schema  $\aschema = \aseg_1 \aloop_1^+ \aseg_2 \aloop_2^+ \ldots \aseg_k \aloop_k^{\omega}$
is \defstyle{valid} 
\iflong
whenever it satisfies the conditions below:
\begin{enumerate}
\itemsep 0 cm
\item[1.] $\effect{\aloop_k} \geq \vect{0}$, 
\item[2.] if all the guards in transitions in $\aloop_k$ are conjunctions of atomic guards,
then  for each guard   occurring in the loop $\aloop_k$ 
of the form $\sum_i \afactor_i \avariable_i \sim b$ 
with $\sim \in \set{\leq, <}$ [resp. with $\sim \in \set{=}$, with 
$\sim \in \set{\geq, >}$] , we have 
$\sum_i \afactor_i \times \effect{\aloop_k}[i] \leq 0$
[resp. $\sum_i \afactor_i \times \effect{\aloop_k}[i] = 0$, 
 $\sum_i \afactor_i \times \effect{\aloop_k}[i] \geq 0$].
\end{enumerate}
\else
whenever $\effect{\aloop_k}[i] \geq 0$  for every $i\in \interval{1}{n}$ 
(see Section~\ref{section-path-schemas} 
for the definition of $\effect{\aloop_k}$) and 
if all the guards in transitions in $\aloop_k$ are conjunctions of atomic guards,
then  for each guard   occurring in the loop $\aloop_k$ 
of the form $\sum_i \afactor_i \avariable_i \sim b$ 
with $\sim \in \set{\leq, <}$ [resp. with $\sim \in \set{=}$, with 
$\sim \in \set{\geq, >}$] , we have 
$\sum_i \afactor_i \times \effect{\aloop_k}[i] \leq 0$
[resp. $\sum_i \afactor_i \times \effect{\aloop_k}[i] = 0$, 
 $\sum_i \afactor_i \times \effect{\aloop_k}[i] \geq 0$].
\fi 
It is easy to check that these conditions are necessary to visit the last
loop $\aloop_k$ infinitely. More specifically, if a path schema is not valid, then
no infinite run can respect it. Moreover, given a path schema, one can decide in polynomial
time whether it is valid. Note that below we deal with path schemas $\aschema$ that are not
necessarily minimal. 

Now, let us consider a (not necessarily minimal) valid path schema $\aschema = 
\aseg_1 \aloop_1^+ \aseg_2 \aloop_2^+ \ldots \aseg_k \aloop_k^{\omega}$  ($k \geq 1$)
obtained from a flat counter system $\asys$ such that all the guards on
transitions are conjunctions of atomic guards of the form
$\sum_{i} \afactor_i \acounter_i \sim \aconstant$ where 
$\afactor_i \in \Zed$, $\aconstant \in \Zed$ and $\sim \in \set{=,\leq,\geq,<,>}$.
Hence, disjunctions are \emph{disallowed} in guards.
The goal of this section (see Lemma~\ref{lemma-constraint-system} below) 
is to characterize the set $\loopsof{\aschema}{\aconf_0}
\subseteq \Nat^{k-1}$ for some \iflong initial \fi  configuration $\aconf_0$ as the set of solutions
of a constraint system. 
For each loop $\aloop_i$, we introduce a 
variable $\avariablebis_i$, 
whence the number of variables of the system/formula is 
precisely $k-1$.
A \defstyle{constraint system} $\aconstraintsystem$ over the set of
variables $\set{\avariablebis_1,
\ldots, \avariablebis_n}$ is a quantifier-free Presburger formula built over 
 $\set{\avariablebis_1, \ldots, \avariablebis_n}$ as a conjunction of atomic constraints
of the form $\sum_{i} \afactor_i \avariablebis_i \sim \aconstant$ where
$\afactor_i, \aconstant \in \Zed$ and  $\sim \in \set{=,\leq,\geq,<,>}$. 
Conjunctions of atomic counter constraints and constraint systems
are essentially the same objects but the  distinction in this place allows
to emphasize the different purposes: guard on counters in operational
models and symbolic representation of sets of tuples.%

Let us build a constraint system $\aconstraintsystem$ defined from $\aschema$ that characterizes 
the set $\loopsof{\aschema}{\aconf_0}$
included in $\Nat^{k-1}$ for some initial configuration $\aconf_0 = 
\pair{\astate_0}{\vec{v_0}}$. 
For all $\alpha \in \interval{1}{k}$ and all $i \in \interval{1}{n}$, we write 
$\preeffect{\aloop_{\alpha}}[i]$ to denote the term below:
$$\vec{v_0}[i] + (\effect{\aseg_1} + \cdots + \effect{\aseg_{\alpha}})[i] +
 \effect{\aloop_1}[i] \avariablebis_1 + \ldots + \effect{\aloop_{\alpha-1}}[i] 
\avariablebis_{\alpha-1}$$ 
It corresponds to the value of 
the counter $i$ just before entering in the loop $\aloop_{\alpha}$.
Similarly, for all $\alpha \in \interval{1}{k}$  and all $i \in \interval{1}{n}$, 
we write $\preeffect{\aseg_{\alpha}}[i]$ to denote
$$\vec{v_0}[i] + (\effect{\aseg_1} + \cdots + \effect{\aseg_{\alpha-1}})[i] +
 \effect{\aloop_1}[i] \avariablebis_1 + \ldots + \effect{\aloop_{\alpha-1}}[i] 
\avariablebis_{\alpha-1}$$
 It corresponds to the value of 
the counter $i$ just before entering in the  segment $\aseg_{\alpha}$.
In this way, for each segment $\aseg$ in $\aschema$ (either a loop or a non-loop segment) and 
each $\beta \in \interval{0}{\length{\aseg}-1}$ the term below
refers to the value of counter $i$ just before entering
for the first time in the $(\beta+1)$th transition of $\aseg$:
$$
\preeffect{\aseg}[i] + \effect{\aseg[0] \cdots \aseg[\beta-1]}[i]
$$
Similarly,  the value of counter $i$ just before entering
for the last time in the $(\beta+1)$th transition of $\aloop_{\alpha}$ is represented
by the term below:
$$
\preeffect{\aseg}[i] +  \effect{\aloop_{\alpha}}[i] (\avariablebis_{\alpha}
- 1) + \effect{\aloop_{\alpha}[0] \cdots \aloop_{\alpha}[\beta-1]}[i]
$$
The set of conjuncts in $\aconstraintsystem$ is defined as follows. Each conjunct
corresponds to a specific constraint in runs  respecting $\aschema$. 

\begin{description}
\itemsep 0 cm 

\item[$\aconstraintsystem_1$:]  Each loop is visited at least once:
      $$
      \avariablebis_1 \geq 1 \wedge \cdots \wedge  \avariablebis_{k-1} \geq 1
      $$

\item[$\aconstraintsystem_2$:] Counter values are non-negative. Let us consider
the following constraints.
     \begin{itemize}
     \item For each segment $\aseg$ and  each 
           $\beta \in \interval{0}{\length{\aseg}-1}$,  
           the value of counter $i$ just before entering
           for the first time in the $(\beta+1)$th transition of $\aseg$ is non-negative:
      $$
\preeffect{\aseg}[i] + \effect{\aseg[0] \cdots \aseg[\beta-1]}[i] \geq 0
      $$
      The segment $\aseg$ can be either a loop or a non-loop segment. 
      \item For each $\alpha \in \interval{1}{k-1}$ and 
             each 
           $\beta \in \interval{0}{\length{\aloop_{\alpha}}-1}$,  
           the value of counter $i$ just before entering
           for the last time in the $(\beta+1)$th transition of 
          $\aloop_{\alpha}$ 
          is non-negative:
          $$
          \preeffect{\aloop_{\alpha}}[i] +  
          \effect{\aloop_{\alpha}}[i] (\avariablebis_{\alpha}
          - 1) + \effect{\aloop_{\alpha}[0] \cdots \aloop_{\alpha}[\beta-1]}[i]
          \geq 0 
         $$
     \end{itemize}
     Convexity guarantees that this is sufficient for preserving non-negativity.
\item[$\aconstraintsystem_3$:] Counter values should satisfy the guards
the first time when a transition is visited. 
For each segment $\aseg$ in $\aschema$, 
       each $\beta \in \interval{0}{\length{\aseg}-1}$ 
      and each atomic guard 
      $\sum_i \afactor_i \avariable_i \sim \aconstant$ occurring in $\guard{\aseg(\beta)}$,   
      we add the atomic 
      constraint:
$$
\sum_i \afactor_i (\preeffect{\aseg}[i] + 
                   \effect{\aseg[0] \cdots 
\aseg[\beta-1]}[i]) \sim \aconstant
$$


\item[$\aconstraintsystem_4$:]
Counter values should satisfy the guards
the last time when a transition is visited. This applies to loops only. 
For each $\alpha \in \interval{1}{k-1}$,
       each $\beta \in \interval{0}{\length{\aloop_{\alpha}}-1}$ 
      and each atomic guard 
      $\sum_i \afactor_i \avariable_i \sim \aconstant$ occurring in $\guard{\aloop_{\alpha}(\beta)}$,   
      we add the atomic 
      constraint:
$$
\sum_i \afactor_i (\preeffect{\aloop_{\alpha}}[i] + 
\effect{\aloop_{\alpha}}[i] (\avariablebis_{\alpha}
          - 1) +
                   \effect{\aloop_{\alpha}[0] \cdots 
\aloop_{\alpha}[\beta-1]}[i]) \sim b
$$
No condition is needed for the last loop since the path schema $\aschema$ is valid. 

\end{description}

Now, let us bound the number of equalities or inequalities above.
To do so, we write
$N_1$ to denote the number of atomic guards in $\aschema$. 
\begin{itemize}
\itemsep 0 cm
\item The number of conjuncts in $\aconstraintsystem_1$ is $k$. 
\item 
      The number of conjuncts in $\aconstraintsystem_2$ is bounded
      by
      $$
        \length{\aschema} \times n +
        \length{\aschema} \times n = 2 n \times  \length{\aschema}.
      $$
\item The number of conjuncts in $\aconstraintsystem_3$ [resp. 
      $\aconstraintsystem_4$] is bounded
      by
      $
      \length{\aschema} \times N_1
      $.
\end{itemize}
So, the number of conjuncts in $\aconstraintsystem$ is bounded by
$2 \times \length{\aschema} \times (1 + n + N_1) \leq 2 \times \length{\aschema} \times n (1 + N_1)$. 
Since $n, 1+N_1 \leq \sizeof{\aschema}$,
we get that this number is bounded by $\length{\aschema} \times 2 \times \size{\aschema}^2$.

Let  $K$ be
the  maximal absolute value of constants occurring in 
either in $\aschema$ or in $\vec{v_0}$.
Let us bound the maximal absolute value of constants in $\aconstraintsystem$.
To do so, we start by a few observations.
\begin{itemize}
\item A  path segment $\aseg$ has at most $\length{\aschema}$ transitions and therefore
      the maximal absolute value occurring in $\effect{\aseg}$ is at most
     $K \times \length{\aschema}$. 
\item The maximal absolute value occurring in $\preeffect{\aseg}$ is at most
      $K + K \times \length{\aschema} = K (1 + \length{\aschema})$. The first occurence of $K$ comes from the counter
      values in the initial configuration. 
\end{itemize}
Consequently, we can make the following conclusions.
\begin{itemize}
\item The maximal absolute values of constants in $\aconstraintsystem_1$ is 1.
\item The maximal absolute values of constants in the first part of $\aconstraintsystem_2$ is  bounded by
     $K (1 + \length{\aschema}) + K \length{\aschema} \leq (K+1) (\length{\aschema}+1)$.
\item The maximal absolute values of constants in the second part of $\aconstraintsystem_2$ is  bounded by
     $K (1 + \length{\aschema}) + K \length{\aschema} + K \length{\aschema} \leq 2 (K+1) (\length{\aschema}+1)$.
     So, the maximal absolute values of constants in  $\aconstraintsystem_2$ is  bounded by
     $2 (K+1) (\length{\aschema}+1)$.
\item  The maximal absolute values of constants in  $\aconstraintsystem_3$ or  $\aconstraintsystem_4$ 
      is  bounded by
     $n \times K \times  2 (K+1) (\length{\aschema}+1)  + K$. The last occurrence of $K$ is due to the constant 
     $\aconstant$ in the atomic constraint. 
\end{itemize}

Consequently, the maximal absolute value of constants in $\aconstraintsystem$
is bounded by 
$2n \times  K (K+2) \times  (\length{\aschema}+1)$.
When $\aschema$ is a minimal path schema, note that $\length{\aschema} \leq
2 \times \card{\edges} \leq 2 \times \size{\asys}$ and $k \leq \card{\states} \leq \size{\asys}$. 


\iflong
\begin{lemma} \label{lemma-constraint-system}
Let $\asys =  \tuple{\states,\counters_n,\edges,\alabelling}$ be a flat counter system
without disjunctions in guards, $\aschema = 
\aseg_1 \aloop_1^+ \aseg_2 \aloop_2^+ \ldots \aseg_k \aloop_k^{\omega}$
be one of its valid path schemas and $\aconf_0$ be a configuration. 
One can compute a constraint system $\aconstraintsystem$ such that
\begin{itemize}
\itemsep 0 cm 
\item the set of solutions of 
$\aconstraintsystem$ is equal to  $\loopsof{\aschema}{\aconf_0}$,
\item $\aconstraintsystem$ has $k-1$ variables, 
\item $\aconstraintsystem$ has at most $\length{\aschema} \times 2 \times \size{\aschema}^2$ conjuncts,
\item the greatest absolute value from constants in $\aconstraintsystem$  is bounded
by $2n \times  K (K+2) \times  (\length{\aschema}+1)$.
\end{itemize}
\end{lemma}
\else
\begin{lemma} \label{lemma-constraint-system}
Let $\asys =  \tuple{\states,\counters_n,\edges,\alabelling}$ be a flat counter system
without disjunctions in guards, $\aschema$
be a valid path schema and $\aconf_0$ be a configuration. 
One can compute in polynomial time a constraint system $\aconstraintsystem$ such that
 the set of solutions of 
$\aconstraintsystem$ is equal to  $\loopsof{\aschema}{\aconf_0}$,
$\aconstraintsystem$ has $\nbloops{\aschema}-1$ variables, 
$\aconstraintsystem$ has at most $\length{\aschema} \times 2 \times \size{\asys}^2$ conjuncts
and the greatest absolute value from constants in $\aconstraintsystem$  is bounded
by $n \times \nbloops{\aschema} \times K^4 \times \length{\aschema}^3$ where $K$ is the greatest absolute
value for constants occurring in $\asys$. 
\end{lemma}
\fi

\ifshort
\else 

\begin{proof} 
\iflong
The constraint system $\aconstraintsystem$ is the one built above. 
\else

\fi 

($\star$)
Let $\arun = \pair{\astate_0}{\vec{v_0}} \pair{\astate_1}{\vec{v_1}} 
\pair{\astate_2}{\vec{v_2}} \cdots$ be an infinite run respecting the path schema $\aschema$
with $\aconf_0 = \pair{\astate_0}{\vec{v_0}}$. 
We write $V: \set{\avariablebis_1, \ldots, \avariablebis_{k-1}} \rightarrow \Nat$ to denote
the valuation such that for every $\alpha \in \interval{1}{k-1}$, we
have $V(\avariablebis_{\alpha}) =  \loopsof{\aschema}{\arun}[\alpha]$. 
$V$ is extended naturally to terms built over variables in 
$\set{\avariablebis_1, \ldots, \avariablebis_{k-1}}$, the range becoming $\Zed$. 
Let us check that $V \models \aconstraintsystem$. 
\begin{enumerate}
\itemsep 0 cm 
\item Since $\arun$ respects $\aschema$, 
each loop $\aloop_i$ is visited at least once and therefore $V \models \aconstraintsystem_1$.
\item  We have seen that the value below
$$
V(\preeffect{\aseg}[i] + \effect{\aseg[0] \cdots \aseg[\beta-1]}[i])
$$
is equal to the value of counter $i$ just before entering
           for the first time in the $(\beta+1)$th transition of $\aseg$.
Similarly,  the value below
$$
V(\preeffect{\aloop_{\alpha}}[i] +  
          \effect{\aloop_{\alpha}}[i] (\avariablebis_{\alpha}
          - 1) + \effect{\aloop_{\alpha}[0] \cdots \aloop_{\alpha}[\beta-1]}[i])
$$
is equal to the value of counter $i$  before entering
           for the last time in the $(\beta+1)$th transition of 
          $\aloop_{\alpha}$.
Since $\arun$ is a run, these values are non-negative, whence
$V \models \aconstraintsystem_2$. 
\item Since $\arun$ is a run, whenever a transition is fired, all its guards are satisfied.
Hence, for each segment $\aseg$ in $\aschema$, 
       for each $\beta \in \interval{0}{\length{\aseg}-1}$ 
      and each atomic guard 
      $\sum_i \afactor_i \avariable_i \sim b \in \guard{\aseg(j)}$,
we have 
$$
\sum_i \afactor_i V(\preeffect{\aseg}[i] + 
                   \effect{\aseg[0] \cdots 
\aseg[\beta-1]}[i]) \sim b
$$
Similarly,
for each $\alpha \in \interval{1}{k-1}$,
each $\beta \in \interval{0}{\length{\aloop_{\alpha}}-1}$ 
      and each atomic guard 
      $\sum_i \afactor_i \avariable_i \sim b \in \guard{\aloop_{\alpha}(\beta)}$,   
      we have
$$
\sum_i \afactor_i V(\preeffect{\aloop_{\alpha}}[i] + 
\effect{\aloop_{\alpha}}[i] (\avariablebis_{\alpha}
          - 1) +
                   \effect{\aloop_{\alpha}[0] \cdots 
\aloop_{\alpha}[\beta-1]}[i]) \sim b
$$
Consequently, $V \models \aconstraintsystem_3 \wedge \aconstraintsystem_4$.
\end{enumerate}

($\star \star$) It remains to show the property in the other direction. 

Let  $V: \set{\avariablebis_1, \ldots, \avariablebis_{k-1}} \rightarrow \Nat$
be a solution of $\aconstraintsystem$. Let $$\aword = 
\aseg_1 \aloop_1^{V(\avariablebis_1)} \cdots \aseg_{k-1} \aloop_1^{V(\avariablebis_{k-1})} 
 \aseg_{k} \aloop_k^{\omega} \in \edges^{\omega}$$ and
let us build an $\omega$-sequence 
$\arun' = \pair{\astate_0}{\vec{x_0}}
\pair{\astate_1}{\vec{x_1}} \pair{\astate_2}{\vec{x_2}} \cdots \in (\states \times \Zed^n)^{\omega}$,
that will be later shown to be an infinite run respecting the path schema $\aschema$
with $\aconf_0 = \pair{\astate_0}{\vec{v_0}}$. Here is how $\arun'$ is defined (note that the definition does not assume that $\arun'$ needs to be a run):
\begin{itemize}
\itemsep 0 cm 
\item For every $i \geq 0$, $\astate_i \egdef \source{\aword(i)}$,
\item $\vec{x_0} \egdef \vec{v_0}$ and for every $i \geq 1$, we have
      $\vec{x_i} \egdef \vec{x_{i-1}} + \update{\aword(i)}$. 
\end{itemize}
In order to show that $\arun'$ is an infinite run respecting $\aschema$, we have to check three
main properties.
\begin{enumerate}
\item Since $V \models \aconstraintsystem_2$,  for each segment $\aseg$ in $\aschema$ 
      and  each  $\beta \in \interval{0}{\length{\aseg}-1}$,  
           counter values  just before entering
           for the first time in the $(\beta+1)$th transition of $\aseg$ are non-negative.
      Moreover, for each $\alpha \in \interval{1}{k-1}$ and 
             each 
           $\beta \in \interval{0}{\length{\aloop_{\alpha}}-1}$,  
           counter values  just before entering
           for the last time in the $(\beta+1)$th transition of 
          $\aloop_{\alpha}$ 
          are non-negative too. We have also to guarantee that for 
          $j \in \interval{2}{V(\avariablebis_{\alpha})-1}$, counter values just before
          entering for the $j$th time in  the $(\beta+1)$th transition of 
          $\aloop_{\alpha}$ 
          are non-negative. This is a consequence of the fact that
         if $\vec{z}, \vec{z} + V(\avariablebis_{\alpha}) \effect{\aloop_{\alpha}} \geq 0$, then
         for 
          $j \in \interval{2}{V(\avariablebis_{\alpha})-1}$, we have 
          $\vec{z} + j \times \effect{\aloop_{\alpha}} \geq 0$ (convexity). 
         Consequently, for $i \geq 0$, we have $\vec{x_i} \geq \vec{0}$. 
\item Similarly,  counter values should satisfy the guards for each fired transition. 
Since  $V \models \aconstraintsystem_3$, for each segment $\aseg$ in $\aschema$, 
each $\beta \in \interval{0}{\length{\aseg}-1}$
 and each atomic guard 
 $\sum_i \afactor_i \avariable_i \sim b \in \guard{\aseg(j)}$, 
counter values satisfy it 
the first time the transition is visited. 
Moreover, since  $V \models \aconstraintsystem_3$, for each $\alpha \in \interval{1}{k-1}$,
       each $\beta \in \interval{0}{\length{\aloop_{\alpha}}-1}$ 
      and each atomic guard 
      $\sum_i \afactor_i \avariable_i \sim b \in \guard{\aloop_{\alpha}(\beta)}$ occurs, 
counter values satisfy it 
the first time the transition is visited. 
However, we have also to guarantee that for 
          $j \in \interval{2}{V(\avariablebis_{\alpha})-1}$, counter values just before
          entering for the $j$th time in  the $(\beta+1)$th transition of 
          $\aloop_{\alpha}$,
all the guards are satisfied. 
This is a consequence of the fact that
         if $\sum_i \afactor_i \vec{z}[i] \sim b$ and $\sum_i \afactor_i 
         (\vec{z} + V(\avariablebis_{\alpha}) \effect{\aloop_{\alpha}})[i] \sim b$, then
         for 
          $j \in \interval{2}{V(\avariablebis_{\alpha})-1}$, we have 
         $\sum_i \afactor_i 
         (\vec{z} + j \times \effect{\aloop_{\alpha}})[i] \sim b$ (convexity). 
Hence, $\arun'$ is a run starting at $\aconf_0$. 
\item It remains to show that $\arun'$ respects $\aschema$. Since $\arun'$ is a run
(see (1) and (2) above), by construction of $\arun'$, it  respects $\aschema$ thanks to
$V \models \aconstraintsystem_1$. Indeed, by definition, each loop has to be visited
at least once. 
\end{enumerate}
\end{proof}

\fi 

\fi 
\iflong

\section{From One Minimal Schema to Several Schemas}
\label{section-disjunction}


Given a flat  counter system 
$\asys = \tuple{\states,\counters_n,\edges,\alabelling}$, 
a configuration $\aconf_0= \pair{\astate_0}{\vec{v_0}}$ and
a minimal path schema $\aschema$ 
starting from the configuration $\aconf_0$, we build a finite set $\asetbis_{\aschema}$ of path schemas
such that \\
\iflong
\begin{enumerate}
\itemsep 0 cm
\item each path schema in $\asetbis_{\aschema}$ has transitions without disjunctions in guards,
\item existence of a run respecting $\aschema$ is equivalent to the existence
of a path schema in $\asetbis_{\aschema}$ having a run respecting it,
\item each path schema in $\asetbis_{\aschema}$ is obtained from $\aschema$ by unfolding loops so that
the terms in each loop satisfy the same atomic guards. 
\end{enumerate}
\else
1. each path schema in $\asetbis_{\aschema}$ has transitions without disjunctions in guards,\\
2. existence of a run respecting $\aschema$ is equivalent to the existence
of a path schema in $\asetbis_{\aschema}$ having a run respecting it,\\
3.  each path schema in $\asetbis_{\aschema}$ is obtained from $\aschema$ by unfolding loops so that
the terms in each loop satisfy the same atomic guards. 
\fi
\iflong
Moreover, we shall see how the cardinal of $\asetbis_{\aschema}$ is at most exponential in the size of
$\aschema$. Note that each path schema in $\asetbis_{\aschema}$ comes with an implicit counter system
(typically containing exactly the states and transitions occurring in the path schema). 
So, below, 
we  explain how we could get rid of disjunctions.
Note also that disjunctions can be easily eliminated at the cost
of adding new transitions between states but this type of transformation may easily
destroy flatness. That is why, we shall follow a different path.
\else 
Note that disjunctions can be easily eliminated at the cost
of adding new transitions between states but this type of transformation may easily
destroy flatness. Hence, the elimination .
\fi 

\subsection{Term maps}
\label{section-term-maps}
Before defining $\asetbis_{\aschema}$, let us introduce a few definitions.
Let $B$ be a finite non-empty set of integers containing all the constants $\aconstant$ 
occurring in guards
of $\asys$ of the form $\aterm \sim \aconstant$ and $T$ be a finite set of terms containing all the terms $\aterm$
occurring in guards of $\asys$ of the form $\aterm \sim \aconstant$. 
Assuming that $B = \set{\aconstant_1, \ldots,\aconstant_m}$ with $\aconstant_1 < \cdots < \aconstant_m$,
we write $I$ to denote the finite set of intervals
$I=\{(-\infty,\aconstant_{1}-1],\interval{\aconstant_{1}}{\aconstant_{1}},\interval{\aconstant_{1}+1}{\aconstant_{2}-1},
\interval{\aconstant_{2}}{\aconstant_{2}}, \allowbreak \cdots , 
\interval{\aconstant_{m}}{\aconstant_{m}}, [\aconstant_{m}+1,\infty)\} \setminus \set{\emptyset}$.
Note that $\interval{\aconstant_{j}+1}{\aconstant_{j+1}-1} = \emptyset$ if $\aconstant_{j+1} = 
\aconstant_{j} + 1$. 
Note that $I$ contains at most $2m+1$ intervals and at least
$m+2$ intervals. We consider the natural
linear ordering $\leq$ on intervals in $I$ that respects the standard relation $\leq$ on integers.
A \defstyle{term map} $\atermmap$ is a map $\atermmap: T \rightarrow I$ that abstracts term
values. 

\begin{definition} Given a loop effect $\anupdate \in \Zed^n$, we define the relation $\preceq_{\anupdate}$
such that $\atermmap \preceq_{\anupdate} \atermmap'$ $\equivdef$
for every term $\aterm = \sum_i \afactor_i \avariable_i  \in T$, we have 
\iflong
\begin{itemize}
\itemsep 0 cm
\item $\atermmap(\aterm) \leq \atermmap'(\aterm)$ if $\sum_i \afactor_i \anupdate[i] \geq 0$,
\item $\atermmap(\aterm) \geq \atermmap'(\aterm)$ if $\sum_i \afactor_i \anupdate[i] \leq 0$,
\item $\atermmap(\aterm) = \atermmap'(\aterm)$ if $\sum_i \afactor_i \anupdate[i] = 0$.
\end{itemize}
\else
$\atermmap(\aterm) \leq \atermmap'(\aterm)$ if $\sum_i \afactor_i \anupdate[i] \geq 0$,
$\atermmap(\aterm) \geq \atermmap'(\aterm)$ if $\sum_i \afactor_i \anupdate[i] \leq 0$
and 
$\atermmap(\aterm) = \atermmap'(\aterm)$ if $\sum_i \afactor_i \anupdate[i] = 0$.
\fi 
We write $\atermmap \prec_{\anupdate} \atermmap'$ whenever  $\atermmap \preceq_{\anupdate} \atermmap'$ 
and  $\atermmap \neq \atermmap'$. 
\end{definition}

Sequences of strictly increasing term maps have bounded length. 

\begin{lemma} 
\label{lemma-bounded-map}
Let $\anupdate \in \Zed^n$ and 
$\atermmap_1 \prec_{\anupdate} \atermmap_2 \prec_{\anupdate} \cdots \prec_{\anupdate}  \atermmap_L$.
Then, $L \leq \card{I}   \times  \card{T} \leq 2 \times \card{T} \times \card{B} + \card{T} $. 
\end{lemma}

\iflong

\begin{proof} 
Given a term map $\atermmap$ and a term $\aterm$, 
$\atermmap(\aterm)$ can obviously take one of the
$\card{I}$ values from $I$. 
For each term $\aterm$, 
\begin{description}
\itemsep 0 cm 
\item[(increasing)] either $\atermmap_1(\aterm) \preceq_{\anupdate} \cdots \preceq_{\anupdate} 
\atermmap_L(\aterm)$
\item[(decreasing)] or $\atermmap_L(\aterm) \preceq_{\anupdate} \cdots \preceq_{\anupdate} 
\atermmap_1(\aterm)$.
\end{description}
Also, there are $\card{T}$ number of terms. Hence, the number of different maps that are either decreasing or increasing can be
$\card{T}\times\card{I}$. Again, we know that $\card{I} \leq 2\times\card{B}+1$.
 Hence, $L$, the number of 
different term maps in a sequence which is either
increasing or decreasing, can be at most 
$\card{I} \times \card{T} \leq 2 \times \card{T} \times \card{B} + \card{T}$.
\end{proof}

\fi 

Given a guard $\aguard$ using the syntactic resources from  $T$ and $B$, 
and a term map $\atermmap$, we write
$\atermmap \vdash \aguard$ with the following inductive definition:
\begin{itemize}
\itemsep 0 cm
\item $\atermmap\vdash \aterm = \aconstant$ $\equivdef$ $\atermmap(\aterm)= \interval{\aconstant}{\aconstant}$;
\item $\atermmap\vdash \aterm \leq \aconstant$ $\equivdef$ $\atermmap(\aterm)\subseteq(-\infty,\aconstant]$; 
\item $\atermmap\vdash \aterm \geq \aconstant$ $\equivdef$ $\atermmap(\aterm)\subseteq[\aconstant,+\infty)$,
\item $\atermmap\vdash \aterm < \aconstant$ $\equivdef$ $\atermmap(\aterm)\subseteq(-\infty,\aconstant)$;
\item  $\atermmap\vdash \aterm > \aconstant$ $\equivdef$ $\atermmap(\aterm)\subseteq(\aconstant,+\infty)$,
\item $\atermmap\vdash \aguard_1 \wedge \aguard_2$ $\equivdef$
      $\atermmap\vdash \aguard_1$ and  $\atermmap\vdash \aguard_2$;
\item $\atermmap\vdash \aguard_1 \vee \aguard_2$ $\equivdef$
      $\atermmap\vdash \aguard_1$ or  $\atermmap\vdash \aguard_2$.
\end{itemize}

The relation $\vdash$ is nothing else than the symbolic satisfaction relation
between term values and guards. Since term maps and guards are built over the same sets
of terms and constants, completeness is obtained as stated in
Lemma~\ref{relation-property}(II) below. Furthermore, 
Lemma~\ref{relation-property}(I)  states that the relation $\vdash$ is easy to check.

\begin{lemma} \label{relation-property}
\iflong
\
\begin{itemize}
  \item[(I)] $\atermmap \vdash \aguard$ can be checked in \ptime \ in $\size{\atermmap} + \size{\aguard}$.
  \item[(II)] $\atermmap\vdash\aguard$ iff for all $v:\set{\avariable_1,\avariable_2,
\cdots,\avariable_n}\rightarrow \Nat$ and for all $\aterm\in T$, 
$v(\aterm) \in \atermmap(\aterm)$ implies  $v \models \aguard$.
\end{itemize}
\else
(I) $\atermmap \vdash \aguard$ can be checked in \ptime \ in $\size{\atermmap} + \size{\aguard}$. \\
(II) $\atermmap\vdash\aguard$ iff for all $v:\set{\avariable_1,\avariable_2,
\cdots,\avariable_n}\rightarrow \Nat$,  (for all $\aterm\in T$, 
$v(\aterm) \in \atermmap(\aterm)$) implies  $v \models \aguard$.
\fi 
\end{lemma}
It is worth noting that $\size{\atermmap}$ is in $\mathcal{O}(\card{I} \times \card{T})$. 
\iflong

\begin{proof}
\begin{itemize}
\itemsep 0 cm 
\item[(I)] For the \ptime \ algorithm we follow the following steps. 
First, for each constraint 
$\aterm\sim b$ 
appearing in $\aguard$, we replace 
it either $\top$ (true) or $\perp$ (false) depending whether 
 $\atermmap \vdash \aterm \sim b$ or not. 
After replacing all constraints, we are left with a positive Boolean formula
whose atomic formulae are either $\top$ or $\perp$. It can be evaluated in logarithmic space in the size
of the resulting formula (less than $\size{\aguard}$), see e.g.~\cite{Lynch77}.

Note that given a term map $\atermmap$ and a constraint $\aterm\sim b$, checking
 $\atermmap \vdash \aterm \sim b$ amounts to checking the containement of 
interval $\atermmap(\aterm)$ 
in a specified interval depending on $\sim$. This can be achieved by comparing the 
end-points of the intervals, which can be done in polynomial time in 
$\size{\aterm} + \size{\atermmap}$. As the number of constraints is also bounded
by   $\size{\aguard}$, the replacement of atomic constraints can be performed
in polynomial time in $\size{\atermmap} + \size{\aguard}$.
Thus, the procedure completes in time polynomial in  $\size{\atermmap} + \size{\aguard}$.

\item[(II)] Consider that $\atermmap\vdash\aguard$ and some 
$v:\{\avariable_1,\avariable_2,\cdots,\avariable_n\} \rightarrow \Nat$
such that $v(\aterm)$ lies in the interval $\atermmap(\aterm)$ for each term $\aterm\in T$. Now 
we  prove inductively on the structure of $\aguard$ that $v \models \aguard$.
\begin{itemize}
\itemsep 0 cm 
  \item \textbf{Base Case:} As base case we have arithmetical constraints of the guard. Consider the constraint is of the 
form $\aterm\leq b$. Since, $\atermmap\vdash\aguard$, we have that $\atermmap(\aterm)\subseteq(-\infty,b]$. Since, 
$v(\aterm)$ lies in the interval $\atermmap(\aterm)$, $v(\aterm)\in(-\infty,b]$. Note that, in this case 
$v\models\aterm\leq b$. Similarly, for other type of constraints $\aterm\sim b$, observe that if 
$v(\aterm)\in\atermmap(\aterm)$ then $v(\aterm)$ lies in the interval specified in the definition of 
$\vdash$ and thus, $v\models\aterm\sim b$.
  \item \textbf{Inductive step:} The induction step for $\wedge$ and $\vee$, follows easily.
\end{itemize}

On the other hand, consider some valuation $v$ with $v(\aterm)\in\atermmap(\aterm)$ for each $\aterm\in T$ and $v\models\aguard$. Similar to above, we will use inductive argument to show that $\atermmap\vdash\aguard$
\begin{itemize}
  \item \textbf{Base Case:} Again consider arithmetical constraints of the guard. Specifically, we consider 
constraints of the form $\aterm \geq \aconstant$. As $v\models\aterm\geq b$, we know that 
$v(\aterm)\in[\aconstant,+\infty)$. Since, $v(\aterm)\in\atermmap(\aterm)$, we have that, 
$\atermmap(\aterm)\subseteq[\aconstant,+\infty)$. Hence, $\atermmap\vdash\aterm\geq b$. Similarly, 
for constraints of other forms $\aterm\sim b$, $v(\aterm)$ lies in the interval exactly 
specified in the definition of $\vdash$. Thus, $\atermmap\vdash\aterm\sim b$.
  \item \textbf{Inductive step:} Again, the induction step for $\wedge$ and $\vee$ follows easily.
\end{itemize}
\end{itemize} 
\end{proof}

\fi 

A \defstyle{resource} $\aresource$ is a triple $\triple{\aset}{T}{B}$ such that
$\aset$ is a finite set of propositional variables,
$T$ is a finite set of terms and $B$ is a finite set of integers.
Without any loss of generality, we assume that all these sets are non-empty in order
to avoid treatments of (easy) particular cases. 
A formula $\aformula \in \PLTL[\counters]$ is \defstyle{built over} 
$\aresource$ whenever the atomic formulae are of the form
either $\avarprop \in \aset$ or $\aterm \sim b$ with $\aterm \in T$ and $b \in B$.
A \defstyle{footprint} is an abstraction of a model for $\PLTL[\counters]$ 
restricted to elements from the resource $\aresource$. 
More precisely, a fooprint $\afootprint$ is of the form
$\afootprint: \Nat \rightarrow \powerset{\aset} \times I^T$ where
$I$ is the set of intervals built from $B$, whence the first element of $\afootprint(i)$
is a propositional valuation and the second one is a term map.
The satisfiability relation $\models$ involving models or runs can be adapted 
to footprints as follows where formulae and footprints are obtained
from the same resource $\aresource$:
\begin{itemize}
\itemsep 0 cm
\iflong 
\item $\afootprint, i \symbmodels \avarprop$ $\equivdef$
      $\avarprop \in \pi_1(\afootprint(i))$,
\item $\afootprint, i \symbmodels \aterm \geq \aconstant$  $\equivdef$
      $\pi_2(\afootprint(i))(\aterm) \subseteq [\aconstant, + \infty)$,
\else
\item $\afootprint, i \symbmodels \avarprop$ $\equivdef$
      $\avarprop \in \pi_1(\afootprint(i))$;
      $\afootprint, i \symbmodels \aterm \geq \aconstant$  $\equivdef$
      $\pi_2(\afootprint(i))(\aterm) \subseteq [\aconstant, + \infty)$,
\fi 
\item $\afootprint, i \symbmodels \aterm \leq \aconstant$  $\equivdef$
      $\pi_2(\afootprint(i))(\aterm) \subseteq (- \infty, \aconstant]$,
\item $\afootprint, i \symbmodels \mynext \aformula$ $\equivdef$ 
       $\afootprint, i+1 \symbmodels  \aformula$,
\item $\afootprint, i \symbmodels \aformula \until \aformulabis$ $\equivdef$ 
      there is $j \geq i$ such that  $\afootprint, j \symbmodels  \aformulabis$
      and for $j' \in \interval{i}{j-1}$, we have 
       $\afootprint, j' \symbmodels  \aformula$.
\end{itemize}
\iflong
We omit the clauses for Boolean connectives, past-time operators and
other arithmetical constraints since their definitions are as expected.
Actually, $\symbmodels$ is exactly the satisfaction relation for plain Past LTL 
when arithmetical constraints are understood as abstract propositions. 
\else
We omit the other obvious clauses. 
$\symbmodels$ is the satisfaction relation for  LTL with past
when arithmetical constraints are understood as abstract propositions.
\fi 

\begin{definition} 
\label{definition-footprint}
Let  $\aresource = \triple{\aset}{T}{B}$ be a resource
and $\arun = \pair{\astate_0}{\vec{v_0}}, \pair{\astate_1}{\vec{v_1}} \cdots$ 
be an infinite run of $\asys$. The \defstyle{footprint} of $\arun$
with respect to $\aresource$ is the footprint $\footprint{\arun}$ 
such that for $i \geq 0$, we have $\footprint{\arun}(i) 
\egdef \pair{\alabelling(\astate_i) \cap \aset}{\atermmap_i}$
where  
for every term $\aterm = \sum_j \afactor_j \avariable_j  \in T$, we have 
 $\sum_j \afactor_j \vec{v_i}[j] \in \atermmap_i(\aterm)$. 
\end{definition}

Note that $\sum_j \afactor_j \vec{v_i}[j]$ belongs to a unique element of $I$ since
$I$ is a partition of $\Zed$. Hence, Definition~\ref{definition-footprint} makes sense. 
Lemma~\ref{lemma-footprint-equivalence} below roughly states that 
satisfaction of a formula on a run can be checked symbolically from the footprint
(this turns out to be useful for the correctness of forthcoming 
Algorithm~\ref{algorithm-main}). 

\begin{lemma} 
\label{lemma-footprint-run}
Let $\aresource = \triple{\aset}{T}{B}$ be a resource,
$\arun = \pair{\astate_0}{\vec{v_0}}, \pair{\astate_1}{\vec{v_1}} \cdots$  be an infinite run,
$i \geq 0$ be a position 
and $\aformula$ be a formula in $\PLTL[\counters]$ built over $\aresource$.
Then $\arun, i \models \aformula$ iff $\footprint{\arun}, i \symbmodels \aformula$. 
\end{lemma}

\iflong

\begin{proof} The proof is by structural induction.
\begin{itemize}
\itemsep 0 cm
\item Base Case 1 ($\avarprop \in \aset$): we have the following equivalences:
\begin{itemize}
\itemsep 0 cm
\item $\arun, i \models \avarprop$,
\item $\avarprop \in \alabelling(\astate_i)$ (by definition of $\models$),
\item $\avarprop \in  \pi_1(\afootprint(i))$ (by definition of $\footprint{\arun}$),
\item $\footprint{\arun}, i \symbmodels \avarprop$  (by definition of $\symbmodels$).
\end{itemize}

\item Base Case 2 ($\sum_j \afactor_j \avariable_j \leq \aconstant$ with 
      $\sum_j \afactor_j \avariable_j \in T$ and $\aconstant \in B$):  we have the following equivalences:
\begin{itemize}
\itemsep 0 cm
\item $\arun, i \models \sum_j \afactor_j \avariable_j \leq \aconstant$,
\item  $\sum_j \afactor_j \vec{v_i}[j] \leq \aconstant$ (by definition of $\models$),
\item $\pi_2(\afootprint(i))(\atermmap_i( \sum_j \afactor_j \avariable_j))
      \subseteq (- \infty, \aconstant]$ (by definition of  $\footprint{\arun}$),
\item $\footprint{\arun}, i \symbmodels \sum_j \afactor_j \avariable_j \leq \aconstant$  
      (by definition of $\symbmodels$).
\end{itemize}
The base cases for the other arithmetical constraints can be shown similarly. 
\item For the induction
step, by way of example we deal with the case $\aformula=\mynext\aformulabis$ (the cases
for the Boolean operators or for the other temporal operators are analogous).
We have the following equivalences:
\begin{itemize}
\itemsep 0 cm
\item $\arun, i \models \mynext \aformulabis$,
\item  $\arun, i+1 \models  \aformulabis$  (by definition of $\models$),
\item  $\footprint{\arun}, i+1 \symbmodels  \aformulabis$  (by induction hypothesis),
\item  $\footprint{\arun}, i \symbmodels  \mynext \aformulabis$  (by definition $\symbmodels$).
\end{itemize}
\end{itemize}

\end{proof}

\fi 

As a corollary, we obtain.

\begin{lemma} 
\label{lemma-footprint-equivalence}
Let $\aresource = \triple{\aset}{T}{B}$ be a resource and 
$\arun$ and $\arun'$ be two infinite runs with identical
footprints with respect to $\aresource$. 
For all formulae $\aformula$ built over $\aresource$ and positions
$i \geq 0$, we have $\arun, i \models \aformula$ iff $\arun', i \models \aformula$. 
\end{lemma}

Given a minimal path schema 
$\aschema = \aseg_1 \aloop_1^+ \aseg_2 \aloop_2^+ \ldots \aseg_k \aloop_k^{\omega}$ and a 
run $\arun$ respecting $\aschema$, 
$\footprint{\arun}$ (with respect to  a resource  $\aresource = \triple{\aset}{T}{B}$)  
is an ultimately periodic word that can be written of 
the form $\aword  \cdot \awordbis^{\omega}$ 
where $\length{\awordbis} = \length{\aloop_k}$.
Note that $\aschema$ denotes also a class of ultimately periodic words but over a different alphabet. 

\subsection{Unfolding}

Let $\aresource = \triple{\aset}{T}{B}$ be a resource and 
$\aschema = \aseg_1 \aloop_1^+ \aseg_2 \aloop_2^+ \ldots \aseg_k \aloop_k^{\omega}$
be a minimal path schema.
In order to define the set of path schemas $\asetbis_{\aschema}$, we need to define
other objects such as guards (from the set $\guards^{\star}(T,B,U)$ defined below),
control states (from the set  $\states' = \states\times I^{T}$), 
transitions 
from the set $\edges'$ defined below  and using $\guards^{\star}(T,B,U)$, $\states'$ and other objects
from $\aschema$. 
 
Let $\edges_{\aschema}$ be the set of 
transitions occurring in $\aschema$ and $\states'$ be $\states\times I^{T}$. 
Given $\aterm = \sum_j \afactor_j \avariable_j \in T$, $\anupdate \in \Zed^n$ and a term map
$\atermmap$, 
we write $\aformulabis(\aterm, \anupdate,\atermmap(\aterm))$ to denote the formula below
($\aconstant, \aconstant' \in B$):
\begin{itemize}
\itemsep 0 cm
\iflong
\item $\aformulabis(\aterm, \anupdate,(-\infty,\aconstant]) \egdef
       \sum_j \afactor_j (\avariable_j + \anupdate(j)) \leq \aconstant$,
\item $\aformulabis(\aterm, \anupdate,[\aconstant,+\infty)) \egdef
       \sum_j \afactor_j (\avariable_j + \anupdate(j)) \geq \aconstant$,
\else
\item $\aformulabis(\aterm, \anupdate,(-\infty,\aconstant]) \egdef
       \sum_j \afactor_j (\avariable_j + \anupdate(j)) \leq \aconstant$; 
      $\aformulabis(\aterm, \anupdate,[\aconstant,+\infty)) \egdef
       \sum_j \afactor_j (\avariable_j + \anupdate(j)) \geq \aconstant$,
\fi 
\item $\aformulabis(\aterm, \anupdate,\interval{\aconstant}{\aconstant'}) \egdef
       ((\sum_j \afactor_j (\avariable_j + \anupdate(j)) \leq \aconstant') \wedge 
       ((\sum_j \afactor_j (\avariable_j + \anupdate(j)) \geq \aconstant)$.  
\end{itemize}

The formulae of the form $\aformulabis(\aterm, \anupdate,int)$, where $int \in I$ have been designed to satisfy
the property below.

\begin{lemma} 
\label{guard-constraint}
Let $v:\set{\acounter_1, \ldots, \acounter_n} \rightarrow \Nat$ and
$v':\set{\acounter_1, \ldots, \acounter_n} \rightarrow \Nat$ be such that for every
$i \in \interval{1}{n}$, $v'(\acounter_i) = v(\acounter_i) + \anupdate(i)$. 
For every interval $int \in I$, for every term $\aterm \in T$, 
$v \models 
\aformulabis(\aterm, \anupdate,int)$ iff $v'(\aterm) \in int$. 
\end{lemma}

The proof is by an easy verification. 

We write $\guards^{\star}(T,B,U)$ to denote the set of guards of the form 
 $\aformulabis(\aterm, \anupdate,\atermmap(\aterm))$ where
$\aterm \in T$, $U$ is the finite set of updates from $\aschema$
and $\atermmap: T \rightarrow I$. 
Each guard in  $\guards^{\star}(T,B,U)$
is of linear size in the size of $\aschema$. 

We define $\edges'$ as a finite  subset of $\states' \times \edges_{\aschema} \times 
\guards^{\star}(T,B,U)
\times U \times \states'$ such that
for every $\pair{\astate}{\atermmap} \step{\atransition,\pair{\aguard_{\atermmap'}}{\anupdate}} 
\pair{\astate'}{\atermmap'} \in \edges'$, the conditions below are satisfied:
\iflong 
\begin{itemize}
\itemsep 0 cm
\item $\astate=\source{\atransition}$ and $\astate'=\target{\atransition}$, 
\item $\aguard_{\atermmap'}$  is a guard that states that after the update $\anupdate$,
 for each $\aterm \in T$, 
 its value belongs to $\atermmap'(\aterm)$.  $\aguard_{\atermmap'}$ is equal to 
$
\bigwedge_{\aterm \in T} \aformulabis(\aterm, \anupdate,\atermmap(\aterm))$
   \item Term values belong to intervals that make true $\guard{\atransition}$, i.e. 
         $\atermmap \vdash \guards(\atransition)$.
\item $\anupdate = \update{\atransition}$. 
\end{itemize}
\else
\begin{itemize}
\itemsep 0 cm
\item $\astate=\source{\atransition}$, $\astate'=\target{\atransition}$ and  
      $\anupdate = \update{\atransition}$,
\item $\aguard_{\atermmap'} \egdef \bigwedge_{\aterm \in T} \aformulabis(\aterm, \anupdate,\atermmap(\aterm))$  
      is a guard that states that after the update $\anupdate$,
      for each $\aterm \in T$, its value belongs to $\atermmap'(\aterm)$,
\item Term values belong to intervals that make true $\guards(\atransition)$, i.e. 
         $\atermmap \vdash \guards(\atransition)$.
\end{itemize}
\fi 

We extend the definition of $\source{\atransition}$ to $\atransition'=\pair{\astate}{\atermmap} \step{\atransition,\pair{\aguard_{\atermmap'}}{\anupdate}}\pair{\astate'}{\atermmap'} \in \edges'$. We define $\source{\atransition'}=\pair{\astate}{\atermmap}$ and $\target{\atransition'}=\pair{\astate'}{\atermmap'}$. Similarly, for a finite word $\aword\in (\edges')^*$, we define $\source{\aword}=\source{\aword(1)}$ and $\target{\aword}=\target{\aword(\length{\aword})}$.

We define skeletons as slight variants of path schemas in $\asetbis_{\aschema}$ with slight differences
explained below. 
A \defstyle{skeleton} (compatible with $\aschema$ and $\pair{\astate_0}{\vec{v_0}}$) 
$\askeleton$, say  $\pair{\astate_1}{\atermmap_1} \step{\atransition_1,\pair{\aguard_{\atermmap'}^1}{\anupdate_1}} 
\pair{\astate_2}{\atermmap_2} \step{\atransition_2,\pair{\aguard_{\atermmap'}^2}{\anupdate_2}} 
\pair{\astate_3}{\atermmap_3} \cdots \step{\atransition_K,\pair{\aguard_{\atermmap'}^K}{\anupdate_K}}  
\pair{\astate_{K+1}}{\atermmap_{K+1}}$, 
 is a finite word over $\edges'$ such that
\begin{description}
\itemsep 0 cm
\item[(init)] For every term $\aterm = \sum_j \afactor_j \avariable_j \in T$,
we have $\sum_j \afactor_j \vec{v_0}[j] \in \atermmap_1(\aterm)$ where $\vec{v_0}$ is the initial vector. 
\item[(schema)] Let $\amap: (\edges')^* \rightarrow \edges^*$ be the map
      such that $\amap(\varepsilon) = \varepsilon$,
      $\amap(\aword \cdot \aword') = \amap(\aword) \cdot \amap(\aword')$
      and $\amap(\pair{\astate}{\atermmap} \step{\atransition,\pair{\aguard_{\atermmap'}}{\anupdate}} 
\pair{\astate'}{\atermmap'}) = \atransition$.
      We require that $\amap(\askeleton) \in
      \aseg_1 \aloop_1^+ \aseg_2 \aloop_2^+ \ldots \aseg_k \aloop_k^+$.
\item[(minimality)] For every factor
$\aword =  \pair{\astate_I}{\atermmap_I} \step{\atransition_I,\pair{\aguard_{\atermmap'}^I}{\anupdate_I}} 
\pair{\astate_{I+1}}{\atermmap_{I+1}} \cdots$ \allowbreak 
$\step{\atransition_{J-1},\pair{\aguard_{\atermmap'}^{J-1}}{\anupdate_{J-1}}}  
\pair{\astate_J}{\atermmap_J}
$ of $\askeleton$ such that $\amap(\aword) = (\aloop)^3$ for some loop $\aloop$ of 
$\aschema$ (therefore $J = I + 3 \times \length{\aloop}$),
there is $\alpha \in \interval{1}{\length{\aloop}}$ such that
$\atermmap_{I+\alpha} \prec_{\effect{\aloop}} \atermmap_{I+\alpha + 2 \times \length{\aloop}}$.
\item[(last-loop)] For the unique suffix $\aword$ of $\askeleton$ of length
$\length{\aloop_k}$, we have $\amap(\aword)=\aloop_k$
and $\source{\aword} = \target{\aword}$. 
\end{description}

\begin{lemma} 
\label{lemma-bounded-skeleton}
For a skeleton $\askeleton$, $\length{\askeleton} \leq (\length{\aseg_1} + \cdots + \length{\aseg_k})
+ 2 \times ( 2 \times \card{T} \times \card{B} + \card{T}) 
\times  (\length{\aloop_1} + \cdots + \length{\aloop_k})$ 
\end{lemma}
\iflong

\begin{proof}
Since $\amap(\askeleton) \in
      \aseg_1 \aloop_1^+ \aseg_2 \aloop_2^+ \ldots \aseg_k \aloop_k^+$, 
let $\amap(\askeleton) = \aseg_1 \aloop_1^{n_1} \aseg_2 \aloop_2^{n_2} \ldots \aseg_k \aloop_k^{n_k}$ for some
$n_1, \ldots, n_k \geq 1$. We have 
$\length{\askeleton} \leq (\length{\aseg_1} + \cdots + \length{\aseg_k})
+ max(n_i) \times  (\length{\aloop_1} + \cdots + \length{\aloop_k})$. It remains to bound
the values  $n_1, \ldots, n_k$. 
For each factor $\aword$ of $\askeleton$ such that
$\amap(\aword)=(l_i)^{n_i}$ with $i \in \interval{1}{k}$,  
by the $\textbf{(minimality)}$ condition and  Lemma~\ref{lemma-bounded-map},
we conclude that $n_i \leq  2 \times ( 2 \times \card{T} \times \card{B} + \card{T})$.
Consequently, $\length{\askeleton} \leq (\length{\aseg_1} + \cdots + \length{\aseg_k})
+ 2 \times ( 2 \times \card{T} \times \card{B} + \card{T}) 
\times  (\length{\aloop_1} + \cdots + \length{\aloop_k})$.
\end{proof}

\fi 
From skeletons, we shall define path schemas built over the alphabet
$\states' \times 
\guards^{\star}(T,B,U)
\times U \times \states'$  (transitions are not anymore formally labelled by elements in $\edges_{\aschema}$;
sometimes we keep these labels for convenience). 
As for the definition of $\amap$, let $\nodistrans$ be a finite subset of $(\states' \times
\guards^{\star}(T,B,U) \times U \times \states')$ and 
let $\amapbis: (\edges')^* \rightarrow (\nodistrans)^*$ be the map
      such that $\amapbis(\varepsilon) = \varepsilon$,
      $\amapbis(\aword \cdot \aword') = \amapbis(\aword) \cdot \amapbis(\aword')$
      and $\amapbis(\pair{\astate}{\atermmap} \step{\atransition,\pair{\aguard_{\atermmap'}}{\anupdate}} 
\pair{\astate'}{\atermmap'}) = \pair{\astate}{\atermmap} \step{\pair{\aguard_{\atermmap'}}{\anupdate}} 
\pair{\astate'}{\atermmap'}$. 
This time, elements of $\edges_{\aschema}$ are removed instead of being kept as for $\amap$. 
Given a skeleton $\askeleton$, we shall define
a path schema $\aschema_{\askeleton} = \aseg'_1 (\aloop'_1)^+ \aseg'_2 (\aloop'_2)^+ \ldots \aseg'_{k'}(\aloop'_{k'})^{\omega}$
such that 
$
\amapbis(\askeleton) 
= 
\aseg'_1  \aloop'_1 \aseg'_2 \aloop'_2  \ldots \aseg'_{k'} \aloop'_{k'}
$. Hence, skeletons slightly differ from the path schemas. 
It remains to specify how the loops in  $\aschema_{\askeleton}$ are identified. 

Every factor
$\aword =  \pair{\astate_I}{\atermmap_I} \step{\atransition_I,\pair{\aguard_{\atermmap'}^I}{\anupdate_I}} 
\pair{\astate_{I+1}}{\atermmap_{I+1}} \cdots \step{\atransition_{J-1},\pair{\aguard_{\atermmap'}^{J-1}}{\anupdate_{J-1}}}  
\pair{\astate_J}{\atermmap_J}
$ 
of $\askeleton$ such that 
\begin{enumerate}
\itemsep 0 cm
\item $\amap(\aword) = \aloop$ for some loop $\aloop$ of 
$\aschema$,
\item $\aword$ is not the suffix of $\askeleton$ of length $\length{\aloop_k}$,
\item the sequence of the $\length{\aloop}$ next elements after $\aword$  is also equal to $\aword$,
\end{enumerate}
 is replaced by $(\amapbis(\aword))^+$.
Finally, $\aloop'_{k'}$ is equal to $\amapbis(\aword)$ where $\aword$ is the unique suffix of $\askeleton$
of length $\length{\aloop_k}$. Note that the path schema $\aschema_{\askeleton}$ is unique by the condition ({\bf minimality}).
Indeed, there is no factor of $\askeleton$ of the form $\aword^3$ such that 
$\amap(\aword) = \aloop$ for some loop $\aloop$ of 
$\aschema$. As far as the labelling function is concerned, 
the labels of $\astate$ and $\pair{\astate}{\atermmap}$
are identical with respect to the set $\aset$, i.e. $\alabelling'(\pair{\astate}{\atermmap}) 
\egdef \alabelling(\astate) \cap \aset$.  
Hence, 
\begin{enumerate}
\itemsep 0 cm
\item  $k' \leq k \times ( 2 \times \card{T} \times \card{K} + \card{T})$, 
\item  $\length{\aschema_{\askeleton}} \leq 
(\length{\aseg_1} + \cdots + \length{\aseg_k})
+ 2 \times ( 2 \times \card{T} \times \card{K} + \card{T}) 
\times  (\length{\aloop_1} + \cdots + \length{\aloop_k})$,
\item $\aschema_{\askeleton}$ has no guards with disjunctions. 
\end{enumerate}

Note that construction of a path schema from a skeleton, cannot be done by simply taking the path segments as 
before and the copies of the loop segments as alternating path and loop segments in the new path schema. 
For example consider this system with one counter $\acounter$, 
$I=\set{(-\infty,-1],\interval{0}{0},\interval{1}{1},\allowbreak \interval{2}{2},[3,\infty)}$ and 
$T=\set{\avariable+1}$.
\begin{figure}
  \begin{center}
    \begin{tikzpicture}[shorten >=1pt,node distance=1cm,auto,>=stealth']
  \node[state] (q_0)                      {$\astate_0$};
  \node[state] (q_1) [right      =of q_0] {$\astate_1$};
  \node[state] (q_2) [right      =of q_1] {$\astate_2$};
  \path[->] (q_0) edge  node   [swap]{$\top,+1$} (\astate_1)
  	    (q_1) edge  node   [swap]{$\top,+1$} (\astate_2)
                  edge [out=125,in=55,loop] node [above]{$\top,+1$} ()
	    (q_2) edge [out=125,in=55,loop] node {} ();
  \node at (2.1,1) {\fbox{\tiny $\geq 1$}};
  \node at (4.1,1) {\fbox{$\omega$}};
  \node at (2,-0.85) {$\aschema$};  
\end{tikzpicture}  
\scalebox{0.95}{
\begin{tikzpicture}[shorten >=1pt,node distance=1.7cm,auto,>=stealth']
  \node[state] (q_00)                      {$\astate_0,[0]$};
  \node[state] (q_11) [right      =of q_00] {$\astate_1,[1]$};
  \node [above=0.3cm of q_11]  {\fbox{$\geq 1$}};
  \node[state] (q_12) [right      =of q_11] {$\astate_1,[2]$};
  \node [above=0.3cm of q_12]  {\fbox{$\geq 1$}};
  \node[state] (q_13) [right      =of q_12] {$\astate_1,[3,\infty)$};
  \node [above=0.3cm of q_13]  {\fbox{$\geq 1$}};
  \node[state] (q_24) [right      =of q_13] {$\astate_2,[3,\infty)$};
   \node [above=0.3cm of q_24]  {\fbox{$\omega$}};
  \path[->] (q_00) edge  node   [swap]{\tiny $\begin{array}{c}\avariable+1=1,+1\end{array}$} (q_11)
  	    (q_11) edge  node   [swap]{\tiny $\begin{array}{c}\avariable+1=2,+1\end{array}$} (q_12)
                  edge [out=125,in=55,loop] node [above]{\avariable+1=1,+1} ()
	    (q_12) edge  node   [swap]{\tiny $\begin{array}{c}\avariable+1\geq 3,+1\end{array}$} (q_13)
                  edge [out=125,in=55,loop] node [above]{\avariable+1=2,+1} ()
	    (q_13) edge  node   [swap]{\tiny $\begin{array}{c}\avariable+1\geq 3,+1\end{array}$} (q_24)
                  edge [out=125,in=55,loop] node [above]{\avariable+1>2,+1} ()
	    (q_24) edge [out=125,in=55,loop] node {} ();
  \node at (4.9,-0.95) {$\aschema'$};
\end{tikzpicture}  
}
\scalebox{0.95}{
\begin{tikzpicture}[shorten >=1pt,node distance=1.7cm,auto,>=stealth']
  \node[state] (q00)                      {$\astate_0,[0]$};
  \node[state] (q11) [right      =of q00] {$\astate_1,[1]$};
  \node[state] (q12) [right      =of q11] {$\astate_1,[2]$};
  \node[state] (q13) [right      =of q12] {$\astate_1,[3,\infty)$};
  \node [above=0.3cm of q_13]  {\fbox{$\geq 1$}};
  \node[state] (q24) [right      =of q13] {$\astate_2,[3,\infty)$};
  \node [above=0.3cm of q_24]  {\fbox{$\omega$}};
  \path[->] (q00) edge  node   [swap]{\tiny $\begin{array}{c}\avariable+1=1,+1\end{array}$} (q11)
  	    (q11) edge  node   [swap]{\tiny $\begin{array}{c}\avariable+1=2,+1\end{array}$} (q12)
	    (q12) edge  node   [swap]{\tiny $\begin{array}{c}\avariable+1\geq 3,+1\end{array}$} (q13)
	    (q13) edge  node   [swap]{\tiny $\begin{array}{c}\avariable+1\geq 3,+1\end{array}$} (q24)
                  edge [out=125,in=55,loop] node [above]{\avariable+1>2,+1} ()
	    (q24) edge [out=125,in=55,loop] node {} ();
  \node at (4.9,-0.95) {$\aschema''$};
\end{tikzpicture}  
}
  \end{center}
  \caption{A path $\aschema$ and two unfolded path schemas in $\asetbis_{\aschema}$}
  \label{fig:unfolding}
\end{figure}

In Figure \ref{fig:unfolding}, 
\begin{itemize}
\itemsep 0 cm  
 \item $p_1(l_1^1)^+l_1^2(l_1^3)^+l_1^4(l_1^5)^+p_2(l_2)^{\omega}$ does not have any valid run 
 respecting it as a path schema, as the loops $l_1^1,l_1^3$ cannot be taken even once in any run.
  \item $p_1l_1^2l_1^4(l_1^5)^+p_2(l_2)^{\omega}$ has a valid run respecting it as a path schema. But, here all the 
  unfoldings of the loop $l_1$ are taken as path segments.
\end{itemize}

We write $\asetbis_{\aschema}$ to denote the set of unfolded path schemas obtained
from $\aschema$ with respect to $\aresource$. A skeleton is \defstyle{compatible} with
$\aschema$ whenever its corresponding path schema belong to  $\asetbis_{\aschema}$.

\begin{lemma}
\label{lemma-path-check}
Checking whether a word $\aword\in(\states' \times \edges \times \guards^{\star}(T,B,U)
\times U \times \states')^*$ is a skeleton compatible with
$\aschema$  and $\pair{\astate_0}{\vec{v_0}}$ assuming that 
$\length{\aword}\leq (\length{\aseg_1} + \ \cdots \ + \length{\aseg_k})
+ 2 ( 2 \times \card{T} \times \card{B} + \card{T}) 
\times  (\length{\aloop_1} + \cdots + \length{\aloop_k})$ can be done 
in polynomial time in the size of $\pair{\astate_0}{\vec{v_0}}$, $\aschema$, $T$ and $B$.
\end{lemma}

\iflong

\begin{proof} 
Let $\aword$ be a word over $\states' \times \edges_{\aschema} \times  \guards^{\star}(T,B,U)
\times U \times \states'$ whose length is bounded by $(\length{\aseg_1} + \cdots + \length{\aseg_k})
+ 2 ( 2 \times \card{T} \times \card{B} + \card{T}) 
\times  (\length{\aloop_1} + \cdots + \length{\aloop_k})$. 
Let $N$ be the sum of the respective sizes of  $\pair{\astate_0}{\vec{v_0}}$, $\aschema$, $T$ and $B$.
Since the length of $\aword$ is bounded, its size is also polynomial in $N$. 

Checking whether an element in $\states' \times \edges_{\aschema} \times  \guards^{\star}(T,B,U)
\times U \times \states'$ belongs to $\edges'$ can be done in
polynomial time in $N$ thanks to  Lemma~\ref{relation-property}(I). 
Hence, checking whether $\aword$ belongs to $(\edges')^*$ can be done in polynomial time in $N$ too
since its length is also polynomial in $N$. It remains to check the conditions for skeletons.
\begin{itemize}
\itemsep 0 cm
\item Condition \textbf{(schema)} can be checked by building first $\amap(\aword)$ (this requires linear time in
      $N$) and then by checking
      whether it belongs to $\aseg_1 \aloop_1^+ \aseg_2 \aloop_2^+ \ldots \aseg_k \aloop_k^+$ (requires also
      linear time in $N$).
\item Condition \textbf{(last-loop)} can be checked by extracting the suffix of $\aword$ of length $\length{\aloop_k}$.
\item Condition \textbf{(minimality)} can be checked by considering all the factors $\aword'$ of $\aword$ (there are less
      than $\length{\aword}^2$ of them) and whenever $\amap(\aword') = \aloop^3$ for some loop $\aloop$, we verify
that the condition is satisfied. All these operations can be done in polynomial time in $N$. 
\item Finally, condition \textbf{(init)} is also easy to check in polynomial time in $N$. 
\end{itemize} 

\end{proof}

\fi 

The main property about  $\asetbis_{\aschema}$ is stated below. 

\iflong
\begin{proposition} \
\label{run-equivalence}
\begin{description}
\itemsep 0 cm 
\item[(I)] Let $\arun$ be an infinite run respecting $\aschema$ and starting at $\pair{\astate_0}{\vec{v_0}}$.
Then, there is  a path schema $\aschema'$ in $\asetbis_{\aschema}$ and
an infinite run $\arun'$  respecting $\aschema'$ such that $\footprint{\arun} = \footprint{\arun'}$. 
\item[(II)] Let $\arun$ be an infinite run respecting $\aschema'$ for some  $\aschema' \in 
\asetbis_{\aschema}$.
Then, there is an infinite run $\arun'$  respecting $\aschema$ such that  
$\footprint{\arun} = \footprint{\arun'}$. 
\end{description}
\end{proposition}
\else
\begin{proposition} \label{run-equivalence}
(I) Let $\arun$ be an infinite run respecting $\aschema$ and starting at $\pair{\astate_0}{\vec{v_0}}$.
Then, there are  a path schema $\aschema'$ in $\asetbis_{\aschema}$ and
an infinite run $\arun'$  respecting $\aschema'$ such that $\footprint{\arun} = \footprint{\arun'}$. 
(II) Let $\arun$ be an infinite run respecting $\aschema'$ for some  $\aschema' \in 
\asetbis_{\aschema}$.
Then, there is an infinite run $\arun'$  respecting $\aschema$ such that  
$\footprint{\arun} = \footprint{\arun'}$. 
\end{proposition}
\fi 

 \iflong

\begin{proof}
\textbf{(I)}
Let $\arun = \pair{\astate_0}{\vec{v_0}} \step{\atransition_0} \pair{\astate_1}{\vec{v_1}}\step{\atransition_1} \cdots$
be an infinite run respecting $\aschema$ with footprint $\footprint{\arun}: \Nat \rightarrow \powerset{AT} \times
I^{T}$. We write $\pair{\asetter_i}{\atermmap_i}$ to denote $\footprint{\arun}(i)$. 
In order to build $\arun'$ and $\aschema'$, first we enrich the structure $\arun$ and then we define
a skeleton from the enriched structure that allows us to define $\aschema'$. 
The run $\arun'$ is then defined from $\arun$ so that the sequences of counter values are identical.
From $\arun$, we consider the  
\iflong infinite sequence below:
$$
\aword = \pair{\astate_0}{\atermmap_0} \step{\atransition_0, \pair{\aguard_{\atermmap_1}}{\update{\atransition_0}}} 
\pair{\astate_1}{\atermmap_1} \step{\atransition_1, \pair{\aguard_{\atermmap_2}}{\update{\atransition_1}}}
\cdots 
$$
\else
infinite sequence 
$
\aword = \pair{\astate_0}{\atermmap_0} \step{\atransition_0, \pair{\aguard_{\atermmap_1}}{\update{\atransition_0}}} 
\pair{\astate_1}{\atermmap_1} \step{\atransition_1, \pair{\aguard_{\atermmap_2}}{\update{\atransition_1}}}
\cdots 
$.
\fi 
It is easy to check that $\aword$ can be viewed as an element of  $(\edges')^{\omega}$ where 
$\edges'$ is defined as a finite subset of $\states' \times \edges_{\aschema} \times 
\guards^{\star}(T,B,U) \times U \times \states'$ where $U$ is the finite set of updates from $\aschema
= \aseg_1 (\aloop_1)^{+} \aseg_2  (\aloop_2)^{+} \cdots (\aloop_{k-1})^{+} \aseg_k  
 (\aloop_{k})^{\omega}$. 
Moreover, we have $\amap(\aword) \in \languageof{\aschema}$, that is $\amap(\aword) =
\aseg_1 (\aloop_1)^{n_1} \aseg_2  (\aloop_2)^{n_2} \cdots (\aloop_{k-1})^{n_{k-1}} \aseg_k  
 (\aloop_{k})^{\omega}$ for some $n_1, \ldots, n_{k-1} \geq 1$. 
From $\aword$, one can build a skeleton $\askeleton$ compatible with $\aschema$ and $\pair{\astate_0}{\vec{v_0}}$. 
$\askeleton$ is formally a subword of $\aword$ such that 
$$\amap(\askeleton) =  
\aseg_1 (\aloop_1)^{n_1'} \aseg_2  (\aloop_2)^{n_2'} \cdots (\aloop_{k-1})^{n_{k-1}'} \aseg_k  
(\aloop_{k})^{n_k'}$$ with  $1 \leq n_i' \leq min(n_i, 2 \times ( 2 \times \card{T} \times \card{B} + \card{T})$
for every $i \in \interval{1}{k-1}$ 
and $1 \leq n_k' \leq   2 \times ( 2 \times \card{T} \times \card{B} + \card{T})$.
We have $\aword = \aword' \cdot \aword_0 \cdot \aword_0 \cdot (\aword_0)^{\omega}$ with 
$\amap(\aword_0) = \aloop_k$.
The skeleton $\askeleton$ is obtained from $\aword' \cdot \aword_0 \cdot \aword_0$ 
by deleting copies of loops as soon as two copies are consecutive.
More precisely, every maximal factor of $\aword' \cdot \aword_0 \cdot \aword_0$ of the form
$(\aword^{\star})^N$ with $N >2$ such that $\amap(\aword^{\star}) = \aloop_i$ for some loop 
$\aloop_i$ of $\aschema$, is replaced by $(\aword^{\star})^2$. 
This type of replacement can be done at most $k \times (2 \times ( 2 \times \card{T} \times \card{B} + \card{T}))$
times. One can check thay $\askeleton$ is indeed a skeleton compatible with  $\aschema$ and $\pair{\astate_0}{\vec{v_0}}$. 
Let us consider that $\askeleton$ can be written as 
$$\pair{\astate_1}{\atermmap_1} \step{\atransition_1,\pair{\aguard_{\atermmap'}^1}{\anupdate_1}} 
\pair{\astate_2}{\atermmap_2} \step{\atransition_2,\pair{\aguard_{\atermmap'}^2}{\anupdate_2}} 
\pair{\astate_3}{\atermmap_3} \cdots \step{\atransition_K,\pair{\aguard_{\atermmap'}^K}{\anupdate_K}}  
\pair{\astate_{K+1}}{\atermmap_{K+1}}
$$
Considering the path schema $\aschema_{\askeleton}$ built from $\askeleton$, one can show that
the sequence $\arun'$ below is an infinite run respecting $\aschema_{\askeleton}$: 
$$
\pair{\pair{\astate_0}{\atermmap_0}}{\vec{v_0}} 
\step{\pair{\aguard_{\atermmap_1}}{\update{\atransition_0}}} 
\pair{\pair{\astate_1}{\atermmap_1}}{\vec{v_1}} 
\step{\pair{\aguard_{\atermmap_2}}{\update{\atransition_1}}}
\pair{\pair{\astate_2}{\atermmap_2}}{\vec{v_2}} 
\cdots 
$$
so that $\footprint{\arun} = \footprint{\arun'}$. When entering in the last loop of $\aschema_{\askeleton}$,
counter values still evolve but the sequence of control states forms a periodic word made of
the last $\length{\aloop_k}$  control states of $\askeleton$.
By construction of $\askeleton$ and 
 $\aschema_{\askeleton}$, it is clear that $\arun$ and $\arun'$ have the same sequences of counter values
(they have actually the same sequences of updates)
and by definition of the labellings, they have also the same sequences of sets of atomic propositions.
It remains to check that $\arun'$ is indeed a run, which amounts to verify that guards are satisfied but this
is guaranteed by the way guards are defined and by the completeness result in Lemma~\ref{relation-property}(II).

\textbf{(II)} Let $\arun$ be some run respecting some $\aschema' \in \asetbis_{\aschema}$ of the form below:
$$
\pair{\pair{\astate_0}{\atermmap_0}}{\vec{v_0}} 
\step{\atransition_0, \pair{\aguard_{\atermmap_1}}{\update{\atransition_0}}} 
\pair{\pair{\astate_1}{\atermmap_1}}{\vec{v_1}} 
\step{\atransition_1, \pair{\aguard_{\atermmap_2}}{\update{\atransition_1}}}
\pair{\pair{\astate_2}{\atermmap_2}}{\vec{v_2}} 
\cdots 
$$
In the above run, we have decorated the steps by transitions from $\aschema$ as $\aschema'$ is defined
from a skeleton in which transitions are decorated by such transitions. 
\iflong
After a tedious verification, one can show that 
$$ 
\arun' = \pair{\astate_0}{\vec{v_0}} \step{\atransition_0} \pair{\astate_1}{\vec{v_1}}\step{\atransition_1} \cdots
$$
is a run respecting $\aschema$ such that $\footprint{\arun} = \footprint{\arun'}$. 
\else
After a tedious verification, one can show that 
$
\arun = \pair{\astate_0}{\vec{v_0}} \step{\atransition_0} 
\pair{\astate_1}{\vec{v_1}}\step{\atransition_1} \cdots
$
is a run respecting $\aschema$ such that $\footprint{\arun} = \footprint{\arun'}$. 
\fi 
Satisfaction of guards is guaranteed by the way $\edges'$ is defined. The fact that $\arun'$ respects
$\aschema$ is even easier to justify since all the path schemas in $\asetbis_{\aschema}$ can be viewed
as specific instances of $\aschema$ that differ in the way the term maps  evolve. Details are omitted. 

\end{proof}

 \fi 

Let $\aschema' = 
\aseg'_1 (\aloop'_1)^+ \aseg'_2 (\aloop'_2)^+ \ldots \aseg'_{k'}(\aloop'_{k'})^{\omega}$ 
be a path schema in $\asetbis_{\aschema}$ and 
$\arun$ be a run
$
\pair{\pair{\astate_0}{\atermmap_0}}{\vec{v_0}} 
\step{\pair{\aguard_{\atermmap_1}}{\update{\atransition_0}}} 
\pair{\pair{\astate_1}{\atermmap_1}}{\vec{v_1}} 
\step{\pair{\aguard_{\atermmap_2}}{\update{\atransition_1}}}
\pair{\pair{\astate_2}{\atermmap_2}}{\vec{v_2}} 
\cdots 
$ respecting $\aschema'$. 
It is easy to show that for $i \geq 0$, we have $\pi_2(\footprint{\arun}(i)) = \atermmap_i$
and $\footprint{\arun}$ is an ultimately periodic word of the form
$\aword \cdot \awordbis^{\omega}$ where
$\length{\awordbis} = \length{\aloop_{k'}'} = 
\length{\aloop_{k}}$ and
$\length{\aword} = 
(\length{\aseg_1'} + \cdots + \length{\aseg_{k'}'}) + 
(\loopsof{\aschema'}{\arun}[1] \times \length{\aloop_1'} + \cdots +  \loopsof{\aschema'}{\arun}[k'-1]
\times \length{\aloop_{k'-1}'} )$. 
As seen previously, we have $\arun, 0 \models \aformula$ iff
$\footprint{\arun}, 0 \symbmodels \aformula$. 
We also define a function $\proj$ which associates to $\aword \in  \nodistrans^{\omega}$ the $\omega$-sequence $\proj(\aword) : \Nat \rightarrow \powerset{\aset} \times I^T$ such that for all $i \in \Nat$, if $\aword(i)=\
\tuple{\tuple{\astate,\atermmap},\aguard,\anupdate,\tuple{\astate',\atermmap'}}$ 
and $\alabelling(\astate) \cap \aset=\labels$ then $\proj(\aword)(i) \egdef \tuple{\labels,\atermmap}$. 
Now, we can state the main theorem about removing disjunction in the guards by unfolding of loops.
It entails the main properties we expect from $\asetbis_{\aschema}$. 

\begin{theorem} 
\label{theorem-main-disjunction}
Given a flat counter system $\asys$, a minimal path schema $\aschema$, a set of terms $T$
including those in $\aschema$, a set of constants $B$ including those in $\aschema$ and
an initial configuration $\pair{\astate_0}{\vec{v_0}}$, 
there is a finite set of path schemas $\asetbis_{\aschema}$, such that:
\begin{enumerate}
\itemsep 0 cm
  \item No path schema in $\asetbis_{\aschema}$ contains guards with disjunctions in it. 
  \iflong
  \item For every path schema $\aschema' \in \asetbis_{\aschema}$,  its length 
        $\length{\aschema'}$ is polynomial in $\length{\aschema} + \card{T} + \card{B}$.
   \else 
    \item For every $\aschema' \in \asetbis_{\aschema}$, 
        $\length{\aschema'}$ is polynomial in $\length{\aschema} + \card{T} + \card{B}$.
   \fi 
  \item Checking whether a path schema $\aschema'$ belongs to $\asetbis_{\aschema}$ can be 
        done in polynomial time in $\size{\aschema} + \card{T} + \card{B}$.
  \item For every run $\arun$ respecting $\aschema$ and starting at $\pair{\astate_0}{\vec{v_0}}$, 
        we can find a run $\arun'$ respecting some $\aschema'\in\asetbis_{\aschema}$ such that 
         $\arun\models\aformula$ iff $\arun'\models\aformula$ for every $\aformula$ built
        over some resource $\triple{\aset}{T}{B}$.
  \item For every run $\arun'$ respecting some $\aschema'\in \asetbis_{\aschema}$ with initial 
        counter values
        $\vec{v_0}$, we can find a run $\arun$ respecting $\aschema$ such that 
        $\arun\models\aformula$ iff $\arun'\models\aformula$ for every $\aformula$ built
        over some resource $\triple{\aset}{T}{B}$.
  \item For every ultimately periodic word $\aword \cdot \awordbis^{\omega}\in \languageof{P'}$, 
      for every $\aformula$ built
        over $\aresource$ 
       checking whether $\proj(\aword \cdot \awordbis^{\omega}),0 \symbmodels \aformula$ can be done in polynomial time in the size of $\aword \cdot \awordbis$ and in the size of $\aformula$. 
\end{enumerate}
\end{theorem}

\iflong

\begin{proof} Let $\asetbis_{\aschema}$ be the set of path schemas defined from the minimal path schema
$\aschema$.
%
\begin{enumerate}
\itemsep 0 cm 
\item For every path schema in $\asetbis_{\aschema}$, the guards on transitions are of the form
      $\bigwedge_{\aterm \in T} \aformulabis(\aterm, \anupdate,\atermmap(\aterm))$
      and each guard $\aformulabis(\aterm, \anupdate,\atermmap(\aterm))$ is itself
      an atomic guard and a conjunction of two atomic guards. 
      Hence, no path schema in $\asetbis_{\aschema}$ contains guards with disjunctions in it. 
 \item By Lemma~\ref{lemma-bounded-skeleton}, every skeleton defining a path schema 
       in $\asetbis_{\aschema}$ has polynomial length in  $\length{\aschema} + \card{T} + \card{B}$.
       Each path schema in $\asetbis_{\aschema}$ has a linear length in the length of its corresponding skeleton.
       Consequently, for  $\aschema' \in \asetbis_{\aschema}$,  its length 
        $\length{\aschema'}$ is polynomial in $\length{\aschema} + \card{T} + \card{B}$.
  \item Given a path schema $\aschema'$ in $\asetbis_{\aschema}$, one can easily identify its underlying
        skeleton $\askeleton$ by removing iteration operators such as $^+$ and $^{\omega}$ (easy at the cost of keeping track
        of transitions from $\edges_{\aschema}$). 
        By Lemma~\ref{lemma-path-check},  checking whether $\askeleton$ is compatible
        with $\aschema$ and $\pair{\astate_0}{\vec{v_0}}$  can be 
        done in polynomial time in $\size{\aschema} + \card{T} + \card{B}$.
        In particular, if $\askeleton$ is too long, this can be checked in polynomial time too. 
  \item By Proposition~\ref{run-equivalence}(I), 
        for every run $\arun$ respecting $\aschema$ and starting at $\pair{\astate_0}{\vec{v_0}}$, there 
        are $\aschema' \in\asetbis_{\aschema}$  and a run $\arun'$ respecting $\aschema'$ 
        such that $\footprint{\arun} = \footprint{\arun'}$. By 
        \iflong
        Lemma~\ref{lemma-footprint-equivalence},
        \else 
        Lemma~\ref{lemma-footprint-run},
        \fi 
        $\arun\models\aformula$ iff $\arun'\models\aformula$.
  \item Similar to (4.) by using Proposition~\ref{run-equivalence}(II).
  \item We consider an ultimately periodic word $\aword \cdot \awordbis^{\omega}\in \languageof{P'}$.
        From it we can build in linear time the 
        ultimately periodic word $\aword' \cdot \awordbis'^\omega=\proj(\aword \cdot \awordbis^{\omega})$ over 
        the alphabet $\powerset{\aset} \times I^T$ and the size of the word $\aword'$ [resp. $\awordbis'$] is linear in the size of the word $\aword$ [resp. $\aword'$].  By~\cite{Laroussinie&Markey&Schnoebelen02}, we know that 
     $\aword' \cdot \awordbis'^\omega, 0 \symbmodels \aformula$ can be checked
      in time $\mathcal{O}(\sizeof{\aformula}^2 \times 
      \length{\aword' \cdot \awordbis'})$. Indeed, $\symbmodels$ is analogous to the satisfiability
       relation for plain Past LTL.
\end{enumerate}
\end{proof}

\fi 

\fi
\section{Model-checking $\PLTL[\counters]$ over Flat Counter Systems}
\label{section-np-no-counters}
In this section, we provide a nondeterministic polynomial-time algorithm to solve
$\mc{\PLTL[\counters]}{\flatcs}$ (see Algorithm~\ref{algorithm-main}). 
\iflong
To do so, we combine
the properties of the general stuttering theorem for LTL with past-time 
operators (see Theorem~\ref{theorem-stuttering}) 
with small solutions of constraint systems.
\fi 
\iflong
In Algorithm~\ref{algorithm-main} below, 
we have chosen to perform the nondeterministic steps (guesses) at the
beginning of the algorithm only. Note that a polynomial $p^{\star}(\cdot)$ is used. 
\fi 
\ifshort
\else
Let us explain below how 
\iflong
it 
\else
$p^{\star}(\cdot)$
\fi
is
defined. Let $\asys$ be a flat counter system, $\aconf_0 = \pair{\astate_0}{\vec{v_0}}$ be an
initial  configuration and $\aformula \in \PLTL[\counters]$. Let 
$N = \size{\asys} + \allowbreak \size{\pair{\astate_0}{\vec{v_0}}}\allowbreak + \size{\aformula}$.
Let $\aschema$ be a minimal path schema of $\asys$. We have:
\begin{itemize}
\itemsep 0 cm
\item $\length{\aschema}  \leq 2 \times \card{\edges} \leq 2N$,
\item $\nbloops{\aschema} \leq \card{\states} \leq N$.
\end{itemize}
Let $T$ be the set of terms $\aterm$ occurring in $\asys$ and $\aformula$ in guards of the form
$\aterm \sim b$. We have $\card{T} \leq  \size{\asys} + \size{\aformula} \leq N$.
Let $B$ be the set of constants $b$ occurring in $\asys$ and 
$\aformula$ in guards of the form $\aterm \sim b$. We have $\card{B} \leq  
\size{\asys} + \size{\aformula} \leq N$. Let $\aresource = \triple{\aset}{T}{B}$ be
the resource such that $\aset$ is the finite set of propositional variables occurring in $\aformula$.

Let $MAX$ be the maximal absolute value of a constant occurring in $\asys$, $\aformula$, $\vec{v_0}$
(either as an element of $B$ or as a coefficient in front of a counter as a value in $\vec{v_0}$). 
We have $MAX \leq 2^N$.

Now, let $\aschema'$ be a path schema in $\asetbis_{\aschema}$ with
$\aschema' = \aseg_1 (\aloop_1)^+ \aseg_2 (\aloop_2)^+ \cdots \aseg_{k} (\aloop_{k})^{\omega}$. 
Since $\length{\aschema'}\leq (\length{\aseg_1} + \cdots + \length{\aseg_k})
+ 2 \times ( 2 \times \card{T} \times \card{B} + \card{T}) 
\times  (\length{\aloop_1} + \cdots + \length{\aloop_k})$,
we have $\length{\aschema'} \leq 5 \times \card{T} \times \card{B} \times \length{\aschema} \leq 5 N^3$. 
Similarly, $\nbloops{\aschema'} \leq 5 N^3$. The number of guards occurring in $\aschema'$
is bounded by $\length{\aschema'} \times 2 \times \card{T} \leq 10 \times N^4$. 
The maximal constant $MAX'$ occurring in $\aschema'$ is bounded by $MAX + n \times MAX^2$ which
is bounded by $N \times 2^{2 \times N}$. 
Let $\aconstraintsystem$ be the constraint system defined from $\aschema'$.
\begin{itemize}
\itemsep 0 cm
\item The number of variables is equal to $\nbloops{\aschema'}$ which is bounded by $5 N^3$.
\item The number of conjuncts is bounded by 
      $ 2 \times \length{\aschema'} \times n \times (1 + N_1)$ where $N_1$ is the 
      number of atomic guards in $\aschema'$. 
      Hence, this number is bounded by $2 \times  5 N^3 \times N \times (1 + 10 \times N^4) 
      \leq 110 N^8$. 
\item The greatest absolute value from constants in  $\aconstraintsystem$
      is bounded by $n \times \nbloops{\aschema'} \times (MAX')^4 \times \length{\aschema'}^3$,
      which is bounded by $N (5 N^3) (N \times 2^{2 \times N})^4 \times 5^3 N^9 \leq
      625 \times N^{17} \times 2^{8 \times N}$.   
\end{itemize}

Let us show that $\aconstraintsystem \wedge \aformulabis_1 \wedge \cdots
       \wedge \aformulabis_{k-1}$ admits a small solution using the theorem below
for any  $\aformulabis_1 \wedge \cdots
       \wedge \aformulabis_{k-1}$ built from  Algorithm~\ref{algorithm-main}.

\begin{theorem}
\label{theorem:Borosh-Treybig}
\cite{Borosh&Treybig76}
Let $\amatrix \in \interval{-M}{M}^{U \times V}$
and $\vect{b} \in \interval{-M}{M}^U$,
where $U, V, M \in \Nat$.
If there is $\vect{x} \in \Nat^V$
such that $\amatrix \vect{x} \geq \vect{b}$,
then there is $\vect{y} \in [0, (\max \{V, M\})^{\mathsf{C} U}]^V$
such that $\amatrix \vect{y} \geq \vect{b}$,
where $\mathsf{C}$ is some constant.
\end{theorem}

By Theorem~\ref{theorem:Borosh-Treybig}, $\aconstraintsystem \wedge \aformulabis_1 \wedge \cdots
       \wedge \aformulabis_{k-1}$ has a solution iff $\aconstraintsystem \wedge \aformulabis_1 \wedge \cdots
       \wedge \aformulabis_{k-1}$
has a solution whose counter values are bounded by
$$
(625 \times N^{17} \times 2^{8 \times N})^{\mathsf{C} \times  2 \times (110 \times N^8 + 5 \times N^3)}
$$
which can be easily shown to be bounded by $2^{p^{\star}(N)}$ for some polynomial $p^{\star}(\cdot)$ (of degree 9). 
This is precisely, the polynomial  $p^{\star}(\cdot)$ that is used in Algorithm~\ref{algorithm-main}
(for obvious reasons). 
In order to justify the coefficient $2$ before $110$, note that any constraint of the 
form $\sum_{i} \afactor_i \avariablebis_i \sim \aconstant$ 
with $\sim \in \set{=,\leq,\geq,<,>}$ can be equivalently replaced by 1 or 2 atomic 
constraints of the form  $\sum_{i} \afactor_i \avariablebis_i \geq \aconstant$.
 
\fi 
\begin{algorithm}
{\footnotesize
\caption{The main algorithm in \np \ with inputs $\asys$, 
$\aconf_0 = \pair{\astate}{\vec{v_0}}$, $\aformula$}
\label{algorithm-main}
\begin{algorithmic}[1]
\STATE guess a minimal  path schema $\aschema$
       of $\asys$ 
\STATE build a resource $\aresource = \triple{\aset}{T}{B}$ coherent with $\aschema$ and $\aformula$
\STATE guess a valid \iflong path \fi  schema  $\aschema' = 
       \aseg_1 \aloop_1^+ \aseg_2 \aloop_2^+ \ldots \aseg_k \aloop_k^\omega$ 
       such that {\footnotesize $\length{\aschema'} \leq q^{\star}(\length{\aschema} + \card{T} + \card{B})$}
\iflong 
\STATE guess $\vec{y} \in \interval{1}{2 \td{\aformula} + 5}^{k-1}$
\STATE guess $\vec{y'} \in \interval{1}{2^{p^{\star}(\size{\asys} + \size{\aconf_0} +  \size{\aformula})}}^{k-1}$ 
\fi 
\STATE check that $\aschema'$ belongs to $\asetbis_\aschema$ 
\STATE check that $\proj(\aseg_1 \aloop_1^{\vec{y}[1]} \aseg_2 \aloop_2^{\vec{y}[2]} \ldots 
       \aloop_{k-1}^{\vec{y}[k-1]} 
       \aseg_k \aloop_k^\omega), 0 \symbmodels \aformula$ 
\STATE  build \iflong the constraint system \fi $\aconstraintsystem$ 
        over \iflong the variables \fi $\avariablebis_1$, \ldots, $\avariablebis_{k-1}$
        for $\aschema'$ with initial
        \iflong counter \fi  values $\vec{v_0}$ (obtained from Lemma~\ref{lemma-constraint-system})
\FOR{$i = 1 \to k-1$} 
\iflong
\IF{$\vec{y}[i] = 2 \td{\aformula} + 5$}
   \STATE $\aformulabis_i \gets ``\avariablebis_i \geq 2 \td{\aformula} + 5"$
   \ELSE
   \STATE  $\aformulabis_i \gets ``\avariablebis_i = \vec{y}[i]"$
\ENDIF
\fi 
\ENDFOR
\STATE check that $\vec{y'} \models \aconstraintsystem \wedge \aformulabis_1 \wedge \cdots
       \wedge \aformulabis_{k-1}$
\end{algorithmic}
}
\end{algorithm}

Algorithm~\ref{algorithm-main} starts by guessing a path schema  $\aschema$ (line 1)
and an unfolded path schema $\aschema'= \aseg_1 \aloop_1^+ \aseg_2 \aloop_2^+ \ldots \aseg_k \aloop_k^\omega$ (line 3)
and check whether $\aschema'$ belongs to  $\asetbis_\aschema$ (line 5). 
It remains to check whether there is a run $\arun$ respecting $\aschema'$ such that
$\arun \models \aformula$. Suppose there is such a run $\arun$;  
let $\vec{y}$ be the unique tuple in $\interval{1}{2 \td{\aformula}+5}^{k-1}$ such that
$\vec{y} \approx_{2 \td{\aformula} + 5} \loopsof{\aschema'}{\arun}$. 
By Proposition~\ref{proposition-iter-flatks}, we have
 $\proj(\aseg_1 \aloop_1^{\vec{y}[1]} \aseg_2 \aloop_2^{\vec{y}[2]} \ldots 
       \aloop_{k-1}^{\vec{y}[k-1]} 
       \aseg_k \aloop_k^\omega), 0 \symbmodels \aformula$.
Since the set of tuples of the form $\loopsof{\aschema'}{\arun}$ is characterized by a system of equations,
by the existence of  small solutions from~\cite{Borosh&Treybig76}, we can assume that $\loopsof{\aschema'}{\arun}$
contains only small values. 
Hence line 4 guesses  $\vec{y}$ and $\vec{y'}$ (corresponding to  $\loopsof{\aschema'}{\arun}$ 
with small values). Line 6 precisely checks $\proj(\aseg_1 \aloop_1^{\vec{y}[1]} \aseg_2 \aloop_2^{\vec{y}[2]} \ldots 
       \aloop_{k-1}^{\vec{y}[k-1]} 
       \aseg_k \aloop_k^\omega), 0 \symbmodels \aformula$
whereas line 11 checks whether $\vec{y'}$ encodes a run respecting $\aschema'$ 
with $\vec{y'} \approx_{2 \td{\aformula} + 5} \vec{y}$. 
\begin{lemma} 
\label{lemma-main-algo-in-np}
Algorithm~\ref{algorithm-main} runs in nondeterministic polynomial time. 
\end{lemma}
\ifshort
\else

\begin{proof} First, let us check that all the guesses can be done in 
polynomial time.
\begin{itemize}
\itemsep 0 cm 
\item A minimal path schema $\aschema$ of $\asys$ is of polynomial size
with respect to  the size of $\asys$.
\item The path schema $\aschema'$ is of polynomial size
with respect to the size of $\aschema$, $\aformula$ and $\aconf_0$
(Theorem~\ref{theorem-main-disjunction}(2)).
\item $\vec{y}$ and $\vec{y'}$ are obviously of polynomial size since
their components have values bounded by some exponential expression
(values in $\vec{y}$ can be much smaller than the values in $\vec{y'}$).
\end{itemize}
Now, let us verify that all the checks can be in done in polynomial
time too. 
\begin{itemize}
\itemsep 0 cm
\item Both $\aschema$ and $\aschema'$ are in polynomial size with respect
to the size of the inputs and checking compatibility amounts to verify
that  $\aschema'$ is an unfolding of $\aschema$, which can be done
in polynomial time (see Lemma~\ref{lemma-path-check}).
\item Checking whether  
      $\proj(\aseg_1 \aloop_1^{\vec{y}[1]} \aseg_2 \aloop_2^{\vec{y}[2]} \ldots 
       \aloop_{k-1}^{\vec{y}[k-1]} 
       \aseg_k \aloop_k^\omega), 0 \symbmodels \aformula$ can be done in
polynomial time using Theorem~\ref{theorem-main-disjunction}(6) since 
$\aseg_1 \aloop_1^{\vec{y}(1)} \aseg_2 \aloop_2^{\vec{y}(2)} \ldots \aloop_{k-1}^{\vec{y}(k-1)} \allowbreak \aseg_k \aloop_k$ is of polynomial size with respect to the size of $\aschema'$ and $\aformula$.
%
%
%
\item Building $\aconstraintsystem \wedge \aformulabis_1 \wedge \cdots
       \wedge \aformulabis_{k-1}$
      can be done in polynomial time since 
      $\aconstraintsystem$ can be built in polynomial time with respect
      to the size of $\aschema'$ \iflong (see Section~\ref{section-characterization}) \fi 
      and  $\aformulabis_1 \wedge \cdots
       \wedge \aformulabis_{k-1}$ can be built in 
      polynomial time with respect
      to the size of $\aformula$ ($\td{\aformula} \leq \sizeof{\aformula}$).
\item $\vec{y'} \models \aconstraintsystem \wedge \aformulabis_1 \wedge \cdots
       \wedge \aformulabis_{k-1}$ can be finally checked in polynomial time
     since the values in $\vec{y'}$ are of exponential magnitude
     and the combined constraint system is of polynomial size.
\end{itemize}
\end{proof}

\fi 
It remains to check that Algorithm~\ref{algorithm-main} is correct, which is stated below.
\begin{lemma} 
\label{lemma-correctness-main-algorithm}
$\asys, \aconf_0 \models \aformula$ iff 
Algorithm~\ref{algorithm-main} on inputs $\asys$, $\aconf_0$, $\aformula$
has an accepting run.
\end{lemma}
In the proof of Lemma~\ref{lemma-correctness-main-algorithm}, we take advantage of all our preliminary
results. 
\begin{proof}
\ifshort
\else
 First, let us show that if Algorithm~\ref{algorithm-main} on inputs 
$\asys$, $\aconf_0 = \pair{\astate_0}{\vec{v_0}}$, $\aformula$ has an accepting computation, then 
 $\asys, \aconf_0 \models \aformula$. 
\fi 
This means that there 
are $\aschema$, $\aschema'$, $\vec{y}$, $\vec{y'}$ that satisfy all the checks.
 Let $\aword = 
\aseg_1 \aloop_1^{\vec{y'}[1]} \cdots \aseg_{k-1} \aloop_{k-1}^{\vec{y'}[k-1]} 
 \aseg_{k} \aloop_k^{\omega}$ and
$\arun = \pair{\pair{\astate_0}{\atermmap_0}}{\vec{v_0}}
\pair{\pair{\astate_1}{\atermmap_1}}{\vec{x_1}} 
\pair{\pair{\astate_2}{\atermmap_2}}{\vec{x_2}} \cdots \in (\states' \times \Zed^n)^{\omega}$
be defined as follows: 
\iflong
\begin{itemize}
\itemsep 0 cm 
\item For every $i \geq 0$, $\astate_i \egdef \pi_1(\source{\aword(i)})$,
\item $\vec{x_0} \egdef \vec{v_0}$ and for every $i \geq 1$, we have
      $\vec{x_i} \egdef \vec{x_{i-1}} + \update{\aword(i)}$. 
\end{itemize}
\fi 
By Lemma~\ref{lemma-constraint-system}, since $\vec{y'} \models \aconstraintsystem  \wedge
\aformulabis_1 \wedge \cdots
       \wedge \aformulabis_{k-1}$, 
$\arun$ is a run respecting $\aschema'$
starting at the configuration $\pair{\pair{\astate_0}{\atermmap_0}}{\vec{v_0}}$. 
Since $\vec{y'} \models \aformulabis_1 \wedge \cdots
       \wedge \aformulabis_{k-1}$ and $\vec{y} \models \aformulabis_1 \wedge \cdots
       \wedge \aformulabis_{k-1}$, by Proposition~\ref{proposition-iter-flatks},
\iflong
the propositions below are equivalent:
\begin{itemize}
\itemsep 0 cm
\item[($\maltese$)]
$\proj(\aseg_1 \aloop_1^{\vec{y}[1]} \aseg_2 \aloop_2^{\vec{y}[2]} \ldots 
       \aloop_{k-1}^{\vec{y}[k-1]} 
       \aseg_k \aloop_k^\omega), 0 \symbmodels \aformula$,
\item[($\maltese \maltese$)] 
$\proj(\aseg_1 \aloop_1^{\vec{y'}[1]} \aseg_2 \aloop_2^{\vec{y'}[2]} \ldots 
       \aloop_{k-1}^{\vec{y'}[k-1]} 
       \aseg_k \aloop_k^\omega), 0 \symbmodels \aformula$. 
\end{itemize}
\fi 
\iflong Line 6 from \fi 
Algorithm~\ref{algorithm-main} guarantees that 
$\proj(\aseg_1 \aloop_1^{\vec{y}[1]} \aseg_2 \aloop_2^{\vec{y}[2]} \ldots 
       \aloop_{k-1}^{\vec{y}[k-1]} 
       \aseg_k \aloop_k^\omega), 0 \symbmodels \aformula$, whence we have
 ($\maltese \maltese$).
        Since $\proj(\aseg_1 \aloop_1^{\vec{y'}[1]} \aseg_2 \aloop_2^{\vec{y'}[2]} \ldots 
       \aloop_{k-1}^{\vec{y'}[k-1]} 
       \aseg_k \aloop_k^\omega)=\footprint{\arun}$, by Lemma~\ref{lemma-footprint-run}, we deduce that $\arun, 0 \models \aformula$.
%
%
%
%
By Theorem~\ref{theorem-main-disjunction}(5), 
there is an infinite 
run $\arun'$, starting at the configuration $\pair{\astate_0}{\vec{v_0}}$ and respecting $\aschema$, 
such that $\arun', 0 \models \aformula$. 

Now, suppose that  $\asys, \aconf_0 \models \aformula$. We shall show that there exist
$\aschema$, $\aschema'$, $\vec{y}$, $\vec{y'}$ that allow to build an accepting  computation of 
Algorithm~\ref{algorithm-main}. 
There is a run $\arun$ starting at $\aconf_0$ such that $\arun, 0 \models \aformula$.
By Corollary~\ref{corollary-ps-cover-runs}, $\arun$ respects some minimal
path schema of $\asys$, say $\aschema$. 
By Theorem~\ref{theorem-main-disjunction}(4), 
there is a path schema $\aschema' = \aseg_1 \aloop_1^+ \aseg_2 \aloop_2^+ \ldots
 \aseg_k \aloop_k^\omega$ in $\asetbis_{\aschema}$ for which there is a run 
$\arun'$ satisfying $\aformula$. Furthermore, since $\aschema' \in \asetbis_\aschema$, $\length{\aschema'}\leq q^{\star}(\length{\aschema} + \card{T} + \card{B})$ for some polynomial $q^{\star}(\cdot)$. 
From $\loopsof{\aschema'}{\arun'} \in (\Nat \setminus \set{0})^{k-1}$,
for every $i \in \interval{1}{k-1}$, we consider $\aformulabis_i$ such that
$\aformulabis_i$ is equal to $\avariablebis_i = \loopsof{\aschema'}{\arun'}[i]$ 
if $\loopsof{\aschema'}{\arun'}[i] \leq 2 \td{\aformula} + 5$, otherwise
$\aformulabis_i$ is equal to $\avariablebis_i \geq 2 \td{\aformula} + 5$. 
Since $\aschema'$ admits at least one infinite run $\arun'$
such that $\loopsof{\aschema'}{\arun'}$ satisfies  $\aformulabis_1 \wedge \cdots \wedge \aformulabis_{k-1}$, 
the constraint system $\aconstraintsystem$
obtained from $\aschema'$  (thanks to  Lemma~\ref{lemma-constraint-system}) but 
augmented with $\aformulabis_1 \wedge \cdots \wedge \aformulabis_{k-1}$   
admits at least one solution.
Let us define  $\vec{y'} \in \interval{1}{2^{p^{\star}(\size{\asys} + \size{\aconf_0} +  \size{\aformula})}}^{k-1}$ 
as a small solution of $\aconstraintsystem  \wedge \aformulabis_1 \wedge \cdots
       \wedge \aformulabis_{k-1}$ and $\vec{y} \in \interval{1}{2 \td{\aformula} +5}^{k-1}$
be defined such that for $i \in \interval{1}{k-1}$, $\vec{y}[i] = max(\vec{y'}[i], 
2 \td{\aformula} + 5)$.
As shown 
\iflong previously, 
\fi 
\iflong the bound \fi $2^{p^{\star}(\size{\asys} + \size{\aconf_0} +  \size{\aformula})}$
is sufficient if there is a solution. 
Clearly,  $\vec{y'} \models \aconstraintsystem \wedge \aformulabis_1 \wedge \cdots
       \wedge \aformulabis_{k-1}$.
So $\aseg_1 \aloop_1^{\vec{y'}[1]} \aseg_2 \aloop_2^{\vec{y'}[2]} \ldots \aloop_{k-1}^{\vec{y'}[k-1]} 
       \aseg_k \aloop_k^\omega$ generates a genuine run. Since $\footprint{\arun'}=\proj(\aseg_1 \aloop_1^{\vec{y'}[1]} \aseg_2 \aloop_2^{\vec{y'}[2]} \ldots \aloop_{k-1}^{\vec{y'}[k-1]} 
       \aseg_k \aloop_k^\omega)$ (see Lemma~\ref{guard-constraint}) and since by Lemma~\ref{lemma-footprint-run}, we have $\footprint{\arun'}\symbmodels \aformula$, we get that 
$\proj(\aseg_1 \aloop_1^{\vec{y'}[1]} \aseg_2 \aloop_2^{\vec{y'}[2]} \ldots \aloop_{k-1}^{\vec{y'}[k-1]} 
       \aseg_k \aloop_k^\omega), 0 \symbmodels \aformula$.
This also implies that $\aschema'$ is valid. Hence
\iflong 
$$\proj(\aseg_1 \aloop_1^{\vec{y}[1]} \aseg_2 \aloop_2^{\vec{y}[2]} \ldots \aloop_{k-1}^{\vec{y}[k-1]} 
       \aseg_k \aloop_k^\omega), 0 \symbmodels \aformula$$
thanks to Proposition~\ref{proposition-iter-flatks}. 
\fi 
\iflong
Consequently, we have all the ingredients to build safely an accepting run for 
Algorithm~\ref{algorithm-main} on inputs $\asys$, $\aconf_0$, $\aformula$.
\fi  
\end{proof}

As a corollary, we can state the main result of the paper.
\begin{theorem} 
\label{theorem-main}
$\mc{\PLTL[\counters]}{\flatcs}$ is \np-complete.
\end{theorem}
\iflong
As an additional corollary, we can solve the global model-checking problem with existential
Presburger formulae (we knew that Presburger formulae exist for global model-checking~\cite{demri-model-10} 
but we can conclude that they are structurally simple and we provide an alternative proof). 
\begin{corollary} Given a flat counter system $\asys$, a control state $\astate_0$ and a formula
$\aformula \in \PLTL[\counters]$ one can effectively build an existential Presburger formula
$\aformula$ that represents the initial counter values $\vec{v_0}$ such that there is
an infinite run $\arun$ starting at $\pair{\astate_0}{\vec{v_0}}$ such that
$\arun, 0 \models \aformula$.
\end{corollary}
It is sufficient to consider the formula below:
$$
\bigvee_{{\rm minimal \ path \ schema} \ \aschema} \ \ \
\bigvee_{\aschema' \in \asetbis_{\aschema}} \ \ \
$$
$$
\bigvee_{\vec{y} \ s.t. \ \footprint{\aseg_1 \aloop_1^{\vec{y}[1]} \aseg_2 \aloop_2^{\vec{y}[2]} \ldots \aloop_{k-1}^{\vec{y}[k-1]} \aseg_k \aloop_k^\omega}, 0 \ \symbmodels \ \aformula} \ \
\exists \ \avariablebis_1 \cdots \ \avariablebis_{k-1} \
\aconstraintsystem_{\aschema'}' \wedge \aformulabis_1 \wedge \cdots \wedge \aformulabis_{k-1}
$$
where the first generalized disjunction deals with minimal path schemas starting on $\astate_0$, 
the third generalized disjunction deals with $\vec{y} \in \interval{1}{2 \td{\aformula}+5}^{k-1}$.
Note that $\aconstraintsystem_{\aschema'}'$ is obtained from $\aconstraintsystem_{\aschema'}$
by replacing initial counter values by free variables. 

\subsection{The special case of path schemas with a single loop}

We have seen that $\mc{\PLTL[\counters]}{\cps(k)}$ is \np-hard as soon as $k \geq 2$.  By contrast, we prove that $\mc{\PLTL[\counters]}{\cps(1)}$ is in \ptime~ by using the previous proof techniques.


Consider a path schema $\aschema=\aseg.\aloop^{\omega}$ in a counter system with only one loop $\aloop$. Due to the structure of $\aschema$ there
exists at most one run $\arun$ respecting $\aschema$ and starting from a given initial configuration $\aconf_0$. $\footprint{\arun}$ (defined in
Section \ref{section-disjunction}) is of the form $u.v^{\omega}$, which is an ultimately periodic word. Since, the only loop $\aloop$ is to be taken an
infinite number of times, we have, $\length{v}=\length{\aloop}$ which is polynomial in size of the input, but $\length{u}$ can be exponential. But,
note that $\wordof{\arun(0)\arun(1)\cdots\arun(\length{u})}\in\aseg \cdot \aloop^{+}$
 where the number of 
repetition of $\aloop$ may be an
exponential number of times. The algorithm computes the number of different possible sets of term maps (defined in Section~\ref{section-disjunction}),
that the nodes of $\aloop$ can have. At most, this can be polynomially many times due to the monotonocity of guards and 
arithmetical constraints. Next, for
each such assignment $i$ of term maps to the nodes of $\aloop$, the algorithm calculates the number of iterations $nl_i$ of $\aloop$, for which the
terms remain in their respective term map. Note that each of these $nl_i$ can be exponentially large. Now, the formula is symbolically verified over
the ultimately periodic path where the nodes of the path schema are augmented with the term maps.

Before defining the algorithm formally, we need to define some notions to be used in the algorithm. For a path segment 
$\aseg= \anedge_1 \anedge_2\cdots
\anedge_{\length{\aseg}}$, we define $\aseg[i,j]=\anedge_i\anedge_{i+1}\cdots \anedge_j$ for $1\leq i\leq j\leq 
\length{\aseg}$. Also, for a loop segment
$\aloop$, we say a tuple of term maps 
$(\atermmap_1,\atermmap_2,\cdots,\atermmap_{\length{\aloop}})$ is \defstyle{final} iff for every term
$\aterm=\sum_j a_j \acounter_j\in T$ and for all $1\leq i\leq \length{\aloop}$,
\begin{itemize}
  \item  $\sum_j a_j \effect{\aloop}[j]>0$ implies $\atermmap_i(\aterm)$ is maximal in $I$.
  \item  $\sum_j a_j \effect{\aloop}[j]<0$ implies $\atermmap_i(\aterm)$ is minimal in $I$.
\end{itemize}
where $\effect{\aloop}$ is as defined in Section \ref{section-path-schemas}.

Since the unique run respecting $\aschema$ must contain $\aseg$ and copies of $\aloop$, we can specify the term maps for
$\aword=\aseg \cdot \aloop$. Consider the function 
$\amap_{{\rm init}} : \{0,1,2,\ldots,\length{\aword}\} \rightarrow I^{T}$ 
for a given configuration
$\aconf=\pair{\astate_0}{\avect_0}$, defined as:
\begin{itemize}
\itemsep 0 cm 
  \item $\amap_{{\rm init}}(0)=\atermmap_0$ iff for each term $\aterm=\sum_j a_j \acounter_j\in T$, we have that, $\sum_j
    a_j.\avect_0[j]\in\atermmap_0(\aterm)$ and $\atermmap_0 \vdash \guard{\aword(0)}$.
  \item for every $i \in \interval{1}{\length{\aword}}$, as $\amap_{{\rm init}}(i)=\atermmap_i$ iff, for each term 
   $\aterm=\sum_j a_j \acounter_j\in T$, we have that,
    $\sum_j a_j.(\effect{\aword[1,i]}[j]+\avect_0[j])\in\atermmap_i(\aterm)$ and $\atermmap_i\vdash \guard{\aword(i)}$.
  \item Otherwise, if the term maps do not satisfy the guards, then there does not exist any run and hence 
  $\amap_{{\rm init}}(i)$ is undefined. 
\end{itemize}
Also, we consider the function $curr: T \rightarrow \rel$ which, in the algorithm, gives the value of the terms at specific positions of the run. The
function $val_{curr}:\edges^{+}\rightarrow I^{T}$, is defined as $val_{curr}(w)=\atermmap$ where for all $\aterm=\sum_j a_j \acounter_j\in T,
curr(\aterm)+\sum_j a_j.(\effect{w}[j])\in \atermmap(\aterm)$. Given a path segment $\aseg=\anedge_1\anedge_2\cdots\anedge_{\length{\aseg}}$ with
$\anedge_i=(\astate_i,\aguard_i,\anupdate_i,\astate_{i+1})\in\edges$ for $i\in[1,\length{\aseg}]$ and a given 
tuple of term maps
$a=(\atermmap_1,\atermmap_2,\cdots,\allowbreak\atermmap_{\length{\aseg}})$, we define $\aseg\times a=\anedge'_1\anedge'_2\cdots\anedge'_{\length{\aseg}}$ where
$\anedge'_i=(\pair{\astate_i}{\atermmap_i},\aguard_i,\anupdate_i,\pair{\astate_{i+1}}{\atermmap_{i+1}})$. 

Given an initial configuration $\aconf$, we calculate the term maps for each position of $\aseg$ and the first iteration of $\aloop$, using
$\amap_{{\rm init}}$. Subsequently, we calculate new tuples of term maps 
$(\atermmap_1,\atermmap_2\cdots\atermmap_{\length{\aloop}})$ for $\aloop$ and the
number of iterations $nl$ of $\aloop$ for which the terms remain in their respective term map from the tuple. We store the tuple of term maps in an
array $A$ and the number of iterations corresponding to tuple $i$ in $nl_i$. In case, at any position, we reach some term maps that does not satisfy
some guard, the procedure is aborted as it means that there does not exist any run. 
Note that there are polynomially many entries in $A$ but each
of the $nl_i$ can be exponential. We perform symbolic model checking over a path schema augmented with the calculated term maps. The augmented
path schema is obtained by performing $\aloop\times A[i]$ for each $i$. But the number of times $\aloop\times A[i]$ is repeated, $nl_i$ can be
exponential. Thus, instead of taking $\aloop\times A[i]$, $nl_i$ times, we take it 
$\mathtt{Min}(nl_i,2\td{\aformula} + 5)$ times. 
By Theorem~\ref{theorem-stuttering}, 
we have that the two path schemas are equivalent in terms of satisfiability of $\aformula$. 
The polynomial-time algorithm is described in Algorithm~\ref{algorithm-cps}.

\begin{algorithm}
\caption{The \ptime \  algorithm with inputs $\aschema=\aseg \cdot \aloop^\omega$, $\aconf=\pair{\astate_0}{\avect_0}$,
$\aformula$}
\label{algorithm-cps}
\begin{algorithmic}[1]
\STATE Build a resource $\aresource = \triple{\aset}{T}{B}$ and a set of intervals $I$ coherent with $\aschema$ and $\aformula$.
\STATE Compute $\amap_{{\rm init}}(i)$ for all $i\in[0,\length{\aseg.\aloop}-1]$.
\STATE {\bf if} for some $i\in \interval{0}{\length{\aseg \cdot \aloop}-1}$, $\amap_{{\rm init}}(i)$ is undefined
       {\bf then} {\bf abort} 
\STATE For each term $\aterm=\sum_j a_j \acounter_j\in T$, $curr(\aterm) :=\sum_j a_j.(\effect{\aseg \cdot \aloop}[j]+\avect_0[j])$.
\STATE $h:=1$; $A[h]:=(\amap_{{\rm init}}(\length{\aseg}),\amap_{{\rm init}}(\length{\aseg}+1)\cdots \amap_{{\rm init}}(\length{\aseg.\aloop}-1))$
\WHILE{$A[h]$ is not final}
\STATE Compute, $nl_h=min\{nl|i\in [1,\length{\aloop}], \aterm\in T, val_{curr}(\aloop^{nl} \cdot \aloop[1,i])(\aterm)\neq A[h](i)(\aterm)\}$.
\STATE $h:=h+1$
\STATE $A[h] :=(\atermmap_1,\atermmap_2\cdots\atermmap_{\length{\aloop}})$, 
       such that at all positions $i$ in $\aloop$ we have that $val_{curr}(\aloop^{nl_{h}} \cdot \aloop[1,i])= \atermmap_i$.
\STATE For every term $\aterm=\sum_j a_j \cdot \acounter_j\in T$, set 
       $curr(\aterm)=curr(\aterm)+\sum_j a_j.(nl_h.\effect{\aloop}[j])$.
\STATE {\bf if} there is $i\in \interval{1}{\length{\aloop}}$ such that  $A[h](i) \nvdash\guard{\aloop(i)}$
       {\bf then} {\bf abort}
\ENDWHILE
\STATE For $j \in \interval{1}{h-1}$, $T[j] := \mathtt{Min}(nl_j,2\td{\aformula} + 5)$ 
\STATE Check that
$\proj((\aseg\times(\amap_{{\rm init}}(0),\ldots,\amap_{{\rm init}}(\length{\aseg}-1)).(\aloop\times A[1])^{T[1]}.(\aloop \times
A[2])^{T[2]}\ldots(\aloop\times A[h-1])^{T[h-1]}(\aloop\times A[h])^{\omega}),0 \symbmodels \aformula$
\end{algorithmic}
\end{algorithm}
It now remains to prove that the algorithm completes in \ptime \ and is correct.
\begin{lemma}
Algorithm \ref{algorithm-cps} terminates in time which is at most a polynomial in the size of the input.
\end{lemma}
\begin{proof}
We will verify that each step of the algorithm can be performed in polynomial time.
\begin{itemize}
\itemsep 0  cm 
  \item Building a resource and a set of intervals can be done by scanning the input once.
  \item Since the updates of $\aschema$ is part of the input, we can compute 
        $\amap_{{\rm init}}$
        for all positions in $\aseg \cdot \aloop$ in polynomial time.
  \item Computation of $curr$ depends on the previous value of $curr$ and the coefficients appearing in the guards of $\aschema$. Hence, it involves 
        addition and multiplication of at most polynomial number of bits. Thus, this can be performed in polynomial time.
  \item The maximum possible value for $h$ is bounded by a polynomial given by Lemma~\ref{lemma-bounded-skeleton}. 
        Indeed,  the process described in the while loop is the same as the creation of unfolded path schema set 
        $Y_{\aschema}$. The only difference being that 
        there exists only one possible run, if any and hence $Y_{\aschema}$ is a singleton set.
  \item Calculation of each $nl_h$ requires computing $val_{curr}$ which again involves arithmetical operations on 
        polynomially many bits. Thus, this requires polynomial time only.
  \item Checking $(\aseg\times(\amap_{{\rm init}}(0),\ldots,\amap_{{\rm init}}(\length{\aseg}-1)).(\aloop\times A[1])^{T[1]} 
       (\aloop\times A[2])^{T[2]} \ldots\allowbreak 
       (\aloop\times A[h-1])^{T[h-1]}\allowbreak(\aloop\times A[h])^\omega, 0 \symbmodels \aformula$ can be 
      done in polynomial time for the following reasons. 
    \begin{itemize}
      \item By definition of $T[h]$, size of $(\aseg\times(\amap_{{\rm init}}(0),\ldots,\amap_{{\rm init}}(\length{\aseg}-1)).
      (\aloop\times A[1])^{T[1]}
        (\aloop\times A[2])^{T[2]} \ldots\allowbreak (\aloop\times A[h-1])^{T[h-1]}(\aloop\times A[h])^\omega$ is
        polynomial in the size of the input.
      \item By~\cite{Laroussinie&Markey&Schnoebelen02}, 
        $(\aseg\times(\amap_{{\rm init}}(0),\ldots,\amap_{{\rm init}}(\length{\aseg}-1)).(\aloop\times
        A[1])^{T[1]} (\aloop\times A[2])^{T[2]} \allowbreak \ldots (\aloop\times
        A[h-1])^{T[h-1]}\allowbreak(\aloop\times A[h])^\omega, 0 \symbmodels \aformula$ can be checked in time
        $\mathcal{O}(\sizeof{\aformula}^2 \times\allowbreak\length{\aseg \cdot \aloop^{T[1]}\allowbreak \aloop^{T[2]} \cdots
          \aloop^{T[h-1]} \aloop})$.  Indeed, $\symbmodels$ is analogous to the satisfaction relation for plain Past LTL.
     \end{itemize}
\end{itemize}
\end{proof}

\begin{lemma}
$\aschema,\aconf\models \aformula$ iff Algorithm \ref{algorithm-cps} on inputs $\aschema,\aconf,\aformula$ has an accepting run.
\end{lemma}
\begin{proof}
Let us first assume that $\aschema,\aconf\models \aformula$. We will show that there exists a vector of positive integers 
$\vec{nL}=(nl_1,nl_2\ldots nl_h)$ for some $h\in\Nat$ such that Algorithm \ref{algorithm-cps} has an accepting run. Clearly, the 
transitions taken by a run $\arun$ respecting $\aschema$ and satisfying $\aformula$ is of the form, $\aseg\aloop^{\omega}$. This 
can be decomposed in the form $\aseg \aloop^{nl_1} \aloop^{nl_2} \ldots \aloop^{nl_h} \aloop^\omega$, depending on the portion of $\aschema$ traversed, 
such that for each consecutive 
copy of $\aloop$, the term maps associated with the nodes change. It is easy to see that this decomposition is same as the one calculated by the algorithm.
 Now, the elements of $\vec{nL}$ can be exponential. 
But due to Lemma \ref{lemma-footprint-run} and Stuttering theorem (Theorem \ref{theorem-stuttering}), we know that,
$(\aseg\times(\amap_{{\rm init}}(0),\ldots,\amap_{{\rm init}}(\length{\aseg}-1)).(\aloop\times A[1])^{nl_1} 
(\aloop\times A[2])^{nl_2} \ldots 
(\aloop\times A[h-1])^{nl_{h-1}}(\aloop\times A[h])^\omega, 0 \symbmodels \aformula$ iff
$(\aseg\times(\amap_{{\rm init}}(0),\ldots,\amap_{{\rm init}}(\length{\aseg}-1)).(\aloop\times A[1])^{T[1]} (\aloop\times A[2])^{T[2]} \allowbreak\ldots\allowbreak 
(\aloop\times A[h-1])^{T[h-1]}\allowbreak(\aloop\times A[h])^\omega, 0 \allowbreak\symbmodels \aformula$. Hence, 
the algorithm has an accepting run.

Now, we suppose that the algorithm has an accepting run on inputs $\aschema,\aconf$ and $\aformula$. We will prove that $\aschema,\aconf\models \aformula$.
Since the algorithm has an accepting run, we assume the integers calculated by it are $nl_1,nl_2,\cdots,nl_h$.
Let $\aword = \aseg \aloop^{nl_1} \aloop^{nl_2} \ldots \aloop^{nl_h} \aloop^\omega$ and
$\arun = \pair{\pair{\astate_0}{\atermmap_0}}{\vec{x_0}}
\pair{\pair{\astate_1}{\atermmap_1}}{\vec{x_1}} 
\pair{\pair{\astate_2}{\atermmap_2}}{\vec{x_2}} \cdots \in (\states' \times \Zed^n)^{\omega}$
be defined as follows: 
for every $i \geq 0$, $\astate_i \egdef \pi_1(\source{\aword(i)})$,
$\vec{x_0} \egdef \vec{v_0}$ and for every $i \geq 1$, we have
      $\vec{x_i} \egdef \vec{x_{i-1}} + \update{\aword(i)}$.  
By the calculation of $l_j$, $1\leq j\leq n$, in the algorithm, it is easy to check that $\pair{\astate_0}{\vec{x_0}}\pair{\astate_1}{\vec{x_1}} 
\pair{\astate_2}{\vec{x_2}} \cdots \in (\states \times \Zed^n)^{\omega}$ is a run respecting $\aschema$. Algorithm \ref{algorithm-cps} guarantees that
$(\aseg\times(\amap_{{\rm init}}(0),\ldots,\amap_{{\rm init}}(\length{\aseg}-1)).(\aloop\times A[1])^{T[1]} (\aloop\times A[2])^{T[2]} \ldots 
(\aloop\times A[h-1])^{T[h-1]}(\aloop\times A[h])^\omega, 0 \allowbreak\symbmodels \aformula$.
And thus, by Lemma~\ref{lemma-footprint-run}  and Theorem~\ref{theorem-stuttering}, we have, 
$\pair{\astate_0}{\vec{x_0}}\pair{\astate_1}{\vec{x_1}}\allowbreak\pair{\astate_2}{\vec{x_2}} \cdots ,0\models \aformula$.
\end{proof}

From the two last lemmas, we deduce the result concerning path schemas of counter systems with a single loop.

\begin{proposition}
\label{lemma-constant-loops1} 
$\mc{\PLTL[\counters]}{\cps(1)}$ is in \ptime.
\end{proposition}

\fi
%

\section{Conclusion}
\label{section-conclusion}

In this paper, we have investigated the computational complexity of the model-checking
problem for flat counter systems with formulae from an enriched version of 
\iflong
LTL 
(with past-time operators and arithmetical
constraints on the counters). 
\fi 
Our main result is the \np-completeness of \iflong the problem \fi 
$\mc{\PLTL[\counters]}{\flatcs}$, significantly improving the complexity upper bound 
from~\cite{demri-model-10}.
This also improves the results about the effective semilinearity of the reachability relations
for such flat counter systems 
\iflong 
from~\cite{Comon&Jurski98,Finkel&Leroux02b}; indeed, our logical dialects
allow to specify whether a configuration is reachable.  
\fi 
\iflong Figure~\ref{figure-summary} presents our main
results and compare them with the complexity of the reachability 
problem.
\fi  
\iflong
Furthermore, our results extend the recent result on the \np-completeness of model-checking 
flat Kripke structures with LTL from~\cite{Kuhtz&Finkbeiner11} (see also~\cite{Kuhtz10}) 
by adding counters and past-time operators.
\fi 
As far as the proof technique is concerned, the \np \ upper bound is obtained as a combination
of a general stuttering  property for LTL with past-time operators (a result extending
what is done in~\cite{Kucera&Strejcek05} with past-time operators) and the use
of small integer solutions for quantifier-free Presburger formulae~\cite{Borosh&Treybig76}.
\iflong
This latter technique is nowadays widely used to obtain optimal complexity upper bounds
for verification problems, see e.g.~\cite{Haaseetal09}.
Herein, our main originality rests on its intricate combination with a very general stuttering principle.
\fi 
There are several related problems which are not addressed in the paper. For instance, 
the extension of the model-checking problem to full CTL$^{\star}$ is known to be decidable~\cite{demri-model-10}
but the characterization of its exact complexity is 
\iflong
open
(note that we can also get decidability by taking advantage of our resolution
of global model-checking by replacing successively innnermost linear-time formulae by QFP formulae). 
\fi 
\iflong
Similarly, the extension
of the model-checking problem
with affine counter systems having the finite monoid property in the sense of~\cite{Finkel&Leroux02b},
is also known to be decidable~\cite{demri-model-10} but not its exact complexity. 
\fi 
\iflong
Another direction for extensions would be to consider richer update functions or guards
and to analyze how much our combined proof technique would be robust in those cases, for instance
by allowing transfer updates. 
\fi 

\begin{figure}
{\footnotesize
  \begin{center}
  \begin{tabular}{|c||c|c|c|}
    \hline
    Classes of  Systems & $\PLTL[\emptyset]$ & $\PLTL[\counters]$ & Reachability\\
    \hline
    \hline
    $\kps$ & \np-complete & ----- & \ptime \\ \hline
    $\cps$ & \np-complete & \np-complete (Theo.~\ref{theorem-main}) & \np-complete \\ \hline
    $\kps(n)$ & \ptime \ (Theo.~\ref{theorem-kps-fixed})  &  ----- & \ptime \\ \hline
    $\cps(n)$, $n > 1$ & ?? &  \np-complete (Lem.~\ref{lemma-constant-loops2}) & ??  \\ \hline
    $\cps(1)$  & \ptime &  \ptime & \ptime  \\ \hline
     $\flatks$ & \np-complete & ----- & \ptime \\ \hline
     $\flatcs$  & \np-complete & {\bf \np-complete} (Theo.~\ref{theorem-main}) & \np-complete\\
    \hline
  \end{tabular}
\end{center}
}
\caption{Summary: computational complexity of the problems  $\mc{\logicfrag}{\csfrag}$}
\label{figure-summary}
\end{figure} 

\bibliographystyle{elsarticle-harv}
\bibliography{biblio-fossacs12}
\newpage
\appendix
\section{Proof of (Claim 3)} 
Before the proof, let us recall what is (Claim 3).
Let 
$\aword= \aword_1 \awordbis^M \aword_2, \aword'=\aword_1 \awordbis^{M'} \aword_2 \in \aalphabet^{\omega}$, 
$i,i' \in \Nat$ and $N \geq 2$ such that  $M,M' \geq 2N+1$ and $\pair{\aword}{i} \equivrel{N} \pair{\aword'}{i'}$.
\begin{description}
\itemsep 0 cm
\item[(Claim 3)]  $\pair{\aword}{i+1} \equivrel{N-1} \pair{\aword'}{i'+1}$.
\end{description}

\section{Proof of (Claim 5)} 
Before the proof, let us recall what is (Claim 5). 
Let 
$\aword= \aword_1 \awordbis^M \aword_2, \aword'=\aword_1 \awordbis^{M'} \aword_2 \in \aalphabet^{\omega}$, 
$i,i' \in \Nat$ and $N \geq 2$ such that  $M,M' \geq 2N+1$ and $\pair{\aword}{i} \equivrel{N} \pair{\aword'}{i'}$.
\begin{description}
\itemsep 0 cm
\item[(Claim 5)] for all $j \leq i$, there is $j' \leq i'$ 
      such that  $\pair{\aword}{j} \equivrel{N-1} \pair{\aword'}{j'}$ and 
       for all $k' \in \interval{j'-1}{i'}$, there is
            $k \in  \interval{j-1}{i}$ such that $\pair{\aword}{k} \equivrel{N-1} \pair{\aword'}{k'}$.
\end{description}

\end{document}